\documentclass[preprint,authoryear]{elsarticle}
\usepackage[T1]{fontenc}
\usepackage[utf8]{inputenc}
\usepackage{tabularx}
\usepackage{supertabular}
\usepackage{booktabs}
\RequirePackage[authoryear]{natbib}
\usepackage{lmodern}
\usepackage{setspace}
\usepackage{nicefrac}
\usepackage{amsmath,dsfont}
\usepackage{framed}
\usepackage{a4wide}
\usepackage{amssymb,amsmath,amsthm}
\usepackage{bbm}
\usepackage{changes}
\usepackage[flushleft]{threeparttable}
\definecolor{dblue}{rgb}{0.21,0.21,0.55}

\newcommand{\hmn}{H\negthinspace M \negthinspace N}
\newcommand{\hn}{H\negthinspace N}
\usepackage{multirow}
\usepackage{lscape}
\usepackage[english]{babel}
\usepackage{graphicx}
\usepackage{arydshln}
\usepackage{rotating}
\usepackage[flushmargin,hang]{footmisc}
\usepackage{amsmath}
\usepackage{lscape}
\usepackage[pdflinkmargin=5pt,pdfstartview={FitBH -32768},plainpages=false]{hyperref}

\renewcommand{\P}{\mathbb{P}}

\newcommand{\E}{\mathbb{E}}
\newcommand{\N}{\mathbb{N}}
\newcommand{\R}{\mathbb{R}}
\newcommand{\1}{\mathbbm{1}}
\newcommand{\KLEINO}{{\scriptstyle{\mathcal{O}}}}
\DeclareMathAccent{\verywidehat}{\mathord}{largesymbols}{'144}
\newcommand{\var}{\mathbb{V}\hspace*{-0.05cm}\textnormal{a\hspace*{0.02cm}r}}

\newcommand{\cov}{\mathbb{C}\textnormal{o\hspace*{0.02cm}v}}

\allowdisplaybreaks[3]
\newtheorem{remark}{Remark}
\newtheorem{theo}{Theorem}
\newtheorem{assump}{Assumption}
\newtheorem{prop}{Proposition}[section]
\newtheorem{lem}{Lemma}

\newtheorem{cor}[prop]{Corollary}

\begin{document}
\renewcommand*{\thefootnote}{\fnsymbol{footnote}}

\title{Jump detection in high-frequency order prices}
\author[1]{Markus Bibinger%\footnote{Financial support from the Deutsche Forschungsgemeinschaft (DFG) under grant 403176476 is gratefully acknowledged.}
\footnote{Financial support from the Deutsche Forschungsgemeinschaft (DFG) under grant 403176476 is gratefully acknowledged.}}
\author[2]{Nikolaus Hautsch}
\author[2]{Alexander Ristig}
\address[1]{Faculty of Mathematics and Computer Science, Institute of Mathematics, University of W\"urzburg} 
\address[2]{Department of Statistics and Operations Research, University of Vienna}
%\normalsize
\begin{frontmatter}
%\maketitle\thispagestyle{empty}
%%\onehalfspacing
%\doublespacing
%
%\vspace{-.7cm} 
%
\begin{abstract}
{{\normalsize \noindent 
We propose methods to infer jumps of a semi-martingale, which describes long-term price dynamics, based on discrete, noisy, high-frequency observations. Different to the classical model of additive, centered market microstructure noise, we consider one-sided microstructure noise for order prices in a limit order book.

We develop methods to estimate, locate and test for jumps using local minima of best ask quotes. We provide a local jump test and show that we can consistently estimate jump sizes and jump times. One main contribution is a global test for jumps. We establish the asymptotic properties and optimality of this test. We derive the asymptotic distribution of a maximum statistic under the null hypothesis of no jumps based on extreme value theory. We prove consistency under the alternative hypothesis. The rate of convergence for local alternatives is determined and shown to be much faster than optimal rates for the standard market microstructure noise model. This allows the identification of smaller jumps. In the process, we establish uniform consistency for spot volatility estimation under one-sided noise. Online jump detection based on the new approach is shown to achieve a speed advantage compared to standard methods applied to mid quotes.

A simulation study sheds light on the finite-sample implementation and properties of the new approach and draws a comparison to a popular method for market microstructure noise. We showcase how our new approach helps to improve jump detection in an empirical analysis of intra-daily limit order book data.}}

\begin{keyword}
Boundary model\sep high-frequency data \sep limit order book \sep market microstructure\sep price jumps\\[.25cm]
\end{keyword}
\end{abstract}
%% \singlespacing 
%
\end{frontmatter}
\thispagestyle{plain}
%\renewcommand{\thefootnote}{}
%\newpage
%\setcounter{page}{2}
%\doublespacing
\section{Introduction}
For price data recorded at high frequencies it is well-known that market microstructure noise dilutes the underlying semi-martingale dynamics due to structural market effects, such as the bid-ask bounce and transaction costs. The standard observation model to account for these effects is
\[Y_i=X_{t_i^n}+\epsilon_i,~0\le i\le n,\]
where $(X_t)$ is a latent, continuous-time semi-martingale, which models the efficient log-price process, and $(\epsilon_i)$ are observation errors due to market microstructure. The classical model of market microstructure noise (MMN) introduces an additive noise process $(\epsilon_i)$  with expectation zero, see, for instance, \cite{howoften}, \cite{hansenlunde}, \cite{bn2}, \cite{sahaliajacod} and \cite{li2022remedi}. It is motivated and typically used for trade prices. Besides volatility estimation, testing for jumps of $(X_t)$, based on discrete observations with or without MMN, is one of the most important topics in statistics for high-frequency data, see \cite{barndorff2006econometrics}, \cite{ait2009testing}, \cite{lietal}, \cite{leemykland} and \cite{lm12}, among others.

\begin{figure}[t]
\begin{framed}
\centering
\includegraphics[width=11cm]{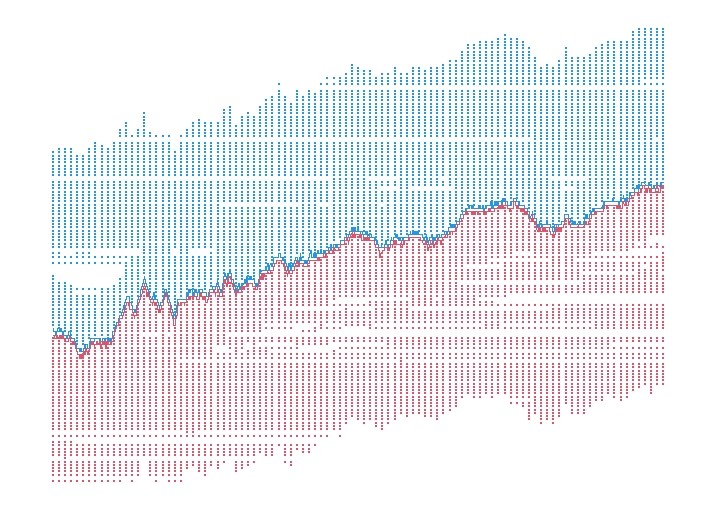}~~~\end{framed}
\caption{\label{Fig:lob} Lines interpolate best ask (blue) and best bid (red) prices of AAPL over 10 minutes. Above (below) prices of other active limit ask (bid) orders are plotted.}
\end{figure}
  
Limit order book data comprises more information than only trade prices. At each time point, it provides different quote levels including the best bid and best ask. Figure \ref{Fig:lob} shows a snapshot of quote dynamics for the AAPL asset traded at Nasdaq\footnote{We use limit order book data provided by LOBSTER, \url{lobsterdata.com}.} within a short time interval of 10 minutes. Only a small subset of all available quotes is displayed by the points. In the upper part, prices of ask limit orders are plotted with the line of best ask quotes as a lower hull. In the lower part, we show bid limit order prices, where the upper hull depicts the path of best bid quotes. At each time, all prices of ask limit orders are above all prices of bid limit orders with a bid-ask spread between the two lines. Since trades occur mainly when market orders are executed against existing limit orders, trade prices bounce between the best bid and the best ask. A natural question is how to efficiently use information from this data for inference on the efficient log-price and how to model and smooth out market microstructure for such data. One approach in recent works is to design structural noise models incorporating observed order book quantities, see \cite{li2016efficient}, \cite{chaker2017high}, \cite{clinet2019testing} and \cite{andersen2022local}. Assuming more structure than fully uninformative noise, these works establish improved volatility estimation and inference on the efficient log-price compared to MMN. We take a different point of view on the market microstructure and do not impose a functional form of the noise. Instead, we argue that the noise of limit order prices should not have unbounded support $(-\infty,\infty)$. Noise distributions with restricted support facilitate improved inference on the efficient log-price without imposing further structural conditions. This intriguing effect was used by \cite{BJR} to estimate the integrated volatility of a continuous semi-martingale with improved optimal rate $n^{-1/3}$, while the slower rate $n^{-1/4}$ is optimal for MMN. A generalization to spot volatility estimation and progress in the theory to establish asymptotic confidence intervals for volatility were recently contributed in \cite{bibinger2022}. Whenever data from a limit order book is available, we suggest to consider the time series of best ask quotes. Best bid quotes can be used equivalently and both be combined in practice. Based on a limit order book, it is so far standard to extract time series referring to both sides of the order book, such as trade prices or mid quotes, and to apply the MMN-model and corresponding methods to it. While the calibration of the MMN-model to mid quotes is reasonable to exploit larger samples than for trade prices, it is not clear why it should be suitable for best ask quotes or best bid quotes. This work proposes to consider local minima computed from samples referring to the upper side of the order book, such as best ask quotes or seller-initiated trades, and to model this data with lower-bounded noise. In an analogous way, we use for applications also local maxima computed from samples referring to the lower side of the order book, such as best bid quotes or buyer-initiated trades, which are modeled with one-sided, upper-bounded noise. Since boundary models facilitate super-efficient statistical inference, this is beneficial compared to combining lower and upper side in one model. In fact, considering best ask price dynamics, we do not have any bid-ask bounce effect. Instead, efficient prices should (usually) lie below the ask quotes and above bid quotes within the spread. Therefore, different to the MMN-model, we consider an additive noise model with \textbf{l}imit \textbf{o}rder \textbf{m}icrostructure \textbf{n}oise (LOMN). For ask quotes this is modeled by \textbf{l}ower-bounded, \textbf{o}ne-sided \textbf{m}icrostructure \textbf{n}oise, i.e., $\epsilon_i\ge 0$. Local minimum and maximum prices are provided also by the popular high-frequency candlestick data. Its use for high-frequency econometrics is recently discussed, see, e.g., \cite{candle2} and \cite{candle1}, so far based on models without microstructure noise. For many stocks, such data is publicly available, typically the highest frequency is so far restricted to one per minute.\,\footnote{e.g.\ on investing.com for, e.g., spx-500, several sources are mentioned in footnote 3 of \cite{candle1}} Extracting data from a limit order book, applying both approaches with the different models MMN and LOMN is possible in practice. For this reason, we compare both in the empirical analysis and in appropriate simulations. When we use adequate methods for the two different models in simulations, let us emphasize that this does not yield a comparison of the performances of different test methods, but rather of the different underlying noise models.

The jump inference based on local minima of ask quotes, and local maxima of bid quotes, has three striking advantages. First, the \emph{convergence rate is faster} than under MMN, such that \emph{smaller jumps can be detected} in theory and practice. Second, it is \emph{not manipulated by a pulverization effect} as local averages under MMN, which is a main problem of jump detection under noise (\cite{zhangmykland3}). Third, we establish in theory and an empirical study that jump detection based on local minima of ask and maxima of bid quotes has an advantage of speed. This means that jumps are detected online faster than based on mid quotes and MMN. In fast-moving electronic financial markets, speed is a main aspect for high-frequency trading algorithms and strategies, see, e.g., \cite{adv3} and \cite{adv4}. Jump events often provide relevant news, see, e.g., \cite{adv1}. The mean speed advantage in our data example is that jumps are detected approx.\ 0{.}32 seconds earlier. This is enough to seize opportunities for high-frequency, computer-based trading by taking advantage of the earlier news arrival, for instance, to rebalance a portfolio.

Since our statistics use only differences between adjacent block-wise local minima, our methods are robust to model generalizations where $\epsilon_i\ge b(i/n)$, for any continuous boundary function $b(t)$ of finite variation. For instance, a positive gap between boundary and the efficient log-price could exist to compensate market processing costs. This gap may be time-varying, but it seems plausible that it should be rather persistent. 

This work establishes inference on jumps of the efficient log-price in the LOMN-model and points out improvements compared to MMN. Methods to estimate, locate and test for jumps are constructed. Our test for jumps based on extreme value theory is in the spirit of the jump tests by \cite{leemykland}, for non-noisy high-frequency prices, and \cite{lm12} and \cite{bibwinkellev} for the MMN-model. The distinct structure of the noise in the LOMN-model, however, leads us to consider statistics based on local minima instead of local averages. The non-linearity of these statistics requires fundamentally different proofs in the asymptotic theory as $n\to\infty$. We prove a Gumbel convergence of a maximum statistic on which our test is based on and consistency under local alternatives with an optimal rate of convergence. While the Gumbel convergences in \cite{leemykland} and \cite{lm12} can be traced back to the weak convergence of the maximum of i.i.d.\ standard normally distributed random variables, this is not the case for our maximum statistic. Based on extreme value theory and bounds for the tails of convolutions, we explicitly derive this convergence in distribution. An important advantage compared to related procedures for MMN is that we establish uniformly consistent spot volatility estimation and the asymptotic theory for jump tests without any assumptions on the moments of the noise. The rate of convergence for local alternatives improves from $n^{-1/4}$ under MMN to the faster rate $n^{-1/3}$. This means that under LOMN we are able to detect smaller jumps than under MMN. For fixed jump size and number of observations, the power of our test outperforms the power of tests under MMN. Beyond improved asymptotic properties, our methods do not cause some intricate finite-sample problems known for MMN. In particular, we show that the effect of pulverization of jumps by pre-averages reported in  \cite{zhangmykland3} under MMN, which can result in spurious jump detection or gradual jumps, is avoided when using local order statistics. 

We develop consistent estimators for jump sizes and jump times as well as a local and a global test for jumps. The global test allows to test for jumps over some time period, usually one trading day. This is the standard problem of testing for jumps and can be used in practice to analyze whether or not jumps have to be taken into account. Detecting specific jumps and estimating their times and sizes is important to separate jumps and continuous price adjustments. This can be used, for instance, to perform high-frequency regression or factor analysis, separately for jumps and continuous components. Due to different mechanisms behind co-jumps and continuous co-movement of prices, it is crucial to split the two price components in such an analysis, see \cite{todorovrough}, \cite{uni}, \cite{jumpregressions}, \cite{J2}, \cite{caporin}, \cite{ait2020high}, \cite{Pelger} and \cite{jumpfactor}, among others. In particular, \cite{todorovrough} highlight the benefits of this separation for asset pricing by identifying higher risk premiums associated with discontinuous betas. Using jumps to model price shocks in response to news and announcements, the estimation is moreover one main ingredient of several macroeconomic studies based on intra-daily high-frequency data, for instance in the research field on monetary policy, see, among others, \cite{app2} and \cite{app1}.

The remainder of this paper is structured as follows. The theoretical contribution is developed in Section \ref{sec:2}. Section \ref{sec:2.1} discusses the LOMN-model. In Section \ref{sec:2.2} we construct and discuss statistical methods for which asymptotic results are presented in Section \ref{sec:2.3} on uniformly consistent spot volatility estimation and in Section \ref{sec:2.4} on jump detection. All proofs are provided in Section \ref{sec:6}. Although jump tests under MMN and under LOMN are designed for two different models, we emphasize the possibility to apply it to time series coming from the same limit order book data. Based on simulations, in Section \ref{sec:3} we study on the one hand the finite-sample implementation and properties of our new methods and, on the other hand, provide a comparison to the test by \cite{lm12}. For this comparison, we simulate the same efficient log-prices alternatively with both LOMN and MMN to apply our methods and as well the classical method by Lee and Mykland. An empirical analysis of limit order book data in Section \ref{sec:4} compares results for methods based on one-sided noise applied to best ask and best bid quotes and the classical MMN-approach considering mid quotes. This empirical part reveals advantages of the new approach and emphasizes stylized facts of limit order book data which are generally relevant for studies of price jumps. Section \ref{sec:5} concludes.

\section{Theory\label{sec:2}}
\subsection{Model with lower-bounded, one-sided microstructure noise}\label{sec:2.1}
On a filtered probability space, $(\Omega^X,\mathcal{F}^X,(\mathcal{F}^X_t),\mathbb{P}^X)$, the latent, efficient log-price process in continuous time is described by an It\^o semi-martingale   
\begin{eqnarray}X_t&=&X_0+\int_0^t a_s\,ds+\int_0^t\sigma_s\,dW_s + \int_0^t \int_{\mathbb{R}}\delta(s,z)\1_{\{|\delta(s,z)|\leq 1 \}}(\mu-\nu)(ds,dz) \nonumber \\ && +\int_0^t \int_{\mathbb{R}}\delta(s,z)\1_{\{|\delta(s,z)|> 1 \}} \mu(ds,dz)\,,\,t\ge 0\,,\label{sm} \end{eqnarray}
with a one-dimensional standard Brownian motion $(W_t)$, the drift process $(a_t)$, the volatility process $(\sigma_t)$, and with $\delta$ defined on $\Omega\times \mathbb{R}_+\times \mathbb{R}$. The Poisson random measure $\mu$ is compensated by $\nu(ds,dz)=\lambda(dz)\otimes ds$, with a $\sigma$-finite measure $\lambda$. We write
\begin{align}\label{smdecomp}X_t=C_t+J_t\,,\end{align}
with the continuous component $(C_t)$, and the c\`{a}dl\`{a}g jump component $(J_t)$.

In the model with lower-bounded, one-sided microstructure noise,
\begin{align}\tag{LOMN}\label{lomn}Y_i=X_{t_i^n}+\epsilon_i\,,\,i=0,\ldots,n,~~\epsilon_i\stackrel{iid}{\sim}F_{\eta},\epsilon_i\ge 0\,,\end{align}
the discretization $(X_{t_i^n})_{0\le i\le n}$, with high-frequency observations of $(X_t)$ on the fix time interval $[0,1]$, is perturbed by exogenous noise $(\epsilon_i)_{0\le i\le n}$. The baseline model considered in \cite{BJR} and \cite{bibinger2022} is with i.i.d.\ noise and for regularly spaced observation times, $t_i^n=i/n,\,i=0,\ldots,n$. In line with stylized facts of the data, we develop the theory under the following, less restrictive conditions.
\begin{assump}\label{noise} All $(\epsilon_i)_{0\le i\le n}$ have a cumulative distribution function (cdf) $F_{\eta}$ which satisfies
\begin{align}\label{noise_dist}
F_{\eta}(x) =\eta x\big(1+\KLEINO(1)\big) ,\;\mbox{as}~x\downarrow 0\,.
\end{align}
The noise is $m_n$-dependent, such that $\cov(\epsilon_k,\epsilon_l)=0$, for $|k-l|>m_n$, with a sequence $m_n=\KLEINO(\log(n))$.
Suppose that the empirical measure 
\[\mu_n^{S}=\frac{1}{n}\sum_{i=0}^n\delta_{t_i^n}\]
of the $(\mathcal{F}^X$-measurable) observation times $(t_i^n)_{0\le i\le n}$ converges as $n\to\infty$ (almost surely) weakly, $\mu_n^S\stackrel{w}{\rightarrow}\mu^S$, to a limit measure $\mu^S$ which has a continuous, strictly positive Lebesgue density $f^S$ on $[0,1]$. 
%Suppose that the observation times $t_i^n=F^{-1}(i/n)$ are determined by a quantile transformation from the equispaced scheme with $F:[0,1]\to[0,1]$ a continuously differentiable cdf with strictly positive derivative $F^{\prime}$.
\end{assump}
The model with condition \eqref{noise_dist} is nonparametric. Close to the boundary, condition \eqref{noise_dist} means that the extreme value index is $-1$ for the \emph{minimum} domain of attraction. We do not make any assumption about the tails of the noise and its maximum domain of attraction. For instance, a uniform distribution on some interval $[0,A]$, $A>0$, an exponential distribution and a heavy-tailed (shifted) Pareto distribution all satisfy \eqref{noise_dist}. In particular, the developed asymptotic theory for the LOMN-model does not require the existence of moments of the noise. Considering block-wise minima instead of block-wise averages as typically in the MMN-model, this is an important advantage of the statistical methods designed for LOMN. The standard assumption \eqref{noise_dist} on one-sided noise has been imposed in the same way by \cite{jirak2014adaptive} and \cite{BJR}.

Assumption \ref{noise} allows for serially correlated noise. The restriction on the growth of $m_n$ is required for Theorem \ref{thmgumbel} with its fast convergence rate and not for the other results. It is possible to extend Theorem \ref{thmgumbel} to larger $m_n$ by adapting the choice of $h_n$, which would result in a slightly slower rate of convergence. 
%The assumption on the observation times $t_i^n$ is one standard condition for non-regularly spaced observations used in Assumption 3.1 of \cite{BHMR}. It is equivalent to a time-change as used in Sec.\ 5.3 of \cite{bn2}. Different as for volatility estimation in these references, the asymptotic distributions in this work will not depend on $F$. The quantile transformation is sufficient to integrate non-regularly spaced observations comprehensibly in the proofs. We expect that the results hold under even more general assumptions on the sampling scheme.
The assumption on the observation times $t_i^n$ is a mild condition which naturally rules out atoms in the limit measure $\mu^S$ requiring absolute continuity with respect to the Lebesgue measure on $[0,1]$. The $t_i^n$ can be deterministic, or $\mathcal{F}^X$-measurable, possibly depending on $(W_s)$ and $(\sigma_s)$. In the latter case, Assumption \ref{noise} grants \emph{almost sure} weak convergence, i.e., existence of a set $\Omega_0^X\subseteq\Omega^X$, with $\P^X(\Omega_0^X)=1$, on that $\mu_n^S(\omega)\stackrel{w}\to\mu(\omega)$ holds pathwise. A standard specification for non-regularly spaced observations is
\[t_i^n=F^{-1}(i/n),~\mbox{with}~F(t)=\int_0^t f^S(u)\,du\,,\]
i.e., a quantile transformation of an equispaced scheme. This is the setting of Assumption 3.1 of \cite{BHMR}, being equivalent to a time-change as used in Sec.\ 5.3 of \cite{bn2}, and it is sufficient for our weaker assumption. Our assumption also covers cases that cannot be represented via a quantile transformation. For instance, if $t_i^n$ are generated by a (possibly inhomogeneous) Poisson process with intensity $nf^S(t)$. Since conditional on $n$, jump times of a standard Poisson process are i.i.d.\ uniformly distributed on $[0,1]$, it satisfies the weak convergence with a uniform limit distribution, i.e.\ $f^S=1$ constant. Asymptotic distributions of our statistics will not depend on $f^S$, but the proofs of our asymptotic results require the mild regularity that $f^S$ is strictly positive and continuous.

\begin{assump}\label{sigma}
The drift $(a_t)_{t\ge 0}$  is a locally bounded process. The volatility is strictly positive, $\inf_{t\in[0,1]}\sigma_t>0$, $\mathbb{P}^X$-almost surely. 
For all $0\leq t+s\leq1$, $t\ge 0$, $s\ge 0$, with some constants $ C_{\sigma}>0$, and $\alpha>0$, it holds that 
\begin{align}
\label{vola}\E\big[(\sigma_{(t+s)}-\sigma_{t})^2\big]  \le C_{\sigma} s^{2\alpha}\,.\end{align}
\end{assump}
Condition \eqref{vola} imposes a certain regularity of the volatility process, measured by the parameter $\alpha$, but does not rule out volatility jumps. Working under Assumption \ref{sigma} with general $\alpha$, our asymptotic theory is developed in a framework which covers different volatility models recently discussed in the literature. For rough volatility, see \cite{chong2} and references therein, $\alpha$ is given by the Hurst exponent while $\alpha=1/2$ holds under the common assumption that $(\sigma_t)$ is another It\^o semi-martingale. We impose the following standard condition on the jumps.
\begin{assump}\label{jumps}
Assume for the predictable function $\delta$ in \eqref{sm} that \(\sup_{\omega,x}|\delta(t,x)|/\gamma(x)\) is locally bounded with a non-negative, deterministic function $\gamma$ that satisfies
\begin{align}\label{BG}\int_{\mathbb{R}}(\gamma^r(x)\wedge 1)\lambda(dx)<\infty\,.\end{align}
\end{assump}
The notation $a\wedge b=\min(a,b)$, and $a\vee b=\max(a,b)$, is used throughout this manuscript. The generalized Blumenthal-Getoor or jump activity index $r,\, 0\le r\le 2$, in \eqref{BG} determines the jump activity. The most restrictive case is $r=0$, when jumps are of finite activity. The larger $r$, the more general jump components are included.

\subsection{Statistical methods\label{sec:2.2}}
We first discuss inference for a possible jump,
\[\Delta X_{\tau}=X_{\tau}- X_{\tau-}=X_{\tau}- \lim_{u\uparrow \tau}X_{u}\,,\]
at some given (stopping) time $\tau\in(0,1)$. For this purpose, consider 
\begin{align}\label{locest2}\hat X_{\tau}=\min_{i=\lfloor n\tau\rfloor+1,\ldots,\lfloor n\tau\rfloor+nh_n} Y_{i}\,,
~\hat X_{\tau-}=\min_{i=\lfloor n\tau\rfloor-nh_n+1,\ldots,\lfloor n\tau\rfloor} Y_{i}\,.\end{align}
These statistics are local minima of $nh_n$ noisy observations over blocks, where we choose the block length $h_n$ such that $nh_n$ is integer-valued. The two disjoint blocks contain observations in a vicinity shortly after and before time $\tau$, respectively. We can estimate the jump $\Delta X_{\tau}$ based on 
\begin{align}\label{bhr}\widehat {\Delta X_{\tau}}=\hat X_{\tau}-\hat X_{\tau-}\,.\end{align}
Asymptotically, $h_n\to 0$, and $nh_n\to\infty$, as $n\to\infty$. Therefore, $\tau\in(h_n,1-h_n)$, for $n$ sufficiently large and any $\tau\in(0,1)$, such that we omit a discussion of adjustments for boundary cases when $\tau\notin(h_n,1-h_n)$.  
In order to construct a global test for jumps and to perform volatility estimation, we partition the whole observation interval $[0,1]$ in $h_n^{-1}$ equispaced blocks, $h_n^{-1}\in\N$, and take local minima on each block. Consider, for $k=0,\ldots,h_n^{-1}-1$, the local block-wise minima
\begin{align}\label{localmin}m_{k,n}=\min_{i\in\mathcal{I}_k^n}Y_i~,~\mathcal{I}_k^n=\{i\in\{0,\ldots,n\}: ~t_i^n \in (kh_n,(k+1)h_n)\}\,.\end{align}
Here, $h_n^{-1}$ is an integer, while in general $nh_n$ not. In particular, $h_n$ can be different for the local and for the global statistics. This is necessarily the case when $n$ is such that a choice $h_n^{-1}\in\N$ and $nh_n\in\N$ is not possible. However, since the asymptotic orders of optimal block lengths will be identical, we use for simplicity the same notation $h_n$ for the block lengths in the construction of local and global statistics. 

Under the global null hypothesis of no price jumps, a consistent estimator for the spot squared volatility $\sigma_{\tau}^2$ is given by 
\begin{align}\label{simpleestimator}\hat\sigma^2_{\tau-}=\frac{\pi}{2(\pi-2)K_n}\sum_{k=(\lfloor h_n^{-1}\tau\rfloor-K_n)\vee 1}^{\lfloor h_n^{-1}\tau\rfloor-1}h_n^{-1}\big(m_{k,n}-m_{k-1,n})^2\,,\end{align}
with sequences $h_n\to 0$ and $K_n\to\infty$, such that $K_nh_n\to 0$. This estimator is available online at time $\tau$ during a trading day because it relies only on past observations before time $\tau$. 
Working with ex-post data over the whole interval, one may use as well
\begin{align}\label{simpleestimator2}\hat\sigma^2_{\tau}=\frac{\pi}{2(\pi-2)K_n}\sum_{k=(\lfloor h_n^{-1}\tau\rfloor-(K_n-1)/2)\vee 1}^{(\lfloor h_n^{-1}\tau\rfloor+(K_n-1)/2)\wedge (h_n^{-1}-1)}h_n^{-1}\big(m_{k,n}-m_{k-1,n})^2\,,\end{align}
for some odd integer $K_n$. A difference of estimators 
\begin{align}\label{simpleestimator3}\hat\sigma^2_{\tau+}=\frac{\pi}{2(\pi-2)K_n}\sum_{k=(\lfloor h_n^{-1}\tau\rfloor+1)}^{(\lfloor h_n^{-1}\tau\rfloor+K_n)\wedge (h_n^{-1}-1)}h_n^{-1}\big(m_{k,n}-m_{k-1,n})^2\,,\end{align}
over a window after time $\tau$ and \eqref{simpleestimator} over a window before time $\tau$ allows to infer a possible jump in the volatility process at time $\tau$. Minima and maxima in the lower and upper summation limit are only relevant when $\tau$ is close to the boundaries in small intervals with lengths that tend to zero. In these boundary cases, the factor $K_n^{-1}$ can also be adjusted to get an average over available $k$. Since the boundary effects are not relevant for the asymptotic theory, we do, however, not incorporate such adjustments in \eqref{simpleestimator}, \eqref{simpleestimator2} and \eqref{simpleestimator3}.

A spot volatility estimator which is robust with respect to jumps in $(X_t)$, is obtained with thresholding. We truncate differences of local minima whose absolute values exceed a threshold $u_n= \beta^{tr}\cdot h_n^{\kappa},\kappa\in(0,1/2)$, with a constant $\beta^{tr}>0$, which leads to
\begin{align}\label{truncatedestimator}\hat\sigma^{2,(tr)}_{\tau-}=\frac{\pi}{2(\pi-2)K_n}\sum_{k=(\lfloor h_n^{-1}\tau\rfloor-K_n)\vee 1}^{\lfloor h_n^{-1}\tau\rfloor-1}h_n^{-1}\big(m_{k,n}-m_{k-1,n})^2\1_{\{|m_{k,n}-m_{k-1,n}|\le u_n\}}\,.\end{align}
The estimators $\hat\sigma^{2,(tr)}_{\tau}$, and $\hat\sigma^{2,(tr)}_{\tau+}$, are constructed analogously.

Our global test for price jumps is based on the maximum statistic
\begin{align}\label{bhrgl}T^{BHR}(Y_0,Y_1\ldots,Y_n):=\max_{k=1,\ldots,h_n^{-1}-1}\Big|\frac{m_{k,n}-m_{k-1,n}}{\hat\sigma_{kh_n}}\Big|\,,\end{align}
with $\hat\sigma_{kh_n}=\big(\hat\sigma^{2,(tr)}_{kh_n-}\big)^{1/2}$, for $k\ge K_n$, and $\hat\sigma_{kh_n}=\big(\hat\sigma^{2,(tr)}_{kh_n+}\big)^{1/2}$, else.

\begin{figure}[t]
\begin{framed}
\includegraphics[width=6.5cm]{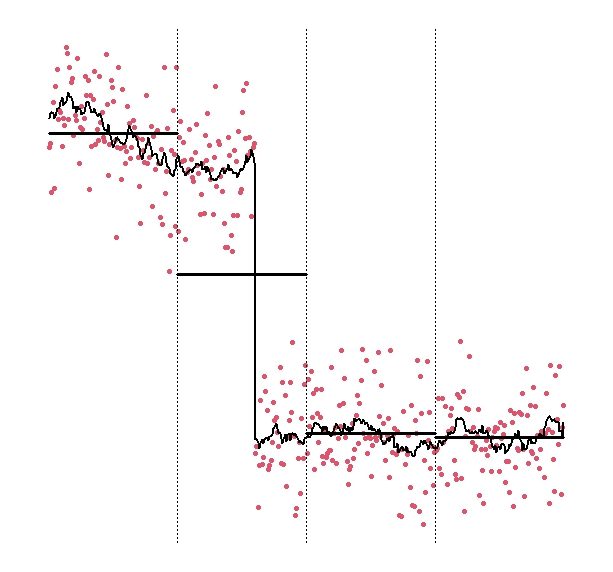}~~~~\includegraphics[width=6.5cm]{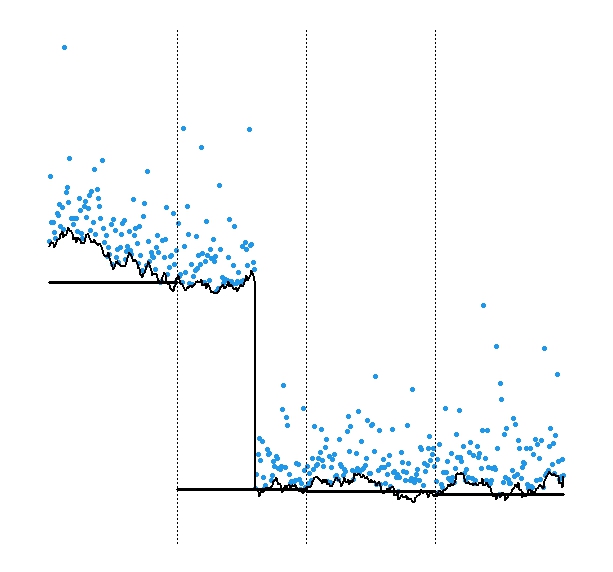}\end{framed}
\caption{\label{Fig:locmin}Simulated efficient price with one jump and local averages of MMN observations (left) and local minima of LOMN observations (right).}
\end{figure}

A benefit of jump tests based on maximum statistics is that they can be used, moreover, to locate jump times. If the test rejects, when statistic \eqref{bhrgl} exceeds a critical value, the time at which the maximum is taken consistently estimates the position of the largest absolute jump. We establish a precise result on this localization in Proposition \ref{localize}. A sequential application and a combination with \eqref{bhr} allow the identification of jumps with their times and sizes.

Similar methods as developed here for observations with LOMN have been discussed in the literature for observations under MMN. While our mathematical analysis requires some new proof techniques, we illustrate next that inference on jumps is statistically less involved using local minima of observations with LOMN instead of local averages of observations with MMN.  Figure \ref{Fig:locmin} shows a simple, illustrative example with four blocks and a simulated diffusive efficient price (solid line) with one huge downward jump. While this efficient price is unobservable in practice, observations with additive noise are available. On the left-hand side we show observations with i.i.d.\ centered normally distributed MMN, while on the right-hand side we depict observations with i.i.d.\ exponentially distributed LOMN. Both noise variances coincide. For MMN, the estimation of jumps is based on local, block-wise averages which are drawn as bars in the left plot of Figure \ref{Fig:locmin}. Instead of identifying one jump, the differences between local averages rather suggest two adjacent jumps of smaller sizes. This effect, caused by averaging over a jump time, has been highlighted by \cite{zhangmykland3} as the ``pulverisation of jumps by pre-averages''. It is a difficult problem when estimating jumps at a priori unknown jump times. Theoretically, it can only be solved in asymptotics with sophisticated statistical methods, see, \cite{vetter} and \cite{bibwinkellev}. In finite samples this pulverization effect impairs jump detection under MMN. In contrast, the differences between local minima of observations with LOMN, which are drawn as bars in the right plot of Figure \ref{Fig:locmin}, correctly suggest one jump and accurately estimate its size. Having a downward jump on the second block, the minimum on this block is taken after the jump. In case of an upward jump, the minimum would be taken before the jump. In any case, this results in only one large difference between block-wise, local minima of the correct size. So, there is \emph{no pulverization effect for local minima under LOMN}. Figure \ref{Fig:locmin} reveals another advantage of using local minima under LOMN. Since the cumulative (=running) minimum on a block is always greater or equal the total minimum, a jump is detected as soon as the distance from an observation to the local minimum on the previous block exceeds a critical value of the test. In this example, this is already the case for the \emph{first} observation after the jump. Since no such monotonicity relation exists between cumulative and total average on a block, online jump detection under LOMN has an \emph{advantage of speed}.

In \ref{S2}, additional plots are provided which complement the illustration considering moreover upward jumps, local maxima of one-sided, upper-bounded noise, illustrations for a jump occurring near the boundary between two blocks and an example with real data. For data, it is standard to compute local averages from samples referring to both sides of the order book such as trade prices or mid quotes. We propose to compute local minima from samples referring to the upper side of the order book only, such as best ask quotes or seller-initiated trades. At the same time, local maxima of bid quotes or buyer-initiated trades are also to be considered in applications.

\subsection{Uniformly consistent spot volatility estimation\label{sec:2.3}}
We begin with the asymptotic theory on spot volatility estimation. The expectation of the volatility estimator hinges on the function
\begin{align}\label{psi}\Psi_n(\sigma^2)&:=\frac{\pi}{2(\pi-2)} h_n^{-1}\E\Big[\Big(\min_{i\in\{0,\ldots,nh_n-1\}}\big(\sigma B_{t_i^n}+\epsilon_i\big)-\min_{i\in\{1,\ldots,nh_n\}}\big(\sigma \tilde B_{t_i^n}+\epsilon_i\big)\Big)^2\Big]\\
&\notag=\frac{\pi}{\pi-2}\, h_n^{-1}\var\Big(\min_{i\in\{0,\ldots,nh_n\}}\big(\sigma B_{t_i^n}+\epsilon_i\big)\Big)(1+\KLEINO(1)),\end{align}
where $(B_t)$ and $(\tilde B_t)$ denote two independent standard Brownian motions.
In \cite{BJR} it was proved that $\Psi_n$ is monotone and invertible. For $n h_n^{3/2}m_n^{-1}\to \infty$, we have
\begin{align}\label{psiapprox}\Psi_n(\sigma^2)=\sigma^2+\KLEINO(1),~\mbox{as}~n\to\infty\,,\end{align}
for any $\sigma^2>0$. A proof for i.i.d.\ noise is given as Step 3 of the proof of Theorem 1 in \cite{bibinger2022} and generalizes to $m_n$-dependent noise with Lemma \ref{lem4}. By \eqref{psiapprox}, we do not require knowledge of $\Psi_n$ for the construction of a consistent estimator.

If $(X_t)$ is continuous, i.e., $J_t=0$ in \eqref{smdecomp}, under Assumptions \ref{sigma} and \eqref{noise_dist} with i.i.d.\ noise, \cite{bibinger2022} proved that when $h_n$ is chosen such that $n h_n^{3/2}\to \infty$, the spot volatility estimator \eqref{simpleestimator} satisfies
\begin{align}\label{spotclt}\hat\sigma^2_{\tau-}-\Psi_n\big(\sigma_{\tau-}^2\big) =\mathcal{O}_{\P}\big(K_n^{-1/2}\big)\,,\end{align}
under the condition
\begin{align}\label{K_n}K_n=C_K h_n^{\delta -2\alpha/(1+2\alpha)},~\mbox{for $0<\delta<2\alpha/(1+2\alpha)$, and with a constant $C_K>0$}\,.\end{align}
Moreover, consistency holds true in general. However, the estimation error $\hat\sigma^2_{\tau-}-\sigma_{\tau-}^2 =\KLEINO_{\P}(1)$ does in general not decay as fast as the one in \eqref{spotclt}. The function $\Psi_n$ is feasible, as \cite{bibinger2022}, Sec.\ 5.1, shows how the function $\Psi_n$ and its inverse $\Psi_n^{-1}$ can be approximated numerically. Analogous results apply in case of the modified versions \eqref{simpleestimator2} and \eqref{simpleestimator3}, respectively. Under the same setup with jumps satisfying Assumption \ref{jumps} with 
\begin{align}\label{jumpres}r< \frac{2+2\alpha}{1+2\alpha}\,,\end{align}
the truncated spot volatility estimator \eqref{truncatedestimator} with
\begin{align}\kappa\in \Big(\frac{1}{2-r}\frac{\alpha}{2\alpha+1},\frac12\Big)\,,\end{align} 
satisfies
\begin{align}\label{spotclttr}\hat\sigma^{2,(tr)}_{\tau-}-\Psi_n\big(\sigma_{\tau-}^2\big)=\mathcal{O}_{\P}\big(K_n^{-1/2}\big)\,.\end{align}
Moreover, feasible central limit theorems and asymptotic confidence intervals for the estimators are established in \cite{bibinger2022}.  
The convergence rate $K_n^{-1/2}$ gets arbitrarily close to $n^{ -2\alpha/(3+6\alpha)}$, which is optimal in the LOMN-model. In the important special case when $ \alpha=1/2$, for a semi-martingale volatility, the rate is arbitrarily close to $n^{ -1/6}$. This is much faster than the known optimal rate of convergence in the MMN-model, which is $n^{ -1/8}$, see \cite{BHMR2}. In \eqref{jumpres} we impose mild restrictions on the jump activity. For the standard model with a semi-martingale volatility, i.e., $\alpha=1/2$, we require that $r<3/2$. For $\alpha=1$, we have the strongest condition implying $r<4/3$. 

In this work, we do not require central limit theorems for spot volatility estimation. Instead, the asymptotic theory for the global jump test relies on \emph{uniformly consistent} spot volatility estimation. Uniform consistency in functional estimation is typically much more difficult to prove than pointwise results. We prove a quite strong result under surprisingly mild assumptions.
\begin{prop}\label{uniform}
Under Assumptions \ref{noise} and \ref{sigma} and when there are no jumps in $(X_t)$ and $(\sigma_t)$, the spot volatility estimator \eqref{simpleestimator} with $ n h_n^{3/2}m_n^{-1}\to \infty$, and $K_n$ chosen as in \eqref{K_n}, satisfies
\begin{align*}\max_{k=1,\ldots,h_n^{-1}-1}\big|\hat\sigma_{kh_n-}^2-{{\Psi_n}}\big(\sigma_{kh_n}^2\big)\big|=\KLEINO_{\P}\Big(K_n^{-\gamma}\Big)\,,\end{align*}
for all $\gamma$, with $\gamma<1/2$.
\end{prop}
Condition \eqref{K_n} grants that the squared bias of order $(K_nh_n)^{2\alpha}$ is smaller in the MSE than the variance which decays with order $K_n^{-1}$. This yields the pointwise rate $K_n^{-1/2}$, and the uniform rate is almost the same. For all upcoming results, we do not require a certain rate of the volatility estimation and also not \eqref{K_n}. In particular, the regularity $\alpha$ from \ref{sigma} does not need to be known and the asymptotic results are valid for any $\alpha>0$.
\begin{cor}\label{uniformcor}Under the assumptions of Proposition \ref{uniform}, if $n h_n^{3/2}m_n^{-1}\to \infty$, the spot volatility estimator \eqref{simpleestimator} is uniformly consistent, i.e.
\[\max_{k=1,\ldots,h_n^{-1}-1}\big|\hat\sigma_{kh_n-}^2-\sigma_{kh_n}^2\big|=\KLEINO_{\P}\big(1\big)\,.\] 
\end{cor}
It is clear that consistency uniformly over the interval $(0,1)$ requires the assumption of a continuous volatility, see e.g., the discussion in Sec.\ 2.2 of \cite{uni2}. A generalization of this result to jumps in $(X_t)$ using the threshold estimator is possible. However, for the construction of our test we will rely on Corollary \ref{uniformcor}. Under MMN a uniformly consistent volatility estimation requires the existence of all moments of the noise, see \cite{mehmet}. It is clear that Rosenthal-type inequalities or related results to prove the uniformity require existence of higher moments of the considered statistics, in our case of the local minima. From this point of view, it might be surprising that we do not have to impose such assumptions for Proposition \ref{uniform}. Although our proof relies as well on maximal and moment inequalities, this is not the case here, since we only need \emph{moments of the local minima} for which \eqref{noise_dist} is sufficient. This is a crucial advantage of inference based on local order statistics compared to local averages, since higher moments of local averages exist only if the corresponding noise moments exist.

\subsection{Asymptotic results on the identification of jumps\label{sec:2.4}}
We start with an asymptotic result on the inference for jumps at some pre-specified time $\tau\in(0,1)$.
\begin{theo}\label{thmloc}
Under Assumptions \ref{noise}, \ref{sigma} and \ref{jumps}, for $n h_n^{3/2}m_n^{-1}\to \infty$, $\widehat {\Delta X_{\tau}}$ from \eqref{bhr} satisfies the stable weak convergence 
\begin{align}\label{stablecltbhr}h_n^{-1/2}\,\big(\widehat {\Delta X_{\tau}}-\Delta X_{\tau}\big)\stackrel{st}{\longrightarrow} Z_2-Z_1\,,\end{align}
with two random variables
\begin{align}Z_1\sim \hmn\big(0,\sigma_{\tau}^2\big)~,~Z_2\sim \hmn\big(0,\sigma_{\tau-}^2\big)\,,\end{align}
which are conditionally on $(\sigma_t)$ independent. HMN refers to the half mixed normal distribution, that is, $Z_1\stackrel{d}{=}\sigma_{\tau}|U|$, for $U\sim\mathcal{N}(0,1)$ standard normal. If $\Delta X_{\tau}=0$, 
\begin{align}\label{stablecltf}h_n^{-1/2}\,\Big(\frac{\hat X_{\tau}}{\hat\sigma^{(tr)}_{\tau+}}-\frac{\hat X_{\tau-}}{\hat\sigma^{(tr)}_{\tau-}}\Big)\stackrel{d}{\longrightarrow} \tilde Z_2-\tilde Z_1\,,\end{align}
with two independent random variables
\begin{align}\label{eqlocaltest}\tilde Z_1\sim \hn\big(0,1\big)~,~\tilde Z_2\sim \hn\big(0,1\big)\,.\end{align}
HN refers to the standard half-normal distribution.
\end{theo}
The standardization in \eqref{stablecltf}, where $\hat\sigma^{(tr)}_{\tau-}$ and $\hat\sigma^{(tr)}_{\tau+}$ are the square roots of the estimators \eqref{truncatedestimator} and the truncated version of \eqref{simpleestimator3}, takes into account possible simultaneous price and volatility jumps, see \cite{voljumps} and \cite{bibwinkellev} for empirical evidence of such simultaneous jumps. Figure~\ref{Fig:coh} moreover illustrates two examples. Stable convergence is stronger than weak convergence and is important here, since the limit random variables hinge on the stochastic volatility. More precisely, we prove stability with respect to $\mathcal{F}^X$ in the sense of \cite{JP}, Sec.\ 2.2.1. Since the asymptotic distribution does not hinge on the noise level $\eta$, in contrast to methods for MMN, we do not require any pre-estimation of noise parameters. Moreover, our methods and results remain valid for time-varying noise levels $\eta_t$ in \eqref{noise_dist}, under the mild assumption that $0<\eta_t<\infty$, for all $t$. 

Theorem \ref{thmloc} shows that {\emph{we can consistently estimate price jumps}}. Setting $h_n=\log(n)n^{-2/3}$, $n h_n^{3/2}m_n^{-1}\to \infty$ is satisfied and the convergence rate is $n^{-1/3}$, up to a logarithmic factor. Moreover, asymptotic confidence is feasible and \eqref{eqlocaltest} yields the following local jump test.
\begin{cor}\label{corollary}The test for the null hypothesis that there is no jump at time $\tau\in(0,1)$, $\Delta X_{\tau}=0$, which rejects when
\begin{align}h_n^{-1/2}\,\Big|\frac{\hat X_{\tau}}{\hat\sigma^{(tr)}_{\tau+}}-\frac{\hat X_{\tau-}}{\hat\sigma^{(tr)}_{\tau-}}\Big|>q_{1-\alpha}^{L}\end{align}
with $q_{1-\alpha}^{L}$ being the $(1-\alpha)$ quantile of the distribution of $|\tilde Z_2-\tilde Z_1|$ and the half-normal random variables from \eqref{eqlocaltest}, has asymptotic level $\alpha$ and asymptotic power 1. That is, under $H_1:|\Delta X_{\tau}|>0$, and $\sigma_{\tau}^{-1}J_{\tau}\ne \sigma_{\tau-}^{-1}J_{\tau-}$, it holds for any $\alpha>0$ that
\begin{align}\label{consistloc}\lim_{n\to\infty}\P_{H_1}\bigg(h_n^{-1/2}\,\Big|\frac{\hat X_{\tau}}{\hat\sigma^{(tr)}_{\tau+}}-\frac{\hat X_{\tau-}}{\hat\sigma^{(tr)}_{\tau-}}\Big|>q_{1-\alpha}^{L}\bigg)=1\,.\end{align}
\end{cor}
Standardizing the local minima with estimates of the volatility before and after time $\tau$, the method is robust with respect to a volatility jump, $\Delta\sigma_{\tau}\ne 0$. It is intuitive that this requires the mild assumption that $\sigma_{\tau}^{-1}J_{\tau}\ne \sigma_{\tau-}^{-1} J_{\tau-}$.
\begin{remark}
In the MMN-model jumps can be estimated with an optimal rate of convergence $n^{-1/4}$ based on $n$ equidistant high-frequency observations, see, for instance, the LAN result in Proposition 5.2 of \cite{koike2017}. Hence, the rate of convergence for LOMN is faster. A lower bound for the rate in the LOMN-model is $n^{-1/3}$, and we attain this optimal rate up to a logarithmic factor. For the same reasons as for volatility estimation discussed in Remark 1 of \cite{bibinger2022}, it is not possible to develop feasible tests attaining the rate $n^{-1/3}$ without a logarithmic factor.
\end{remark}
Using extreme value theory, we present a result for a global test for price jumps in the spirit of the Gumbel test by Lee and Mykland for high-frequency prices without noise in \cite{leemykland} and with MMN in \cite{lm12}, respectively.
\begin{theo}\label{thmgumbel}On the null hypothesis of no jumps,
\begin{align}\label{globalhypo}H_0:\sup_{\tau\in [0,1]}|\Delta X_{\tau}|=0\,,\end{align}
under Assumptions \ref{noise}, \ref{sigma} and \ref{jumps} with $(\sigma_t)$ Hölder continuous with regularity $\alpha$, with $h_n=2\log(2h_n^{-1}-2)n^{-2/3}(1+\KLEINO(1))$, the statistic \eqref{bhrgl} satisfies the Gumbel convergence
\begin{align}\label{gumbel}n^{1/3}\,T^{BHR}(Y_0,Y_1\ldots,Y_n)-2\log(2h_n^{-1}-2)+\log\big(\pi\log(2h_n^{-1}-2)\big)\stackrel{d}{\longrightarrow}\Lambda\,,\end{align}
where $\Lambda$ refers to the standard Gumbel distribution, that is, it holds with $B_n=2\log(2h_n^{-1}-2)-\log\big(\pi\log(2h_n^{-1}-2)\big)$, for all $x\in\R$ that 
\begin{align}\label{gumbelex}
\lim_{n\to\infty}\P_{H_0}\Big(n^{1/3}\,T^{BHR}(Y_0,Y_1\ldots,Y_n)-B_n\le x\Big)= \exp\big({-e^{-x}}\big)\,.
\end{align}
The test that rejects $H_0$ whenever
\begin{align}n^{1/3}\,T^{BHR}(Y_0,Y_1\ldots,Y_n)-B_n>q_{1-\alpha}^{\thinspace\Lambda}\,,\end{align}
for $q_{1-\alpha}^{\thinspace\Lambda}$ being the $(1-\alpha)$ quantile of the Gumbel distribution, has asymptotic level $\alpha$. 
Moreover, under the alternative hypothesis that
\begin{align}H_1:\liminf_{n\to\infty}n^{\beta}\sup_{\tau\in(0,1)}|\Delta X_{\tau}|>0,~\mbox{for some}~\beta<1/3\,,\end{align}
under Assumptions \ref{noise}, \ref{sigma} and \ref{jumps}, the test rejects asymptotically with probability 1:
\begin{align}\label{consist}\lim_{n\to\infty}\P_{H_1}\bigg(n^{1/3}\,T^{BHR}(Y_0,Y_1\ldots,Y_n)-B_n>q_{1-\alpha}^{\thinspace\Lambda}\bigg)=1\,.\end{align}
\end{theo} 
The condition of a continuous volatility is required to use Proposition \ref{uniform} in the proof of the Gumbel convergence on the null hypothesis. Since volatility jumps are typically associated with events that simultaneously trigger price jumps, it is not too restrictive to work under a null hypothesis that there are no price and no volatility jumps. For the results under the alternative hypothesis, we only require pointwise consistency of the spot volatility estimator and allow for volatility jumps.

In contrast to the Gumbel convergences in \cite{leemykland}, \cite{lm12} and \cite{bibwinkellev}, we cannot trace back our result \eqref{gumbel} to the Gumbel convergence for the maximum of i.i.d.\ standard normally distributed random variables. Instead, we prove that the statistic $T^{BHR}$ can be approximated by the maximum of absolute differences of 1-dependent half-normally distributed random variables. We then establish the extreme value theory for these random variables. Since we expect this and related results to be of interest in their own right for extreme value theory and its applications to various high-frequency jump tests, the result embedded into a more general theory has been provided in \cite{arxiv2}. \cite{Nunes} adds a recent discussion of the convergence rates obtained in \cite{leemykland}, \cite{lm12} and \cite{arxiv2} and shows that they are coherent.

The result \eqref{consist} implies \emph{consistency} of the test, i.e., it rejects asymptotically almost surely if there is a jump, $\sup_{\tau\in [0,1]}|\Delta X_{\tau}|>0$. The stronger result \eqref{consist} addresses moreover \emph{local alternatives}. The sequence of tests can detect jumps with decreasing sizes in $n$, as long as the jump sizes decrease slower than $n^{-1/3}$. This result provides information about what jump sizes can be detected for a given sample size $n$. While in the MMN-model, $n^{-1/4}$ is the optimal rate for local alternatives, our rate $n^{-\beta}$, for any $\beta<1/3$ , is much faster. This shows that for the same sample size, we can \emph{detect much smaller jumps}.

The identification of jumps is complemented by the following localization of jump times.
\begin{prop}\label{localize}If the test rejects, and if there is exactly one jump at time $\theta\in(0,1)$, the jump time is consistently estimated by
\begin{align}\hat\theta_n=h_n\operatorname{argmax}_{1\le k\le h_n^{-1}-1}\frac{|m_{k,n}-m_{k-1,n}|}{\hat\sigma_{kh_n}}\,,\end{align}
with rate of convergence $|\hat\theta_n-\theta|=\mathcal{O}_{\P}(h_n)$.
\end{prop}
The argmax identifies the block of length $h_n$ containing the jump, and consistency holds as $h_n\to 0$. If multiple jumps occur on different blocks under the alternative, the argmax-estimator identifies the block with the largest absolute jump. A sequential application of the test and argmax-estimator allows consistent estimation of the multiple jump times. In each iteration of the sequential test, the previous maximum absolute difference, and thus the previously detected block with a jump, is discarded. If there is a finite number $N$ of blocks containing jumps, $N$ can be estimated consistently using a sequence of levels $\alpha_n\to 0$, adopting the procedure from Theorems 3-5 by \cite{leemykland}. Jump sizes can be estimated consistently by Theorem 1, such that our methods allow full recovery of the jumps in the efficient price. Moreover, associated \emph{online jump detection} is feasible, where the \emph{advantage of speed} for local minima of observations under LOMN -- illustrated in Figure 2 -- becomes relevant.
%In case of multiple jumps under the alternative, the estimator locates the time of the maximum absolute jump. A sequential application of the test and this argmax-estimator then allow consistent estimation of multiple jump times. In each iteration of the sequential test, the previous maximum absolute difference is discarded. If there is a finite number $N$ of jumps, moreover, $N$ can be estimated consistently using a sequence of levels $\alpha_n\to 0$, adopting the procedure from Theorems 3-5 by \cite{leemykland}. Jump sizes can be estimated consistently by Theorem \ref{thmloc}, such that our methods allow to fully recover jumps of the efficient price. Moreover, associated \emph{online jump detection} is feasible, where the \emph{advantage of speed} for local minima of observations under LOMN illustrated in Figure \ref{Fig:locmin} becomes relevant.
\begin{cor}\label{onlinecor}
If for an observation $Y_i$ on the $k$th block, $i\in\mathcal{I}_k^n$, the distance to the local minimum $m_{k-1,n}$ on the previous block, with $\mathcal{I}_k^n$ and $m_{k-1,n}$ defined in \eqref{localmin}, exceeds the threshold
\begin{align}\label{onlineeq}Y_i-m_{k-1,n}<n^{-1/3}\hat\sigma_{kh_n}\big(\log(-\log(1-\alpha))-B_n\big)\,,\end{align}
with $B_n$ from Theorem \ref{thmgumbel} and $\hat\sigma_{kh_n}$ the online spot volatility estimator as in \eqref{bhrgl}, a negative jump of the efficient price $(X_t)$ is detected at level $\alpha$ and located on the time interval $(kh_n,t_i^n)$.
\end{cor}
If there is a jump, it is detected online with asymptotic probability 1. A false detection has asymptotic probability $\alpha$ given by the level. Like the global test in Theorem \ref{thmgumbel}, online jump detection can also detect positive jumps and separate positive and negative jumps. The advantage of speed of local minima, however, is particularly of interest to detect negative jumps online as fast as possible. In practice, Corollary \ref{onlinecor} is implemented with running local minima, because a running minimum over a block equals $Y_i$, as soon as $Y_i$ enters this running minimum -- thanks to the monotonicity of minima with respect to subset relations. Local averages of observations under MMN can -- in the best case if not perturbed by the pulverization effect -- detect a jump on the $k$th block at its endpoint $(k+1)h_n$. Moreover, the block lengths $h_n$ used for local averages are proportional to $n^{-1/2}$ and hence longer than for our methods applied to observations with LOMN. As illustrated in Figure \ref{Fig:locmin}, the estimation of the jump time with accuracy $t_i^n-kh_n$ can be more precise than the general upper bound $h_n$ in Proposition \ref{localize}. An analogous result holds true for local maxima of observations with one-sided, upper-bounded noise, which show an advantage of speed in detecting positive jumps, see the illustration in Figure \ref{plotR2} of \ref{S2}.

\section{Simulations and finite sample behavior\label{sec:3}}
The aim of this simulation study is twofold. Firstly, we evaluate the finite-sample performance of the main theoretical results and, secondly, we provide a comparison of jump tests for the LOMN-model and the MMN-model. For both cases, we simulate $n = 23{,}400$ observations, corresponding to one observation per second over a (NASDAQ) trading day of 6{.}5 hours. The efficient log-price process under the null hypothesis is sampled from
\begin{subequations}
\begin{align}
	dX_t &= v_t \sigma_t dW_t \label{dXt}\\
	\label{sim2}d\sigma_t^2 &= 0.0162 \cdot (0.8465 -\sigma_t^2)dt + 0.117\cdot \sigma_tdB_t\\
	\label{sim3}v_t &= (1.2-0.2\cdot \sin(3/4 \pi t)) \cdot 0.01 \quad\text{with}\quad t\in[0,1].
\end{align}
\end{subequations}
The factor $v_t$ generates a typical U-shaped intraday volatility pattern and $(W_t, B_t)$ is a two-dimensional Brownian motion with leverage $d[W,B]_t = -0.5dt$. This setup captures a variety of realistic features of financial high-frequency data. Variants thereof are frequently employed in the literature, see, e.g., \citet{lm12} and  \citet{bibwinkellev} as well as the references therein. Under the alternative, \eqref{dXt} is augmented by a jump occurring at a random time point, however neither close to the beginning nor to the end of the sampled trajectory. We consider both positive and negative jumps. The sizes of the jumps under the alternative are given in Tables \ref{tab:most_global} -- \ref{tab:local_comparison} in absolute value. They are chosen to illustrate the transition from non-detectable to detectable sizes depending on the noise level as well as the block length $h_n$ as discussed below.
R code and replication files for all simulations are publicly available.\footnote{\href{https://github.com/bibinger/LOMN}{github.com/bibinger/LOMN}}

\begin{table}
\centering
	\begin{tabular}{c|ccccccc}
		\hline\hline
			 & &  \multicolumn{6}{c}{$\vert \text{jump size}\vert$} \\
		\hline
		 $q$ & $0.0\%$ & $0.1\%$ & $0.2\%$ & $0.3\%$& $0.4\%$ & $0.5\%$\\
		\hline
			 $ 0.010\%$ & 0.05 & 0.09 & 0.65 & 0.98 & 1.00 & 1.00 \\
			 $ 0.025\%$ & 0.06 & 0.09 & 0.64 & 0.98 & 1.00 & 1.00 \\
		 $ 0.050\%$ & 0.05 & 0.08 & 0.60 & 0.96 & 1.00 & 1.00 \\
			 $ 0.075\%$ & 0.05 & 0.09 & 0.58 & 0.96 & 1.00 & 1.00 \\
			 $ 0.100\%$ & 0.05 & 0.08 & 0.55 & 0.95 & 1.00 & 1.00  \\
		\hline\hline
	\end{tabular}
	\caption{Simulation results for the global test on a significance level of $\alpha=5\%$. The column $\vert \text{jump size}\vert=0.00\%$ gives the estimated size and the columns with $\vert \text{jump size}\vert>0\%$ give the estimated power for the corresponding jump sizes.}
	\label{tab:most_global}
\end{table}
\subsection{Size and power of the global test from Theorem~\ref{thmgumbel}}\label{sec:sim_global}
Based on the setup for $(X_t)$, \eqref{dXt} - \eqref{sim3}, the noisy observations are generated by
\begin{align}
	Z_{i}&= X_{i/n} + \epsilon_{i} \quad\text{with}\quad\epsilon_{i+1}=-0.5 \epsilon_{i} + \eta_{i+1} \quad\text{and}\quad \eta_i\overset{iid}{\sim}N(0, q^2),\quad i=0,\ldots,n,\label{sim_autoc}
\end{align}
where the noise level $q$ is shown in Table~\ref{tab:most_global}. The noise is serially correlated. Autocorrelation patterns of high-frequency returns similar to first-order autoregressive time series with negative parameter are documented in several papers, e.g., see \cite{robert2010} and \cite{Griffin2008}, and for our data this is found in \ref{S3}. It should be noted that we first impose standard two-sided Gaussian noise in \eqref{sim_autoc} and then generate an implicit sample of simulated ask quotes $Y_{i}^{A}$ by using $Y_{i}^{A}= Z_{i}$, whenever $Z_{i} > X_{i/n}$, what yields one-sided additive noise. Bid quotes can be simulated setting $Y_{i}^{B}= Z_{i}$, whenever $Z_{i} < X_{i/n}$. The main reason for this implicit sampling scheme lies in illustrating the coherence of the proposed test with samples stemming from MMN-models. As in our data analysis, methods for MMN and LOMN can be applied to time series based on the same data set. Relying on this sampling scheme makes also clear that our test can be applied to transaction data if information is available whether the transaction is buyer-initiated or seller-initiated. Furthermore, this sampling scheme incorporates the realistic feature that the effective sample size of ask or bid quotes is just about half of the original sample size $n$ of available mid quotes and the observations $Y_{i}^{A}$ are in contrast to $Z_{i}$ not equally spaced any more. The results of this sub-section are based on simulated ask quotes only. 

To perform the jump test following Theorem~\ref{thmgumbel}, the block length $h_n$ and the spot volatility estimates $\hat\sigma_{kh_n}^2$ have to be determined. Our asymptotic results are worked out under the condition that $h_n>n^{-2/3}$, we choose $h_n = 1{.}3 \cdot n^{-2/3}$ here. The spot volatility is estimated by \eqref{simpleestimator2}. As in the simulations of \citet{bibinger2022}, the spot volatility estimator is tuned rather conservatively by averaging over many $(K_n = 200)$ relatively short intervals of local minima $(n h_n = 30)$. Even though this finite-sample tuning can be sub-optimal for the larger noise levels, it estimates the volatility path quite robustly. As demonstrated in Sec.\ 5.1 of \cite{bibinger2022}, the function $\Psi_n$ in \eqref{spotclt} needs to be taken into account as the first-order approximation \eqref{psiapprox} generates a non-negligible finite-sample bias for this tuning. To account for this bias, the spot volatility estimator is multiplied with the correction factor $0{.}954$ as suggested in Sec.\ 5.1 of \citet{bibinger2022}.

\begin{figure}[t]
	\begin{minipage}{1\textwidth}
     		\centering
        		\includegraphics[width=6.5cm]{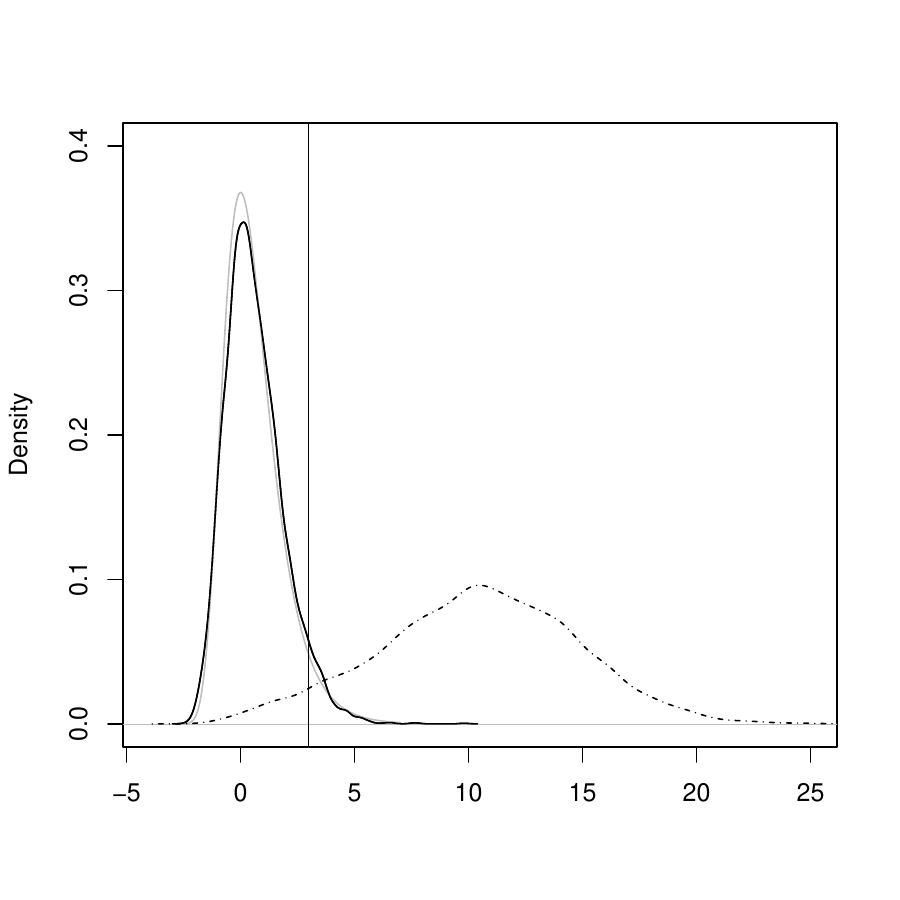}
		\includegraphics[width=6.5cm]{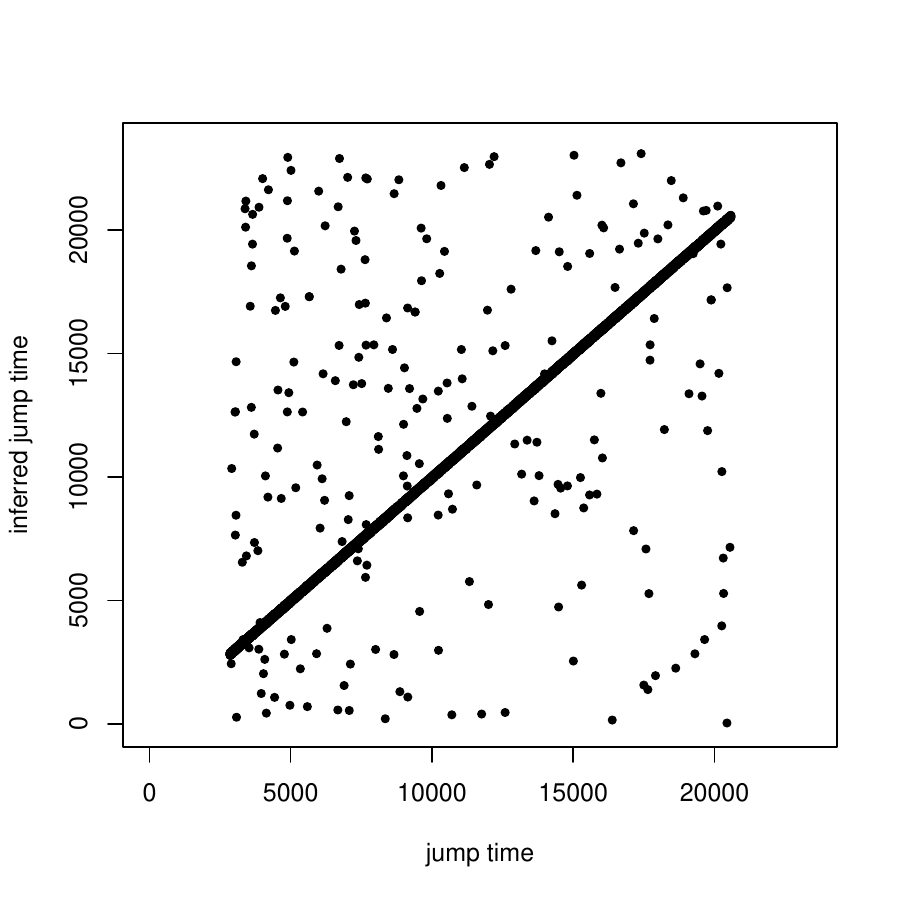}
    	\end{minipage}
		\caption{\label{Fig:density}Left panel: Kernel density estimates of the standardized version of the test statistic $T^{BHR}$ under the null (black solid line) and under the alternative (dashed-dotted line) for $\vert \text{jump size}\vert = 0.3\%$ and noise level $q = 0.1\%$. The gray solid line depicts the density function of the standard Gumbel distribution. The black and gray vertical lines (almost indistinguishable) are the $95\%$-quantile of the standardized version of the test statistic $T^{BHR}$ under the null and the $95\%$-quantile of the standard Gumbel distribution. Right panel: Jump times vs.\ inferred jump times for $\vert \text{jump size}\vert = 0.3\%$ and noise level $q = 0.1\%$.}
\end{figure}

Table~\ref{tab:most_global} documents the size ($\vert \text{jump size}\vert=0\%$) and power ($\vert \text{jump size}\vert>0\%$ ) of the global test for different noise levels $q$. We observe that the test keeps its level. The power of the test decreases slightly for rising noise level and a given jump size. Here, jumps of small size are more likely to be detected in case of low noise, while jumps of larger size can be detected for all noise levels considered. For $\vert \text{jump size}\vert = 0.3\%$ and $q = 0.1\%$, the performance of the test is illustrated in Figure~\ref{Fig:density} which shows kernel density estimates of the standardized version of the test statistic $T^{BHR}$ under the null and under the alternative. For comparison, in Figure~\ref{Fig:density}, we also show the density function of the standard Gumbel distribution. Even though there are minor deviations in the center of the distribution of the standardized test statistics compared to the standard Gumbel distribution, the right tail of the distribution of the test statistics under the null is quite accurately approximated. In the right panel of Figure~\ref{Fig:density} we plot the true jump times against the inferred jump times associated with the interval maximizing the test statistic $T^{BHR}$ for $\vert\text{jump size}\vert = 0.3\%$. The times are thus determined by the estimator from Proposition \ref{localize}. As most of the points are on the main diagonal, we conclude that the jump times are correctly estimated in most cases. As more points are off the diagonal in the morning than around lunch time, we further conclude that the test performs slightly worse in intraday periods, where the volatility tends to be higher. This is expected, since higher volatility makes jump detection more difficult.

For comparison purposes, the size and power of the global test from Theorem~\ref{thmgumbel} are documented under i.i.d.\ exponentially distributed noise and equispaced observations in the \ref{S1}.

\subsection[A comparison of jump detection under MMN and LOMN]{A comparison of jump detection under MMN and LOMN\label{sec:3.2}}
We compare jump detection based on our methods under one-sided noise in simulations with jump detection in the standard MMN-model. As an established jump test under MMN we use the one by \citet{lm12}. That both methods use maximum differences of local statistics and asymptotic Gumbel quantiles yields an informative comparison. This is understood to be a comparison of LOMN- vs.\ MMN-model and not of the different tests, which both attain desirable asymptotic properties in the models they are built for. In practice, limit order book data allows to reconstruct different time series as mid quotes, best ask and best bid quotes modelled with MMN and one-sided noise, respectively. We apply the test by \citet{lm12} only to simulated observations with MMN and our test only to observations with LOMN. The local statistics are local averages under MMN, local minima under LOMN and local maxima for upper-bounded, one-sided noise.

In a {\textbf{first scenario}}, we compare the baseline models with i.i.d.\ noise and equispaced observations for equal sample sizes. We generate observations of both models by
\begin{align}
	Y_{i}^{A} &= X_{i/n} + (1 - 2 \pi^{-1})^{-1/2} q \vert\varepsilon_{i}\vert \quad\text{with}\quad\varepsilon_{i}\overset{iid}{\sim}\mathcal{N}(0, 1),\label{eq.sim1}\\
	Z_{i} &= X_{i/n} + q\varepsilon_{i}\label{eq.sim2}
\end{align}
where $Y_{i}^{A}$ are the simulated ask quotes based on the LOMN-model and $Z_{i}$ are the observations under MMN, $i=0,\ldots,n$. Note that $(1 - 2 \pi^{-1})^{-1/2} q \vert\varepsilon_{i}\vert \sim HN(0, q^2)$, such that the noise variances are identical. It can be easily verified that half-normally distributed noise satisfies \eqref{noise_dist}. In both models, we have the same underlying efficient log-prices $(X_{i/n})_{0\le i\le n}$. 

In a {\textbf{second scenario}}, we generate observations of a model that incorporates both additive MMN and rounding errors, following \citet{li2018}, by
\begin{align}
	Z_{i} &= \log \left( \exp \left(X_{i/n} + \epsilon_{i}\right)^{(s)} \right)\quad \text{with}\quad\epsilon_{i}\overset{iid}{\sim}\mathcal{N}(0, q^2),\label{eq.sim1alt}
\end{align}
with $i=0,\ldots,n$ and $x^{(s)}:=s \vee ([x/s]s)$ for any $x>0$ (i.e., $x$ is rounded to the nearest multiples of tick size $s$ which is set to $s=0.01$ here). Whether the MMN or the rounding error is dominating depends on the price level. In case of a low price level (say $Z_i \leq \log 10$) the rounding error dominates, while in case of high prices (say $Z_i \geq \log 50$) the MMN dominates. As in Section \ref{sec:sim_global}, samples of ask and bid quotes are generated implicitly by using $Y_{i}^{A}= Z_{i}$, whenever $Z_{i} > X_{i/n}$, and $Y_{i}^{B}= Z_{i}$, whenever $ Z_{i} < X_{i/n}$. Such a model with misspecification and endogenous noise is important and might be realistic in view of the discrete fine structure of the data shown in Figure \ref{Fig:lob}.

For the additive MMN-model, the global test for price jumps uses the maximum statistic
\begin{align}\label{lm}T^{LM}:=\max_{k=2,\ldots,h_n^{-1} - 1}\left|\frac{(n h_n)^{-1} \sum_{i = n h_n k}^{n h_n (k + 1) - 1} Z_{i} - (n h_n)^{-1} \sum_{i = n h_n (k-1)}^{n h_n k - 1} Z_{i}}{\sqrt{\frac{2}{3}\sigma_{kh_n}^2C^2 + 2 q^2}}\right|\,,\end{align}
based on differences between local averages, where we assume that $nh_n$ is an integer with $h_n = C n^{-1/2}$, and $C$ is a constant that is documented for different noise levels in \citet{lm12}. The statistic $T^{LM}$ builds on differences of local averages (while $T^{BHR}$ uses local minima) and the asymptotic standard deviation of these differences of local averages (the denominator) depends on the variance of the noise $q^2$. \citet{lm12} show that $T^{LM}$ converges after appropriate standardization to a standard Gumbel distribution that permits testing hypotheses similarly to Theorem~\ref{thmgumbel}.

\begin{table}[t]
\centering
	\begin{tabular}{ccc|cccccc}
		\hline\hline
			&&& \multicolumn{6}{c}{$\vert \text{jump size}\vert$} \\
		\hline
		$q$ & $nh_n$ & test & 0.000\% &  $0.100\%$ & $0.125\%$ & $0.150\%$ & $0.175\%$ & $0.200\%$\\
		\hline
		\multirow{4}{*}{0.05\%}		& \multirow{2}{*}{11} 	& BHR &0.05 & 0.31 & 0.63 & 0.85 & 0.94 & 0.98 \\
							&  				& LM   &0.04 & 0.14 & 0.29 & 0.50 & 0.69 & 0.83 \\
							&\multirow{2}{*}{15}	& BHR &0.05 & 0.29 & 0.59 & 0.83 & 0.94 & 0.98 \\
							& 				& LM   &0.05 & 0.16 & 0.33 & 0.54 & 0.72 & 0.85 \\	
\hline
		\multirow{4}{*}{0.10\%}		& \multirow{2}{*}{20} 	& BHR &0.05 & 0.14 & 0.32 & 0.57 & 0.77 & 0.89 \\
							& 				& LM   &0.05 & 0.07 & 0.10 & 0.19 & 0.32 & 0.47 \\
							& \multirow{2}{*}{34} 	& BHR &0.05 & 0.12 & 0.23 & 0.43 & 0.64 & 0.80 \\
							&  				& LM   &0.05 & 0.08 & 0.14 & 0.24 & 0.37 & 0.52 \\
		\hline\hline
	\end{tabular}
	\caption{Simulation results for the global tests in Scenario 1 on a significance level of $\alpha=5\%$. The column $\vert \text{jump size}\vert=0.000\%$ gives the estimated size and the columns with $\vert \text{jump size}\vert>0\%$ give the estimated power for the corresponding jump sizes. The value $nh_n$ is the number of noisy observations per interval.}
	\label{tab:global_comparison}
\end{table}

It should be noted that the improved power under LOMN is due to the (asymptotically) smaller blocks of order close to $n^{-2/3}$, instead of $n^{-1/2}$ under MMN. When choosing a small constant factor $C$ to determine the block lengths for the MMN-model and a much larger proportionality constant under LOMN, the test can -- in finite samples -- perform better in terms of power than the test for the LOMN-model. This finite-sample phenomenon is in contrast to the asymptotic considerations but can arise in situations when the specific tuning of $h_n$ results in shorter intervals in the MMN-model than in the LOMN-model. In particular, the sample size is in Scenario 2 smaller for LOMN than for MMN. We compare results for different block lengths including values which approximately keep the size and optimize the power of the methods based on a grid search.

To produce comparable results, we employ a simple bootstrap. Following Proposition~2 of \citet{lm12}, an estimator $\hat{q}_n$ of the noise level $q$ is given by
\begin{align}\label{noise_est}
	\hat{q}_n=\Big((2n)^{-1}\sum_{i = 1}^{n} (Z_{i} - Z_{{i-1}})^2\Big)^{1/2},
	\end{align}
that is consistent under MMN in Scenario 1, but might become inconsistent in Scenario 2 including rounding errors. Based on this estimate, we generate $m = 5{,}000$ bootstrap samples 
\begin{align}
	Z_{i, j}^* &= X_{i/n, j}^* + \epsilon^*_{i, j} \quad\text{with}\quad\epsilon_{i, j}^*\overset{iid}{\sim}\mathcal{N}(0, \hat{q}_n^{2}),
\end{align}
for $j = 1,\ldots,m$. In Scenario 2, simulated ask $Y_{i,j}^{A,*}$ and bid quotes $Y_{i,j}^{B,*}$ are generated implicitly again based on $Z_{i, j}^*$. In Scenario 1, we generate additionally bootstrap samples
\begin{align}\label{adboot}
	Y_{i, j}^{A,*} &= X_{i/n, j}^* + (1 - 2 \pi^{-1})^{-1/2} \hat{q}_n \vert\varepsilon_{i, j}^*\vert \quad\text{with}\quad\varepsilon_{i, j}^*\overset{iid}{\sim}\mathcal{N}(0, 1),
\end{align}
with $\hat{q}_n$ as in \eqref{noise_est} using the data $Y_i^{A}$ instead of $Z_i$, as an elementary proof shows that the same noise level estimator is valid. Critical values under the null hypothesis are determined based on the empirical quantiles of $\{T^{BHR^*}_j\}_{j=1}^m$ and $\{T^{LM^*}_j\}_{j=1}^m$. In Scenario 2 the statistics $\{T^{BHR^*}_j\}_{j=1}^m$ based on observations $Y_{i,j}^{A,*}$ are evaluated as well as analogous variants using local maxima of the observations $Y_{i,j}^{B,*}$.

Similar Monte Carlo or wild bootstrap procedures are used and known to perform well for high-frequency statistics, in particular for extreme value statistics, see, for instance, \cite{boot} and \cite{dovonon2019bootstrapping}. Only in case of dominating rounding errors in Scenario 2, the performance of the procedure turns out to be sub-optimal. In practice, this bootstrap relies on an estimate of the stochastic volatility process $v_t \sigma_t$, $t\in[0,1]$, in order to generate $(X_{i/n, j}^*)$ based on an Euler-Maruyama scheme where we set the drift to zero and with the Gaussian increments of a Brownian motion. For this comparison, however, we use the simulated \emph{true} volatility. Otherwise we would require different spot volatility estimators for the two different models, which would complicate the comparison of the combined methods. Using the true volatility for both methods sheds light on their different capabilities to detect jumps.

\begin{table}[t]
\centering
	\begin{tabular}{cccc|cccccc}
		\hline\hline
			&& \multicolumn{6}{c}{$\vert \text{jump size}\vert$} \\
		\hline
		$q$ & $X_0$ & $nh_n$ & data & 0.000\% &  $0.100\%$ & $0.125\%$ & $0.150\%$ & $0.175\%$ & $0.200\%$\\
		\hline
		\multirow{6}{*}{$0.05\%$} & \multirow{6}{*}{$\log 10$}	& \multirow{3}{*}{10} 								& $Y_i^A$	&0.05 & 0.12 & 0.26 & 0.45 & 0.66 & 0.84\\
							&& 				& $Y_i^B$ 	&0.04 & 0.12 & 0.25 & 0.46 & 0.66 & 0.85\\
							&& 				& $Z_i$   	&0.06 & 0.11 & 0.20 & 0.38 & 0.55 & 0.74\\
							&& \multirow{3}{*}{11} 	& $Y_i^A$ 	&0.07 & 0.16 & 0.30 & 0.47 & 0.67 & 0.83\\
							&&				& $Y_i^B$ 	&0.07 & 0.16 & 0.28 & 0.46 & 0.68 & 0.84\\
							&&  				& $Z_i$   	&0.05 & 0.11 & 0.22 & 0.42 & 0.60 & 0.76\\
		\hline
		\multirow{6}{*}{$0.05\%$} & \multirow{6}{*}{$\log 50$}	& \multirow{3}{*}{5} 								& $Y_i^A$ &0.07 & 0.27 & 0.59 & 0.82 & 0.94 & 0.98\\
		 					&& 				& $Y_i^B$ &0.05 & 0.26 & 0.56 & 0.82 & 0.94 & 0.98\\
							&& 				& $Z_i$  &0.04       & 0.06 & 0.11 & 0.23 & 0.38 & 0.55\\
							&& \multirow{3}{*}{12} 	& $Y_i^A$ &0.06 & 0.17 & 0.35 & 0.61 & 0.80 & 0.92\\
							&&				& $Y_i^B$ &0.05 & 0.16 & 0.35 & 0.63 & 0.83 & 0.93\\
							&&  				& $Z_i$   &0.04      & 0.14 & 0.32 & 0.53 & 0.72 & 0.85\\
		\hline\hline
	\end{tabular}
	\caption{Simulation results for the global tests in Scenario 2 on a significance level of $\alpha=5\%$ and with $q=0.05\%$. The column $\vert \text{jump size}\vert=0.000\%$ gives the estimated size and the columns with $\vert \text{jump size}\vert>0\%$ give the estimated power for the corresponding jump sizes. The value $nh_n$ is the number of noisy observations per interval.}
	\label{tab:global_comparison2}
\end{table}

Table~\ref{tab:global_comparison} reports the results of the simulations for Scenario 1 based on 5{,}000 replications with noise levels $q=0.05\%$ and $q=0.10\%$. BHR refers to our test based on observations $Y_i^{A}$ with one-sided noise and LM to the test by \cite{lm12} based on observations $Z_i$ with MMN. We consider different numbers of observations per interval and choose $nh_n$ such that the power of the corresponding test for $q$ is maximized over a large grid of values (the results of the respective other test are presented for comparison). We observe that both tests keep the size and correctly reject the null hypothesis with approximate probability $\alpha$. The power of both tests crucially depends on the noise level $q$ and the $\vert\text{jump size}\vert$. In general, the power of both tests increases with the $\vert\text{jump size}\vert$ and decreases for higher noise levels. In comparison, the power of the test in the LOMN-model always outperforms the test in the MMN-model even with same block lengths, although the magnitude depends admittedly on the specific setting. 

Table~\ref{tab:global_comparison2} reports the results of the simulations for Scenario 2 based on 1{,}000 replications with noise level $q=0.05\%$ and different numbers of observations per interval. The size and power is compared for the tests by \cite{lm12}, based on all observations $Z_i$, our test based on simulated ask quotes $Y_i^{A}$, and the analogous one based on local maxima of simulated bid quotes $Y_i^{B}$. We choose $nh_n$ such that the tests approximately keep the size and maximize the power over a grid of values for $nh_n$ (the results of the other tests are presented for comparison). The tests keep the size for suitable values of $nh_n$, which determine the block lengths, and correctly reject the null hypothesis with approximate probability $\alpha$. Despite the smaller effective sample size of roughly $n/2$ for ask and bid quotes, these tests show a slightly better performance in terms of power for $X_0=\log 10$. The LOMN-based test and its variant for bid quotes even show a significantly better performance for $X_0=\log 50$. The power of the tests crucially depends on the impact of the rounding error compared to the noise level $q$, and again on the $\vert\text{jump size}\vert$. The power of the tests increases with the $\vert\text{jump size}\vert$, as well as with the price level, when the effect of the rounding errors decreases. In comparison, the power of the test based on the LOMN-model frequently outperforms the test in the MMN-model despite the effectively reduced sample size. Figure~\ref{Fig:dens_compare} illustrates the difference in the power of the tests for the special case of $\vert \text{jump size}\vert = 0.125\%$, $q = 0.05\%$ and $X_0=\log 50$. Even though both plots use a different scaling, i.e., the non-standardized test statistics are not directly comparable, the better performance of the test in the LOMN-framework is evident. A similar plot for Scenario 1 is provided in the \ref{S1}. In \ref{S1} we also report with a discussion that for rounding errors too large compared to the noise level $q$, the performance of the test by \cite{lm12} deteriorates, such that the information from a comparison for such a setup without adjustments would be limited.

\begin{figure}[t]
	\begin{minipage}{1\textwidth}
     		\centering
        		\includegraphics[width=6.5cm]{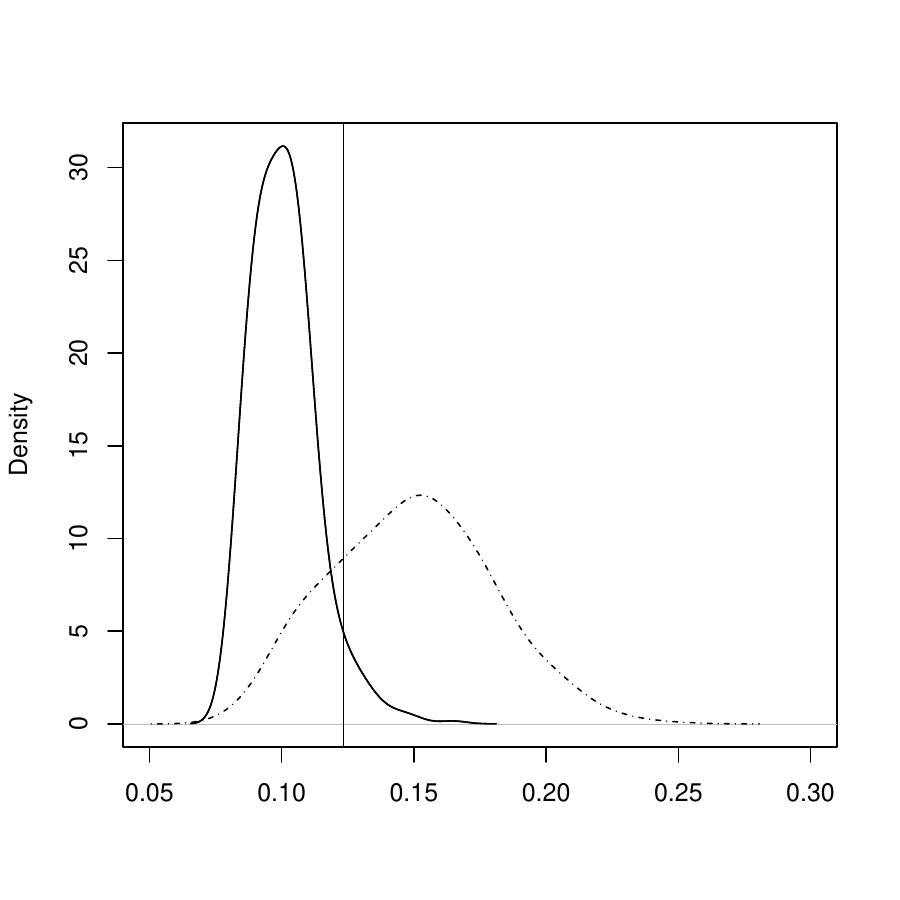}
		\includegraphics[width=6.5cm]{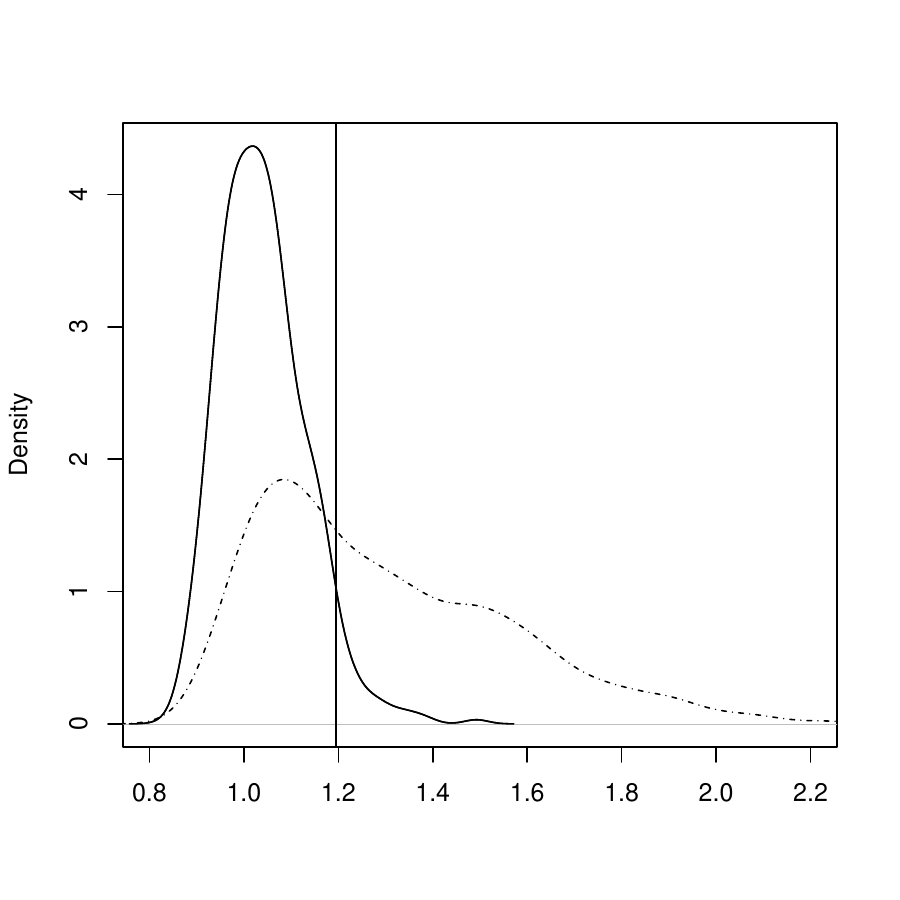}
    	\end{minipage}
	\caption{\label{Fig:dens_compare}Left panel: Kernel density estimates of the test statistic $T^{BHR}$ in Scenario 2 under the null (black solid line) and under the alternative (dashed-dotted line) for $\vert \text{jump size}\vert = 0.125\%$, noise level $q = 0.05\%$ and $X_0 =\log 50$ for the optimal $nh_n=5$. Right panel: Kernel density estimates of the test statistic $T^{LM}$ in Scenario 2 under the null (black solid line) and under the alternative (dashed-dotted line) for $\vert \text{jump size}\vert = 0.125\%$, noise level $q = 0.05\%$ and $X_0 =\log 50$ for the optimal $nh_n=12$. The vertical lines refer to the $95\%$-quantile of the respective test statistic under the null.}
\end{figure}

In summary, our findings suggest that jumps are more likely to be detected under LOMN than under MMN, and that \emph{jumps of smaller size can be detected} in the LOMN-framework. Compared to the results of Section~\ref{sec:sim_global}, we observe that the number of observations per block, $nh_n$, can be chosen smaller in the test setting relying on the bootstrap. These smaller blocks translate into a more precise localization of jumps and the possibility to detect smaller jumps. In the test setting relying on the asymptotic Gumbel distribution of the standardized version of the maximum statistic $T^{BHR}$, the length of the time blocks $h_n$ was chosen larger than $n^{-2/3}$ to be in line with the assumptions of the asymptotic theory. When comparing  the power of the test for the bootstrap-based setting with the setting relying on the asymptotic distribution for specific choices of $q$ and jump sizes (given the same $nh_n$), we did not find noteworthy differences. In this sense, in our experiments, the distribution of the noise has no impact on the power of the test.

\begin{table}[t]
\centering
	\begin{tabular}{ccc|ccccccc}
		\hline\hline
			&&& \multicolumn{7}{c}{$\vert \text{jump size}\vert$} \\
		\hline
		$q$ & $\tau$ & test & 0.000\% &  $0.050\%$ & $0.075\%$ & $0.100\%$ & $0.125\%$ & $0.150\%$ & $0.200\%$\\
		\hline
		\multirow{6}{*}{0.05\%}		& \multirow{2}{*}{bef-} 	& BHR &0.05 & 0.42 & 0.71 & 0.88 & 0.95 & 0.98 & 0.99\\
							&  				& LM   &0.05 & 0.18 & 0.34 & 0.47 & 0.58 & 0.66 & 0.75\\
							&\multirow{2}{*}{at}	& BHR &0.05 & 0.58 & 0.91 & 0.99 & 1.00 & 1.00 & 1.00\\
							& 				& LM   &0.04 & 0.45 & 0.79 & 0.96 & 1.00 & 1.00 & 1.00\\		
							&\multirow{2}{*}{aft+}	& BHR &0.05 & 0.41 & 0.71 & 0.88 & 0.94 & 0.97 & 0.99\\
							& 				& LM   &0.05 & 0.19 & 0.32 & 0.47 & 0.58 & 0.65 & 0.75\\	
\hline
		\multirow{6}{*}{0.1\%}		& \multirow{2}{*}{bef-} 	& BHR &0.06 & 0.26 & 0.49 & 0.70 & 0.84 & 0.91 & 0.96\\
							& 				& LM   &0.05 & 0.12 & 0.20 & 0.31 & 0.42 & 0.51 & 0.64\\
							& \multirow{2}{*}{at} 	& BHR &0.05 & 0.34 & 0.67 & 0.89 & 0.97 & 0.99 & 1.00\\
							&  				& LM   &0.05 & 0.27 & 0.51 & 0.75 & 0.90 & 0.98 & 1.00\\
							&\multirow{2}{*}{aft+}	& BHR &0.06 & 0.26 & 0.49 & 0.70 & 0.84 & 0.91 & 0.96\\
							& 				& LM   &0.05 & 0.12 & 0.20 & 0.32 & 0.42 & 0.51 & 0.64\\
		\hline\hline
	\end{tabular}
	\caption{Simulation results for the local test with a significance level of $\alpha=5\%$. The column $\vert \text{jump size}\vert=0.000\%$ gives the estimated size and the columns with $\vert \text{jump size}\vert>0\%$ give the estimated power for the corresponding jump sizes. The column $\tau$ indicates whether the test time $\tau\in(0,1)$ is either randomly before a negative jump (bef-), at the same time as a positive or negative jump, or randomly after a positive jump (aft+). For $q = 0.05\%$, the number of observations within $h_n$ is $nh_n = 12$ and for $q = 0.1\%$ it is $nh_n = 26$. The distance between the random test time and the jump is at most $(nh_n-1)$ observations.}
	\label{tab:local_comparison}
\end{table}

The origin of the outperformance of jump tests under LOMN compared to MMN in terms of power is due to the faster convergence rate, but also due to the pulverization of jumps by pre-averages in the MMN-framework. To quantify the effect of this pulverization on the power of the local test statistics, we generate again $5{,}000$ Monte Carlo samples of \eqref{eq.sim1} and \eqref{eq.sim2}. Under the alternative, these samples are augmented by analogue jumps as used above. In contrast to the global test, the local test is performed at a pre-specified time, for which we use the following three scenarios: Firstly, the test is performed at the exact time of the jump. Secondly, the test is performed at a random time before the jump but the distance between the jump time and the test time is at most the time between the jump and the $(nh_n-1)$th observation before the jump. Thirdly, the test is performed at a random time after the jump but the distance between the jump time and the test time is at most the time between the jump and the $(nh_n-1)$th observation after the jump. Similarly to the global test, we use a bootstrap-based inference for the local test statistics to obtain comparable results. For $q = 0.05\%$ the number of observations within an interval is $nh_n = 12$, and for $q = 0.1\%$ it is $nh_n = 26$.

Table~\ref{tab:local_comparison} provides the corresponding results. The local test under MMN based on the method by \cite{lm12} is detailed in Sec.\ 3.1.1 of \cite{bibwinkellev}. In analogy to the results of the global test above, both tests keep the size and the power of both tests decreases in the noise level $q$. In the situation when the jump time and the test time coincide (i.e., there is no pulverization of jumps by pre-averages), the power of the test based on the LOMN-model is slightly better for $q=0.05\%$, while for $q=0.10\%$ it is significantly better. In the more realistic scenarios with misspecification, when the jump time and the test time do not exactly coincide, a severe drop in the power can be observed for the test in the MMN-framework compared to the situation when jump time and test time coincide. In the LOMN-framework though, there is only a moderate drop in the power. Due to the specific time shift in the two scenarios with misspecification, this comparable better performance based on local minima of simulated ask quotes materializes if the local test is performed before the jump \emph{and} the jump direction is negative, or the local test is performed after the jump \emph{and} the jump direction is positive. This effect directly relates to the advantage of speed illustrated in Figure \ref{Fig:locmin} and \ref{S2}, which is sensitive to the sign of jumps. Looking at local maxima of best bid quotes, however, this effect reverses, and having both at the same time thus facilitates a good performance in all cases. The connection between the direction of the jump and the test time is not important for the global test in the LOMN-model.

\section[Empirical illustration for JPM stock quotes]{Empirical example for JPM stock quotes\label{sec:4}}
We apply the procedures discussed above to data from actual quotes of JPMorgan Chase \& Co.\ (with symbol JPM). The sample period is from July 2007 to September 2009 covering the most turbulent time of the subprime mortgage crisis, where we expect many large changes in equity prices and corresponding quotes. Thus, the data set is appropriate for a comparative empirical study of jump tests under MMN and LOMN, respectively.

\subsection[Data]{Data}
We use first-level limit order book data of ask and bid quotes at the highest possible frequency from the LOBSTER database\footnote{\url{https://lobsterdata.com/}}, which provides access to reconstructed limit order book data for NASDAQ traded stocks. First-level means that bid and ask quote refer to the \emph{best} bid and \emph{best} ask in the sequel. LOBSTER data has been used in several recent research papers, e.g., \cite{andersen2022local}. The data has at least millisecond precision, which generally permits an analysis at the highest time resolution possible. In a non-negligible proportion of cases, there is no change between subsequent quoted bid or ask prices (and corresponding mid quotes). This is in conflict with both, the LOMN-model and the MMN-model. We therefore select only those observations where the bid or ask quote changes. Moreover, all quotes before 9:30am or after 4:00pm are discarded. We also exclude the data during 9:30am and 9:35am for each trading day in order to avoid peculiarities during the opening period and in view of the reduced reliability of jump detection close to boundaries. No further data cleaning procedures are performed before we apply the statistics.

\begin{table}[t]
\centering
	\begin{tabular}{cl|ccccccc}
		\hline\hline
			\multicolumn{2}{r|}{trading hour} & 09:35-10 & 10-11 & 11-12 & 12-13 & 13-14 & 14-15 & 15-16\\
		\hline
			& $n$ &3081  & 4032   & 2554   & 2027 & 2107  & 2646 & 3412\\ 
	ask		& $100\cdot\hat q$ & 0.0224 & 0.0197 & 0.0195 & 0.0194 & 0.0199 & 0.0197 & 0.0205\\
			& $100\cdot\hat \sigma^2$ & 0.0252 & 0.0287 & 0.0162 & 0.0132 & 0.0138 & 0.0191 & 0.0282\\
		\hline
			& $n$ & 3083  & 4060  & 2561    & 2015 & 2112 & 2649 & 3375\\
	bid 		& $100\cdot\hat q$ &0.0222 & 0.0199 & 0.0196 & 0.0195 & 0.0201 & 0.0196 & 0.0204\\
			& $100\cdot\hat \sigma^2$ & 0.0253 & 0.0286 & 0.0164 & 0.0131 & 0.0139 & 0.0188 & 0.0276\\
		\hline
			& $n$ &6162  & 8092  & 5115    & 4042 & 4218  & 5295 &  6787\\ 
	mid 		& $100\cdot\hat q$ &	0.0112 & 0.0099 & 0.0098 & 0.0097 & 0.0100 & 0.0098 & 0.0102\\
			& $100\cdot\hat \sigma^2$ & 0.0244 & 0.0290 & 0.0167 & 0.0130 &  0.0141 & 0.0193 & 0.0257\\
		\hline\hline
	\end{tabular}
	\caption{Averages of the sample size $n$ of ask, bid and mid quotes, the estimated noise level $\hat q$ and the estimated variance $\hat \sigma^2$ for the in total 3,974 time intervals.}
	\label{tab:summary}
\end{table}

As in \cite{lm12} each considered trading day is split into seven time intervals. Table~\ref{tab:summary} reports the averages of the corresponding sample size $n$ per interval, the estimated noise level $\hat q$ and the estimated variance $\hat \sigma^2$, which is approximated constant within each out of the 3{,}974 intervals considered in total. All quantities ($n$, $\hat q$ and $\hat \sigma^2$) exhibit the expected U-shaped intra-daily seasonal pattern across ask, bid and mid quotes. In line with current research we calibrate the MMN-model to mid quotes having available data from a limit order book. Although \cite{lm12} used trade prices for their data analysis, an analysis based on mid quotes is beneficial given the much larger sample size of available mid quotes compared to trades. A comparison with the MMN-model calibrated to trade prices reconstructed from the limit order book could yield different results. Such a comparison of the LOMN and MMN models would be less informative, however, favouring our approach by providing more information. Using mid, ask and bid quotes instead, we exploit the information from the limit order book efficiently for all models. The noise levels are estimated using \eqref{noise_est}. For the LOMN-model, we employ the same volatility estimator (with same tuning) as in our simulations, while we use the spectral approach from \cite{BHMR} and \cite{BHMR2} for volatility estimation in the MMN-model. The estimated volatilities are rather similar across ask, bid and mid quotes, which is coherent with the idea of the same underlying efficient log-prices in the models, and should hence not affect the comparison between the jump tests considerably. While differences in the sample sizes between ask and bid quotes are not worth mentioning, there are major differences in the sample sizes $n$ between ask and mid quotes, as well as bid and mid quotes, respectively. Since by construction, changes in ask or bid quotes imply changes of mid quotes, the sample size of the mid quotes is (approximately) the sum of the sample sizes of ask and bid quotes. Combining bid and ask quotes would hence result in equal sample sizes for the LOMN- and MMN-model. A further difference between ask and mid quotes, as well as bid and mid quotes, is the estimated noise level $\hat q$, which is just half size for mid quotes compared to ask and bid quotes. This is due to the rather mechanical effect that the tick size for mid quotes is just half the tick size of ask or bid quotes. In other words, if the best ask or bid changes by one tick (e.g., 0.01\$), the mid quote changes by just half a tick (e.g., 0.005\$). As smaller noise levels improve the power of the tests in our simulations, the lower noise level should favor an analysis based on mid quotes with the MMN model.

\subsection{Summary of empirical results}
We perform the global test once for each time interval such that the total number of intervals is equivalent to the number of performed tests. This is similar to \cite{lm12}.  
In contrast to the asymptotic theory, the finite-sample time blocks are chosen on average a bit smaller for the computation of local averages of mid quotes than for taking local minima of ask quotes or maxima of bid quotes, respectively. In light of the data characteristics shown in Table~\ref{tab:summary} and in line with the small proportionality constant for the MMN-model suggested in \cite{lm12}, i.e., $nh_n = \frac{1}{19} \sqrt{n}$ for $q = 0.01\%$, this is not that surprising, however. Thus, we fix block lengths in the same way as in our simulations. In most of the 3{,}974 cases, there are three observations $nh_n$ within one time block for mid quotes and four observations for ask and bid quotes. 

Table~\ref{tab:results} presents the testing outcomes in terms of relative rejection frequencies of the null hypothesis of the global test (no jump) on a significance level of 5\%. We observe that more jumps are detected after opening and before closing. Since more jumps are detected based on ask or bid quotes than based on mid quotes in total, one expects that particularly small-sized jumps are not (easily) detected based on mid quotes. However, estimated sizes of detected jumps are on average slightly smaller for mid quotes compared to ask and bid quotes. While this seems to be surprising (in view of the asymptotic theory), this results from the finite-sample comparison due to the finer approximation of the non-observable $X_t$ based on mid quotes for which time blocks are smaller. Note that also the effect of pulverization of jumps by pre-averages as illustrated in Section \ref{sec:2.1} contributes to this effect. This is for example the case, if a large-sized jump is split in two adjacent jumps of smaller sizes by local averages as illustrated in Figure \ref{Fig:locmin}. 

\begin{table}[t]
\centering
		\begin{tabular}{cc|ccccccc|c}
		\hline\hline
			& trading hour & 09:35-10 & 10-11 & 11-12 & 12-13 & 13-14 & 14-15 & 15-16 & total\\
		\hline
			&ask &13.33 & 4.74 & 5.79 & 6.49 & 6.01 & 7.09 & 10.11 & 7.65 \\ 
		freq.  &bid & 13.68 & 5.44 & 4.56 & 7.02 & 8.30 & 9.04 & 9.04 & 8.15 \\
			&mid & 18.07 & 5.79 & 2.81 & 2.46 & 2.12 & 4.96 & 11.70 & 6.84 \\ 

		\hline
		 		\multicolumn{2}{l|}{time advantage}  & 0.30 & 0.06 & 0.39 & 0.00 & 0.15 & 0.01 & 1.61 & 0.32 \\
		\hline\hline
	\end{tabular}
	\caption{Rejection frequencies of the 3{,}974 performed global tests in $\%$. The bottom row shows medians of the advantage of speed in seconds, i.e.\ the times jumps are earlier detected under one-sided noise compared to MMN. The column total in this last row refers to the overall median.}
	\label{tab:results}
\end{table}

There are, moreover, interesting differences in the results stemming from both methods. During 09:35am-10am and 15am-16am, more jumps are detected by the global test for the MMN-model, while the tests for the LOMN-models reject the null hypothesis more frequently during the day. These differences might put in question our paradigm assuming the same underlying efficient price with the same jumps for the three different time series. Results being fully coherent with this idea should rather yield the same detected jumps of larger absolute sizes and some smaller jumps based on either LOMN- or MMN-data. In many of the incoherent cases though, the detected jumps are large absolute log-returns, which do not fully reflect the stylized picture of large directional jumps. Thus, categorizing these cases into small-sized jumps or false alarms is challenging, which is also true for differences between jumps inferred from bid and ask quotes. We shed light on these examples in the next subsection. 

Overall, the null hypothesis is more frequently rejected by the tests based on LOMN-data compared to MMN-data. One possible way of combining ask and bid quotes is to reject the null hypothesis of the global test (no jump), when at least one of the tests based on local minima of ask quotes, or local maxima of bid quotes, rejects. This results in a rejection frequency which is considerably larger than the rejection frequency based on mid quotes.  
Considering rejections based on bid quotes only yields 8.15\%, and on ask quotes only 7.65\%, such that these values are already larger compared to mid quotes (6.84\%). The results thus support our theoretical finding that {\emph{more jumps can be detected in the LOMN-model}}. In this finite-sample data example, this is not only due to different convergence rates, but also due to the improved robustness of local order statistics compared to local averages, including the pulverization of jumps by pre-averages.

Still, the presented results are not fully coherent in the sense that the tests for the LOMN-model do not always reject the null hypothesis when the test for the MMN-model does. In fact, in 2.29\% of the time intervals we detect jumps only based on mid quotes. Many of these events that systematically induce some incoherence, however, are due to specific bounce-back movements of prices which we discuss in the following subsection. Recall that we did not filter or remove certain ``putative'' jump events, as sometimes suggested in the literature, see, e.g., Sec.\ 6.1 of \cite{jumpfactor}.

The bottom row of Table~\ref{tab:results} presents the median time advantage of online tests for jumps which are coherently detected based on all three considered time series (mid, ask and bid quotes). The time shows how many seconds earlier these jumps are detected under one-sided noise by local (running) minima and maxima of ask and bid quotes, compared to MMN with local averages of mid quotes. Despite the larger sample size of mid quotes compared to ask quotes only and bid quotes only, there is a median advantage in terms of speed of approx.\ 0.32 seconds of the online test from Corollary~\ref{onlinecor} for one-sided noise. Even though this might seem short in absolute time, it is an evident result against the background of the short blocks, which are here on average shorter than 2 seconds. This underpins, moreover, the economic relevance of this aspect, since high-frequency trading algorithms are able to exploit speed advantages in news arrivals which are (only) fractions of a second. Recall that the speed advantage under one-sided noise stems from the monotonicity of maxima and minima with respect to subset relations, whereas no such property holds for averages. As illustrated in \ref{S2}, the positions of blocks relative to jump times influence the severity of the pulverization effect, yet the speed advantage is preserved in all cases. We therefore conclude that there are no adjustments in this direction -- such as rolling blocks -- that can improve detection speed under MMN.
 
\subsection{A closer look at examples\label{sec:4.3}}
There are several situations when all of the three tests detect a jump in the same time interval, with the corresponding statistics identifying almost the same time point. Two typical examples illustrating how stylized large-sized jumps look like are presented in Figure~\ref{Fig:coh}, where the left panel shows a negative jump and the right panel shows a positive jump in the log mid quote. The gray areas in Figure~\ref{Fig:coh} provide a rough orientation of the time of the jump. Interestingly, both jumps are followed by an immediate increase in volatility suggesting that the efficient log-price process and the volatility process jump simultaneously. Such examples of instantaneous, large price adjustments, which are clearly in line with the notion of a jump of $(X_t)$, are coherently found by all considered methods.

\begin{figure}[t]
	\begin{minipage}{1\textwidth}
     		\centering
        		\includegraphics[width=6.5cm]{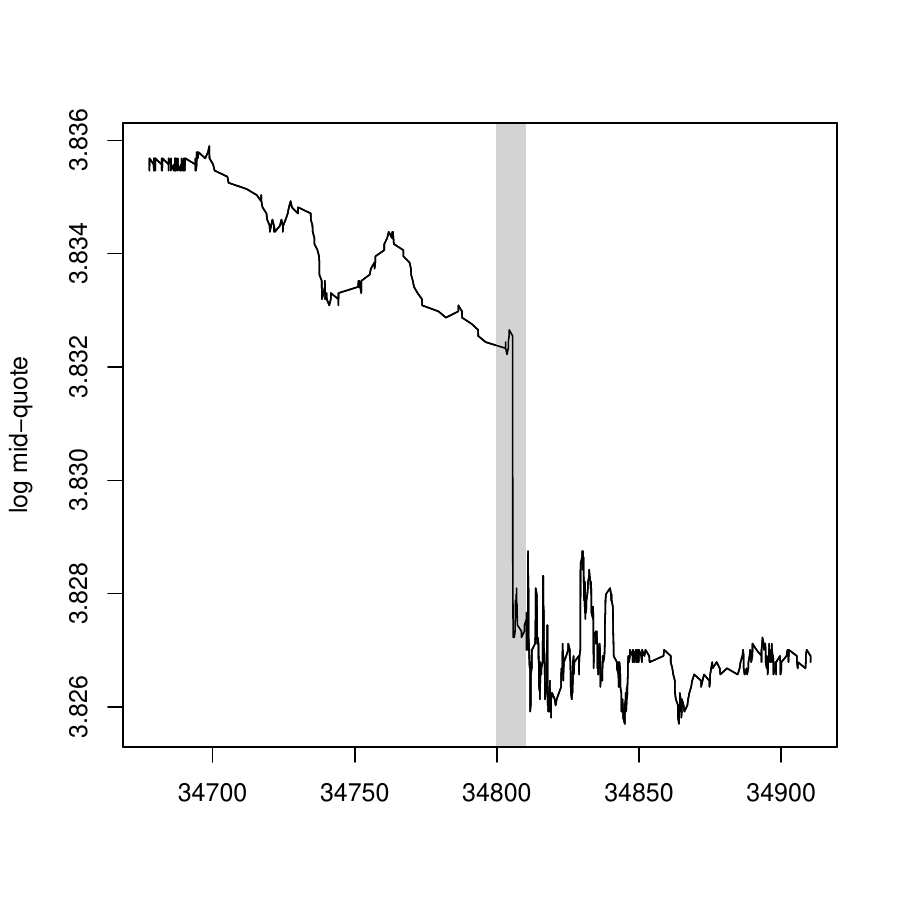}
		\includegraphics[width=6.5cm]{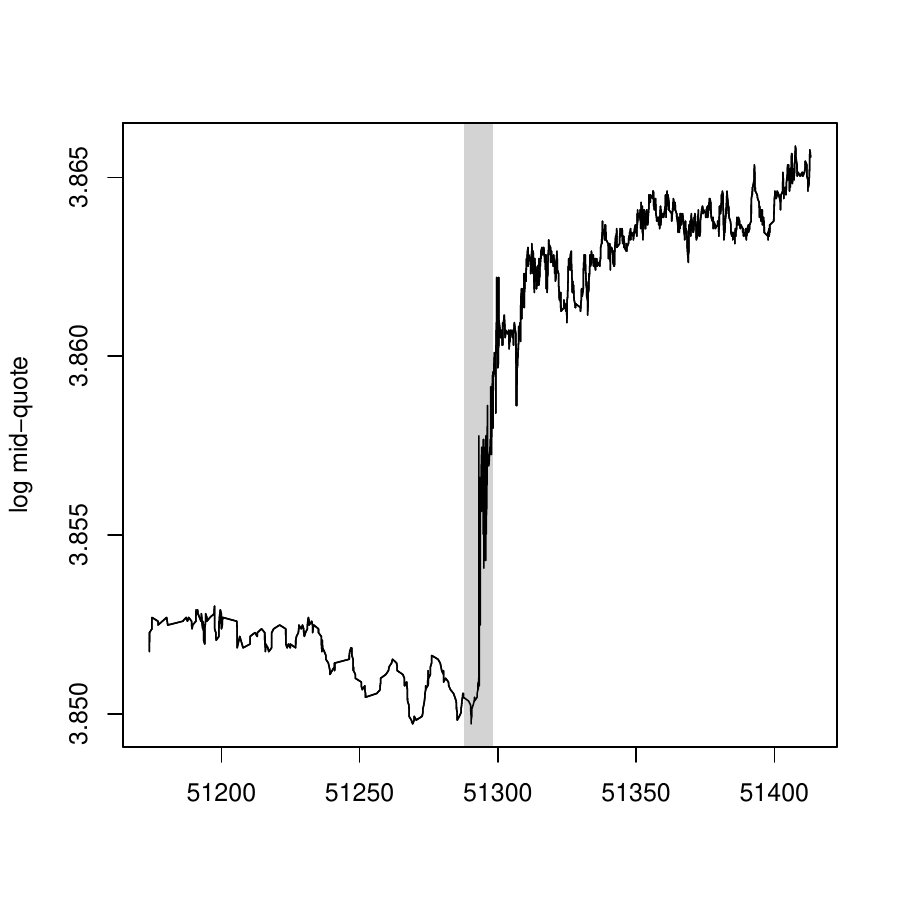}
    	\end{minipage}
	\caption{\label{Fig:coh} Log mid quotes over time provided in seconds after midnight. Areas highlighted in gray contain the detected jumps. Left: negative jump on 24th July 2007 (34800 is 9:40am) detected by all three tests. Right: positive jump on 30th Jan 2008 (51300 is 2:15pm) detected by all three tests.}
\end{figure}

\begin{figure}[t]
	\begin{minipage}{1\textwidth}
     		\centering
        		\includegraphics[width=6.5cm]{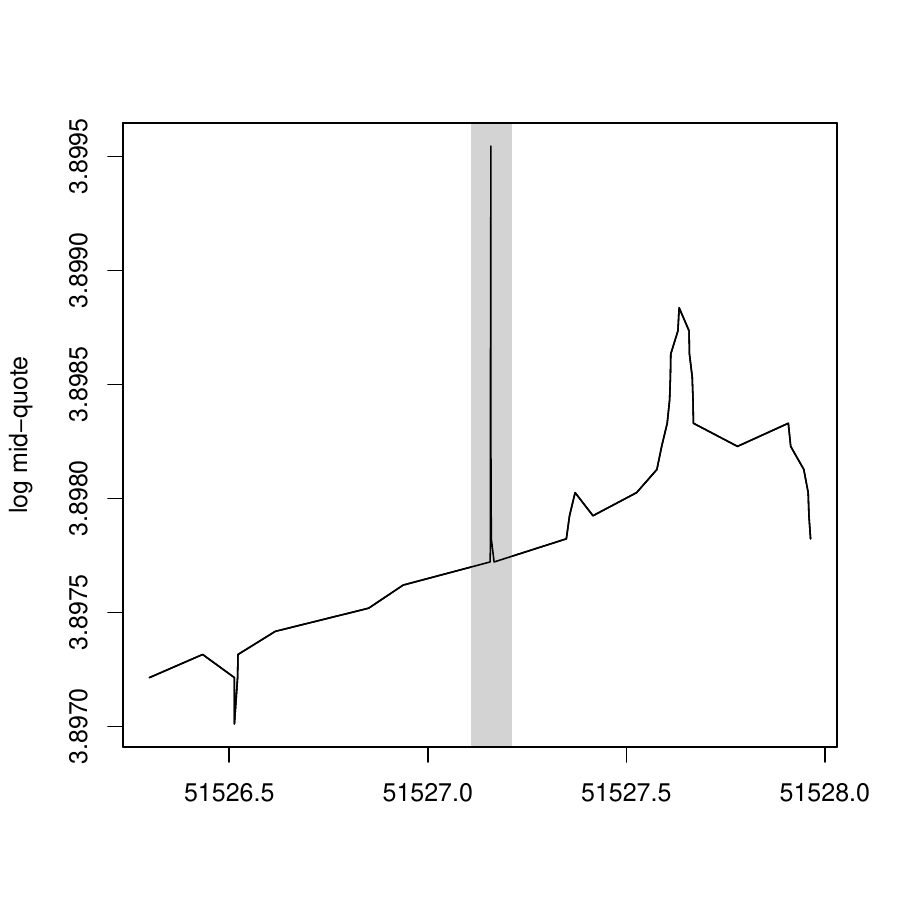}
		\includegraphics[width=6.5cm]{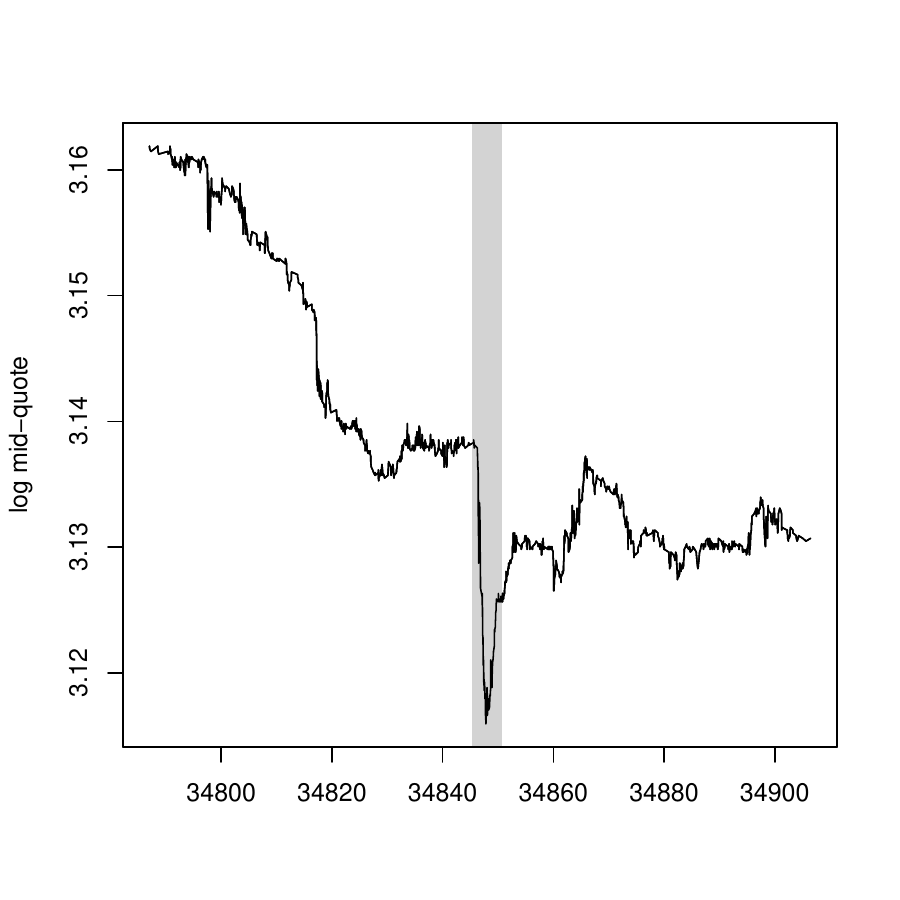}
    	\end{minipage}
	\caption{\label{Fig:incoh} Log mid quotes over time provided in seconds after midnight. Areas highlighted in gray contain the detected jumps. Left: a situation on 28th June 2007 (51527 is 14:18:47pm) when the test for the MMN-model detects a jump, while the test for the LOMN-model does not detect a jump. Right: a situation on 24th Nov 2008 (34850 is 9:40:50am) when the tests for one-sided noise detect a jump, which the test for the MMN-model does not detect.}
\end{figure}

We analyze cases with incoherent test decisions next. These are cases where the LOMN-methods indicate no rejection of the null hypothesis, i.e., they do not detect jumps, while the test based on the MMN-model points in the opposite direction. This could be seen as incoherent with the idea of the same underlying efficient price. However, we find that this discrepancy is frequently explained by rapid bounce-back movements of the observable price of a few tick sizes magnitude. One extreme example is illustrated in Figure~\ref{Fig:incoh} (left panel), where the bounce-back movement of the price occurs within only the hundredth of a second. This mid quote movement originated from the almost simultaneous cancellation of 24 ask orders comprising 5900 shares in total leading to a spontaneous increase of the best ask quote from 49.39\$ to 49.48\$, followed by an immediate bounce back to 49.35\$. It is intuitive that local averages of mid quotes detect a jump. Conversely, local minima of ask quotes are more robust to isolated outliers than local averages and thus do not notify a jump in this example. At the same time, however, similar events occasionally further result in differences between tests based on local maxima of bid quotes and local minima of ask quotes. Moreover, for this example, the null hypothesis is also rejected on a smaller significance level of 1\% for the considered small number of mid quotes per block $nh_n=3$, while no jump is identified as soon as the length of time blocks is increased. In our data, there are several patterns similar to this example which might have been discarded by applying a data cleaning procedure. However, there is no doubt that the cancellations of the ask orders are real and thus, they should not be treated as recording errors that have to be removed from the data. Since the long-term price movement is not affected by this bounce-back effect within a very short time period, we think that this microstructural effect should not be regarded as a jump of $(X_t)$. Identifying a jump in such an event could be prevented by combining bid and ask quotes in a way that we only reject the null if \emph{both} tests reject. However, the question whether such price dynamics should be modeled as jumps, or should be categorized in an alternative way, might depend on the application and concrete research question. This should be thought of with the statistical model and inference in the context of one another.

Finally, the right panel of Figure~\ref{Fig:incoh} highlights an interesting incoherent test decision. This example addresses the sensitivity of the considered global jump tests with respect to the lengths of time blocks. Here, both tests based on one-sided noise models detect a jump, which, however, the MMN-model cannot identify. This is true for the automatically chosen number of observations per time block $nh_n = 8$ for ask quotes, $nh_n = 10$ for bid quotes and $nh_n = 7$ for mid quotes. However, in case of different choices of the number of mid quotes per time block, i.e., $nh_n \in \{5, 6, 9,\ldots, 16, 18, 19, 20\}$, the picture reverses and the MMN test also identifies a jump. Hence, in this example, the tests based on one-sided noise and bid and ask quotes prove to be more robust with respect to changes of time blocks.

\subsection{Insights from the empirical analysis}
Overall, we find that the empirical results mostly support our idea to model mid quotes by the sum of an efficient log-price and MMN, while single ask and bid quotes are modeled by the sum of the same efficient log-price and one-sided noise. We find, however, that some subtle microstructural effects present in the data can result in minor inconsistencies with this idea. The coherence of results can surely be increased by reducing the significance level or by increasing the lengths of the blocks, but the question how to handle these examples should be relevant for all jump tests based on ultra high-frequency data.

The empirical analysis shows that despite a much faster convergence rate for the detection of jumps under LOMN, the performance based on mid quotes and the MMN-model is competitive in finite samples due to the large number of available mid quotes and the tuning with very small time blocks. Nevertheless, we point out several advantages of using our new statistical methods with local minima of ask and local maxima of bid quotes. Our new jump test turns out be more sensitive in practice than the classical one based on mid quotes. Moreover, local order statistics avoid the pulverization effect that manipulates jump detection based on local averages of mid quotes. They are shown to be more robust to bounce-back effects and varying block sizes. In any case, using LOMN-based jump tests additionally to classical jump tests for mid quotes provides more information and a more complete picture of intra-day price jumps. Online jump detection is demonstrated to be considerably faster based on our new approach, what can be useful for high-frequency trading algorithms.

\section{Conclusion\label{sec:5}}
In line with the literature on high-frequency statistics, this work develops methods to infer jumps in a semimartingale efficient price with locally bounded drift and volatility. The methods answer the question if a \emph{jump or no jump} occurs on some time interval under investigation. The main insight of this work is that under one-sided noise (LOMN), smaller jumps in the efficient price can be detected compared to regular market microstructure noise (MMN). We develop methods to infer jumps of absolute size larger than order $n^{-1/3}$, based on $n$ observed best ask quotes. For a fixed jump size, we attain higher power than tests based on observations with MMN. Moreover, the intricate effect of pulverization of jumps by pre-averages vanishes using block-wise local minima and the online jump detection has an advantage of speed. Even for uniformly consistent spot volatility estimation, we do not require existence of moments of the noise distribution. 

Our finite-sample analysis demonstrates that both models (MMN and LOMN) and corresponding inference methods can be applied to different time series from the same limit order book data. Modeling mid quotes with MMN, best ask quotes with one-sided, lower-bounded noise and best bid quotes with one-sided, upper-bounded noise yields overall coherent empirical results. 
We find more jumps in the data study based on local minima of ask and local maxima of bid quotes. The empirical analysis demonstrates the advantage of speed in online jump detection of our approach and that jump detection is not impaired by a pulverization effect (as it is under MMN). We highlight, moreover, some sound finite-sample properties of local maxima of bid and local minima of ask quotes compared to local averages of mid quotes, while it nevertheless appears difficult to conclude that one test approach always outperforms the other. 

Occasional incoherent test decisions based on the different time series are shown to be driven by stylized facts which question the standard semimartingale testing problem. Real high-frequency data from a limit order book shows, beyond instantaneous price jumps in one direction, moreover, occasionally more complex patterns with ``jump-like events''. These include ``bounce-back movements'', also described as ``fast mean-reverting spikes'' (\cite{deschatre2020}), outliers, spoofing effects, (mini) flash crashes and ``gradual jumps'', recently modeled by drift bursts, see \cite{driftburst} and \cite{driftburst2}. Differences in the dynamics of local maxima of best bid quotes, and local minima of best ask quotes, and local averages of mid quotes, on small time blocks related to these events as illustrated in Section \ref{sec:4.3} are in line with the discussed models. They are potentially useful to better identify such events than based on the classical MMN approach only, and to categorize ``jump-like events'' when the question \emph{jump or no jump} is broadened or followed-up by \emph{what type of jump-like event} occurs on some interval. This potential will be explored in future research, e.g., investigating how the methods interact with possible drift bursts. 
 
Extending the theory to more general noise distributions in future research is certainly of interest. Instead of \eqref{noise_dist} we could allow for a general extreme value index at the minimum and develop methods to estimate it for different assets. In particular, we conjecture that this extreme value index influences convergence rates and the minimum jump sizes which can be detected. It might be interesting to investigate if the estimated index is different over different assets, how large its deviation from the standard value -1 is, and if it is constant over different time periods. Furthermore, relaxing the exogeneity assumption on the noise to include, e.g., rounding errors appears relevant. Working with order statistics, however, such extensions are completely different than for local averages in case of regular noise and require new concepts.

\section{Proofs\label{sec:6}}
\subsection{Crucial lemmas on the asymptotic distribution of local minima\label{subsec:pre}}
In the sequel, we write 
\[\mathcal{I}_{\tau}^n=\{\lfloor n\tau\rfloor+1,\ldots,\lfloor n\tau\rfloor+nh_n\}, ~\text{and}~\mathcal{I}_{\tau-}^n=\{\lfloor n\tau\rfloor-nh_n+1,\ldots,\lfloor n\tau\rfloor\}\,.\]
For sufficiently large $n$, it holds for all $\tau\in(0,1)$ that \(nh_n\le \lfloor n\tau\rfloor\le n-nh_n\). The standard localization procedure in high-frequency statistics allows us to assume that there exists a (global) constant $K$, such that
 \begin{align*}\max{\{|a_s(\omega)|,|\sigma_s(\omega)|,|X_s(\omega)|,|\delta_{\omega}(s,x)|/\gamma(x)\}}\le K\,,\end{align*}
for all $(\omega,s,x)\in(\Omega,\mathbb{R}_+,\mathbb{R})$. We refer to \cite{JP}, Sec.\ 4.4.1, for a proof.
\begin{lem}\label{lem1}For any $\tau$, $0\le\tau\le 1-h_n$, we have that
\[\min_{i\in \mathcal{I}_{\tau}^n} \big(Y_{i}-X_{\tau}\big)=\min_{i\in\mathcal{I}_{\tau}^n} \big(M_{t_i^n}+\epsilon_{i}-X_{\tau}\big)+\KLEINO_{\P}\big(h_n^{1/2}\big)\,,\]
where $M_t=X_{\tau}+\int_{\tau}^{t}\sigma_{\tau}\,dW_s$, $t\ge \tau$.
\end{lem}
\begin{proof}
It holds for all $i$ that
\[\big(Y_{i}-X_{\tau}\big)-\big(X_{t_i^n}-M_{t_i^n}\big)=\big(M_{t_i^n}+\epsilon_{i}-X_{\tau}\big)\,.\]
In particular, we conclude that
\[\min_{i\in\mathcal{I}_{\tau}^n}\big(Y_{i}-X_{\tau}\big)-\max_{i\in\mathcal{I}_{\tau}^n}\big(X_{t_i^n}-M_{t_i^n}\big)\le \min_{i\in\mathcal{I}_{\tau}^n}\big(M_{t_i^n}+\epsilon_{i}-X_{\tau}\big)\,.\]
Changing the roles of $\big(Y_{i}-X_{\tau}\big)$ and $\big(M_{t_i^n}+\epsilon_{i}-X_{\tau}\big)$, we obtain by the two analogous bounds and the triangle inequality that
\begin{align*}&\Big|\min_{i\in\mathcal{I}_{\tau}^n}\big(Y_{i}-X_{\tau}\big)-\min_{i\in\mathcal{I}_{\tau}^n}\big(M_{t_i^n}+\epsilon_{i}-X_{\tau}\big)\Big|\le \max_{i\in\mathcal{I}_{\tau}^n}\big|X_{t_i^n}-M_{t_i^n}\big|\\
&\quad \le \sup_{t\in[\tau,\tau+h_n]}\big|X_t-M_t\big|\le \sup_{t\in[\tau,\tau+h_n]}\Big|C_t-C_{\tau}-\smallint_{\tau}^t\sigma_{\tau}\,dW_s\Big|+\sup_{t\in[\tau,\tau+h_n]}\big|J_t-J_{\tau}\big|\,.
\end{align*}
We are left to prove that
\begin{subequations}
\begin{align}\label{help1}\sup_{t\in[\tau,\tau+h_n]}\Big|C_t-C_{\tau}-\smallint_{\tau}^t\sigma_{\tau}\,dW_s\Big|=\KLEINO_{\P}(h_n^{1/2})\,,\end{align}
\begin{align}\label{help2}\sup_{t\in[\tau,\tau+h_n]}\big|J_t-J_{\tau}\big|=\KLEINO_{\P}(h_n^{1/2})\,.\end{align}
\end{subequations}
We begin with the first term and decompose
\begin{align*}\sup_{t\in[\tau,\tau+h_n]}\Big|C_t-C_{\tau}-\smallint_{\tau}^t\sigma_{\tau}\,dW_s\Big|\le \sup_{t\in[\tau,\tau+h_n]}\Big|\smallint_{\tau}^t(\sigma_s-\sigma_{\tau})\,dW_s\Big|+\sup_{t\in[\tau,\tau+h_n]}\int_{\tau}^t |a_s|ds\,.\end{align*}
By It\^{o}'s isometry and Fubini's theorem we obtain under Assumption \ref{sigma} that
\begin{align*}
\E\Big[\Big(\int_{\tau}^t(\sigma_s-\sigma_{\tau})\,dW_s\Big)^2\Big]&=\int_{\tau}^t\E\big[(\sigma_s-\sigma_{\tau})^2\big]\,ds\\
&=\mathcal{O}\Big(\int_{\tau}^t(s-\tau)^{2\alpha}\,ds\Big)=\mathcal{O}\big((t-\tau)^{2\alpha+1}\big)\,.
\end{align*}
Applying Doob's martingale maximal inequality and using that $\sup_{t\in[\tau,\tau+h_n]}\int_{\tau}^t |a_s|ds=\mathcal{O}_{\P}(h_n)$, yields
\[\sup_{t\in[\tau,\tau+h_n]}\Big|C_t-C_{\tau}-\smallint_{\tau}^t\sigma_{\tau}\,dW_s\Big|=\mathcal{O}_{\P}\big(h_n^{(1/2+\alpha)\wedge 1}\big)\,,\]
such that \eqref{help1} holds, since $\alpha>0$. Next, consider the jump term. Under Assumption \ref{jumps} with $r\ge 1$, we obtain for all $t\in[\tau,\tau+h_n]$, with some constant $C_{\negthinspace J}$, the bound
 \begin{align*}
\E\big[\big|J_{t}-J_{\tau}\big|\big]&\le C_{\negthinspace J}\,\Big(\int_{\tau}^{t}\int_{\R}(\gamma^r(x)\wedge 1)\lambda(dx)ds\Big)^{1/r}\\ &  \le C_{\negthinspace J} |t-\tau|^{1/r}\,. \end{align*}
We used Jensen's inequality. Markov's inequality yields
\[\P\Big(\big|J_{t}-J_{\tau}\big|\ge C_{\delta} h_n^{1/r}\Big)\le C_{\delta}^{-1}h_n^{-1/r}\E\big[\big|J_{t}-J_{\tau}\big|\big]\le C_{\delta}^{-1}C_{\negthinspace J}\,,\]
which is bounded from above by $\delta$, if $C_{\delta}=(\delta C_{\negthinspace J})^{-1}$. This shows that 
\[\sup_{t\in[\tau,\tau+h_n]}\big|J_{t}-J_{\tau}\big|=\mathcal{O_\P}\big(h_n^{\max(1/r ,1)}\big)\,\]
and thus \eqref{help2}, since $r<2$.
\end{proof}
Denote by $f_+$ the positive part and by $f_-$ the negative part of some real-valued function $f$. We use the following lemma from \cite{bibinger2022}, Lemma 6.1, on an expansion of the cdf of the integrated negative part of a Brownian motion close to zero.
\begin{lem}\label{lem3}For a standard Brownian motion $(W_t)_{t\ge 0}$, it holds that
\[\P\Big(\int_0^1(W_t)_-\,dt\le x\Big)=\mathcal{O}(x^{1/3}), ~x\to 0\,.\]
\end{lem}
The next result is the key lemma for the proofs of the main results. For better readability, we establish it first for i.i.d.\ noise and equispaced observations. The two subsequent lemmas generalize the same result and its proof to serially correlated noise and non-equispaced observations. Both is combined to establish the results under Assumption \ref{noise}.
\begin{lem}\label{lem2}For any $\tau$, $0\le\tau\le 1-h_n$, we have under Assumptions \ref{sigma}, \ref{jumps} and \eqref{noise_dist} with i.i.d.\ noise and $t_i^n=i/n$, $i=0,\ldots,n$, conditional on $\sigma_{\tau}$ that
\begin{align}\label{lem3eq}-h_n^{-1/2}\min_{i\in\mathcal{I}_{\tau}^n} \big(M_{t_i^n}+\epsilon_{i}-X_{\tau}\big)\stackrel{d}{\longrightarrow} \hmn(0,\sigma_{\tau}^2)\,,\end{align}
with $M_t$ defined in Lemma \ref{lem1}.
\end{lem}
\begin{proof}
We prove pointwise convergence of the survival functions which implies the claimed convergence in distribution. We begin with similar transformations as in the proof of Proposition 3.2 of \cite{BJR}. Conditional on $\sigma_{\tau}$ means that we can treat $\sigma_{\tau}$ as a constant here. For $x\in\R$, we have that
\begin{align*}
&\P\Big(h_n^{-1/2}\min_{i\in\mathcal{I}_{\tau}^n} \big(M_{t_i^n}+\epsilon_{i}-X_{\tau}\big)>x\sigma_{\tau}\Big)=\P\Big(h_n^{-1/2}\min_{i\in\mathcal{I}_{\tau}^n} \big(\sigma_{\tau}\big(W_{t_i^n}-W_{\tau}\big)+\epsilon_{i}\big)>x\sigma_{\tau}\Big)\\
~&=\P\Big(\min_{i\in\mathcal{I}_{\tau}^n} \big(h_n^{-1/2}\big(W_{t_i^n}-W_{\tau}\big)+h_n^{-1/2}\sigma_{\tau}^{-1}\epsilon_{i}\big)>x\Big)\\
~&=\E\bigg[\prod_{i=\lfloor n\tau\rfloor+1}^{\lfloor n\tau\rfloor+nh_n}\P\Big(\epsilon_{i}>h_n^{1/2}\sigma_{\tau}\big(x-h_n^{-1/2}(W_{t_i^n}-W_{\tau})\big)|\mathcal{F}^X\Big)\bigg]\\
~&=\E\bigg[\exp\Big(\sum_{i=\lfloor n\tau\rfloor+1}^{\lfloor n\tau\rfloor+nh_n}\log\big(1-F_{\eta}\big(h_n^{1/2}\sigma_{\tau}\big(x-h_n^{-1/2}(W_{t_i^n}-W_{\tau})\big)\big)\big)\Big)\bigg]
\end{align*}
where we have used the tower rule for conditional expectations, and that $\epsilon_{i}\stackrel{iid}{\sim}F_{\eta}$. We use the illustration
\begin{align*}W_{t_i^n}-W_{\tau}=\sum_{j=1}^{i-\lfloor n\tau\rfloor}\tilde U_j,~&\tilde U_j\stackrel{iid}{\sim}\mathcal{N}(0,n^{-1}),j\ge 2,\tilde U_1\sim\mathcal{N}\big(0,t^n_{\lfloor n\tau\rfloor +1}-\tau\big)\,,\\
&U_j=h_n^{-1/2}\tilde U_j\stackrel{iid}{\sim}\mathcal{N}\big(0,(nh_n)^{-1}\big),j\ge 2,U_1\sim\mathcal{N}\big(0,h_n^{-1}\big(t^n_{\lfloor n\tau\rfloor +1}-\tau\big)\big)\,,\end{align*}
and a Riemann sum approximation with a standard Brownian motion $(B_t)$. We obtain with \eqref{noise_dist}, a first-order Taylor expansion of $z\mapsto \log(1-z)$, and dominated convergence that
\begin{align*}
&\P\Big(h_n^{-1/2}\min_{i\in\mathcal{I}_{\tau}^n} \big(M_{t_i^n}+\epsilon_{i}-X_{\tau}\big)>x\sigma_{\tau}\Big)=\\
~&=\E\Big[\exp\Big(-h_n^{1/2}\sigma_{\tau}\eta\sum_{i=\lfloor n\tau\rfloor+1}^{\lfloor n\tau\rfloor+nh_n}\Big(x-\sum_{j=1}^{i-\lfloor n\tau\rfloor}U_j\Big)_{+}(1+\KLEINO(1))\Big)\Big]\\
~&=\E\Big[\exp\Big(-h_n^{1/2}nh_n\sigma_{\tau}\eta\int_0^1(B_t-x)_{-}\,dt\,(1+\KLEINO(1))\Big)\Big]\,.
\end{align*}
Instead of setting $h_n\propto n^{-2/3}$ as in \cite{BJR}, and trying to deal with the very involved distribution in this case, observe that
\begin{align}
&\notag\P\Big(h_n^{-1/2}\min_{i\in\mathcal{I}_{\tau}^n} \big(M_{t_i^n}+\epsilon_{i}-X_{\tau}\big)>x\sigma_{\tau}\Big)=\\
~&\notag\P\big(\inf_{0\le t\le 1} B_t\ge x\big)+\E\Big[\1\big(\inf_{0\le t\le 1} B_t< x\big)\exp\Big(-h_n^{1/2}nh_n\sigma_{\tau}\eta\int_0^1(B_t-x)_{-}\,dt\,(1+\KLEINO(1))\Big)\Big]\\
~&=\P\big(\inf_{0\le t\le 1} B_t\ge x\big)+\KLEINO(1)\,,\label{crucial}
\end{align}
when $nh_n^{3/2}\to \infty$. The leading term becomes simpler in this case when the minimum of the Brownian motion over the interval dominates the noise compared to a choice of $h_n\propto n^{-2/3}$. However, since we do not have a lower bound for $\int_0^1(B_t-x)_{-}\,dt$, we need a careful estimate to show that the remainder term indeed tends to zero. Using that the first entry time $T_x$ of $(B_t)$ in $x$, conditional on $\{\inf_{0\le t\le 1} B_t< x \}$, has a bounded, continuous conditional density $f(t|T_x<1)$, we use Lemma \ref{lem3} and properties of the Brownian motion what yields for any $\delta>0$ that
\begin{align*}
&\E\Big[\1\big(\inf_{0\le t\le 1} B_t< x\big)\exp\Big(-h_n^{3/2}n\sigma_{\tau}\eta\int_0^1(B_t-x)_{-}\,dt\Big]\\
&\le \exp\big(-\big(h_n^{3/2}n\big)^{\delta}\sigma_{\tau}\eta\big)\P(\inf_{0\le t\le 1} B_t< x)+\P\Big(\inf_{0\le t\le 1} B_t< x,\int_0^1(B_t-x)_{-}\,dt\le \big(h_n^{3/2}n\big)^{-1+\delta}\Big)\\
&\le \bigg(\exp\big(-\big(h_n^{3/2}n\big)^{\delta}\sigma_{\tau}\eta\big)+\int_0^1 \P\Big(\int_s^1(B_t)_{-}\,dt\le \big(h_n^{3/2}n\big)^{-1+\delta}\Big) f(s|T_x<1)\,ds\bigg)\P(\inf_{0\le t\le 1} B_t< x)\\
&\le \bigg(\exp\big(-\big(h_n^{3/2}n\big)^{\delta}\sigma_{\tau}\eta\big)+\int_0^1 \P\Big((1-s)\int_0^1(B_t)_{-}\,dt\le \big(h_n^{3/2}n\big)^{-1+\delta}\Big) f(s|T_x<1)\,ds\bigg)\\
&\hspace*{11cm}\times\P(\inf_{0\le t\le 1} B_t< x)\,.
\end{align*}
We use that
\begin{align*}
&\int_0^1 \P\Big((1-s)\int_0^1(B_t)_{-}\,dt\le \big(h_n^{3/2}n\big)^{-1+\delta}\Big)\,ds
\\ &\le \int_0^{1-b_n} \P\Big((1-s)\int_0^1(B_t)_{-}\,dt\le \big(h_n^{3/2}n\big)^{-1+\delta}\Big)\,ds+\int_{1-b_n}^1 \,ds\\
&\le \P\Big(b_n\int_0^1(B_t)_{-}\,dt\le \big(h_n^{3/2}n\big)^{-1+\delta}\Big)+b_n=\mathcal{O}\Big(\big(h_n^{3/2}nb_n^{-1}\big)^{-\frac{1+\delta}{3}}+b_n\Big)\,,
\end{align*}
holds true with any sequence $(b_n)$, $b_n\in(0,1)$. We apply Lemma \ref{lem3} in the last step. Choosing $b_n$ minimal yields that 
\begin{align*}
&\E\Big[\1\big(\inf_{0\le t\le 1} B_t< x\big)\exp\Big(-h_n^{3/2}n\sigma_{(k-1)h_n}\eta\int_0^1(B_t-x)_{-}\,dt\Big)\Big]=\mathcal{O}\Big( \big(h_n^{3/2}n\big)^{-\frac{1+\delta}{4}}\Big),
\end{align*}
almost surely. From the unconditional Lévy distribution of $T_x$, $f(s|T_x<1)$ is explicit, but its precise form does not influence the asymptotic order. We have verified \eqref{crucial}.

It is well known that by the reflection principle it holds that
\[\P\big(-\inf_{0\le t\le 1} B_t\ge x\big)=\P\big(\sup_{0\le t\le 1} B_t\ge x\big)=2\P\big(B_1\ge x\big)=\P\big(|B_1|\ge x\big)\,,\]
for $x\ge 0$, and since $|B_1|\sim \hn(0,1)$, we conclude the result.
\end{proof}

\begin{lem}\label{lem4}More generally than in Lemma \ref{lem2}, \eqref{lem3eq} holds under Assumptions \ref{sigma}, \ref{jumps} and for serially correlated noise under Assumption \ref{noise} with $t_i^n=i/n$, $i=0,\ldots,n$.
\end{lem}
\begin{proof}Non-negativity of the noise, $\epsilon_i\ge 0$, for all $i\in\{0,\ldots,n\}$, yields that
\begin{align}\label{noise_lowbound}\min_{i\in\mathcal{I}_{\tau}^n}\big(M_{t_i^n}-X_{\tau}+\epsilon_{i}\big)\ge  \min_{i\in\mathcal{I}_{\tau}^n}\big(M_{t_i^n}-X_{\tau}\big)+\min_{i\in\mathcal{I}_{\tau}^n}\epsilon_{i}\ge \min_{i\in\mathcal{I}_{\tau}^n}\big(M_{t_i^n}-X_{\tau}\big)\,.\end{align}
Based on \eqref{noise_lowbound}, we obtain that
\begin{align*}
\P\Big(h_n^{-1/2}\min_{i\in\mathcal{I}_{\tau}^n} \big(M_{t_i^n}+\epsilon_{i}-X_{\tau}\big)>x\sigma_{\tau}\Big)&\ge \P\Big(h_n^{-1/2}\min_{i\in\mathcal{I}_{\tau}^n} \big(M_{t_i^n}-X_{\tau}\big)>x\sigma_{\tau}\Big)\\
&=\P\Big(\min_{i\in\mathcal{I}_{\tau}^n}h_n^{-1/2} \big(W_{t_i^n}-W_{\tau}\big)>x\Big)\\
&=\P\big(\inf_{0\le t\le 1} B_t\ge x\big)+\KLEINO(1)\,.
\end{align*}
In general, irrespective of the serial correlation structure of $(\epsilon_i)$, the leading term in \eqref{crucial} thus provides an asymptotic lower bound for the tail function of local minima under non-negative noise. Therefore, an asymptotic upper bound inequality under $m_n$-dependence is sufficient to conclude. We use the elementary fact that for any subsets $\mathcal{S}_2\subseteq\mathcal{S}_1\subseteq\{0,\ldots,n\}$, and any $a_i\in\R$, it holds that
\begin{align}\min_{i\in\mathcal{S}_{2}}a_i\le \min_{i\in\mathcal{S}_{1}} a_i\,.\label{subsetineq}\end{align}
Setting $\mathcal{I}_{\tau}^{m_n}=\{\lfloor n\tau\rfloor+1,\lfloor n\tau\rfloor+m_n+1,\lfloor n\tau\rfloor+2m_n+1\ldots,\lfloor n\tau\rfloor+\lfloor nh_nm_n^{-1}\rfloor m_n+1\}$, it is under $m_n$-dependence possible to use factorization of the conditional tail function of $\min_{i\in\mathcal{I}_{\tau}^{m_n}}\epsilon_i$ in an analogous way as for $\min_{i\in\mathcal{I}_{\tau}^{n}}\epsilon_i$ in the proof of Lemma \ref{lem2}. This and \eqref{subsetineq} yield
\begin{align*}
&\P\Big(h_n^{-1/2}\min_{i\in\mathcal{I}_{\tau}^n} \big(M_{t_i^n}+\epsilon_{i}-X_{\tau}\big)>x\sigma_{\tau}\Big)\le 
\P\Big(h_n^{-1/2}\min_{i\in\mathcal{I}_{\tau}^{m_n}} \big(M_{t_i^n}+\epsilon_{i}-X_{\tau}\big)>x\sigma_{\tau}\Big)\\
~&=\E\bigg[\exp\Big(\sum_{i\in\mathcal{I}_{\tau}^{m_n}}\log\big(1-F_{\eta}\big(h_n^{1/2}\sigma_{\tau}\big(x-h_n^{-1/2}(W_{t_i^n}-W_{\tau})\big)\big)\big)\Big)\bigg]\,.
\end{align*}
The same steps as in the proof of Lemma \ref{lem2} are applied to this upper bound under $m_n$-dependence. However, the Riemann sum is by the subset $\mathcal{I}_{\tau}^{m_n}\subset \mathcal{I}_{\tau}^{n}$ with distances $m_n(nh_n)^{-1}$ instead of $(nh_n)^{-1}$, such that
\begin{align*}
\P\Big(h_n^{-1/2}\min_{i\in\mathcal{I}_{\tau}^n} \big(M_{t_i^n}+\epsilon_{i}-X_{\tau}\big)>x\sigma_{\tau}\Big)\le
\E\Big[\hspace*{-.05cm}\exp\hspace*{-.05cm}\Big(\hspace*{-.1cm}-\hspace*{-.05cm}h_n^{3/2}nm_n^{-1}\sigma_{\tau}\eta\hspace*{-.05cm}\int_0^1\hspace*{-.05cm}(B_t-x)_{-}\,dt\,(1+\KLEINO(1)\hspace*{-.02cm})\hspace*{-.05cm}\Big)\hspace*{-.05cm}\Big].
\end{align*}
If $nh_n^{3/2}m_n^{-1}\to\infty$ under Assumption \ref{noise}, the same steps as in the proof of Lemma \ref{lem2} then yield
\begin{align*}
&\P\Big(h_n^{-1/2}\min_{i\in\mathcal{I}_{\tau}^n} \big(M_{t_i^n}+\epsilon_{i}-X_{\tau}\big)>x\sigma_{\tau}\Big)\le\P\big(\inf_{0\le t\le 1} B_t\ge x\big)+\KLEINO(1)\,.
\end{align*}
Combining the asymptotic lower and upper bound inequalities (sandwich theorem) finishes the proof.
\end{proof}
\begin{lem}\label{lem5}More generally than in Lemma \ref{lem2}, \eqref{lem3eq} holds with non-regularly spaced observations times under Assumptions \ref{noise}, \ref{sigma} and \ref{jumps}.
\end{lem}
\begin{proof}
An analogous approximation as left-hand side in \eqref{crucial} is established, such that the remainder of the proof of Lemma \ref{lem2} applies in the same way. 
%By the mean value theorem there exist $\xi_i\in[(i-1)/n,i/n]$, such that under Assumption \ref{noise}
%\begin{align*}
%t_i^n-t_{i-1}^n=F^{-1}\left(\frac{i}{n}\right)-F^{-1}\left(\frac{i-1}{n}\right)=(F^{-1})^{\prime}(\xi_i)\cdot \frac{1}{n}\,,
%\end{align*}
%and equivalently
%\begin{align*}
%h_n^{-1}\big(t_i^n-t_{i-1}^n\big)nh_n=(F^{-1})^{\prime}(\xi_i)\,.
%\end{align*}
With a standard Brownian motion $(B_t)$, we obtain that
%\begin{align*}
%&\P\Big(h_n^{-1/2}\min_{i\in\mathcal{I}_{\tau}^n} \big(M_{t_i^n}+\epsilon_{i}-X_{\tau}\big)>x\sigma_{\tau}\Big)=\\
%~&=\E\bigg[\exp\Big(-h_n^{1/2}\sigma_{\tau}\eta \sum_{i=\lfloor n\tau\rfloor+1}^{\lfloor n\tau\rfloor+nh_n}\big(x-h_n^{-1/2}(W_{t_i^n}-W_{\tau})\big)_+\Big)\bigg]\\
%~&=\E\bigg[\exp\Big(-h_n^{1/2}\sigma_{\tau}\eta \sum_{i=\lfloor n\tau\rfloor+1}^{\lfloor n\tau\rfloor+nh_n}\big(B_{t_i^n}-x\big)_{-}\,h_n^{-1}(t_i^n-t_{i-1}^n)nh_nF^{\prime}(\tilde \xi_i)\Big)\bigg]\\
%~&=\E\Big[\exp\Big(-h_n^{3/2}n\sigma_{\tau}\eta\int_0^1(B_t-x)_{-}\,F^{\prime}(t)\,dt\,(1+\KLEINO(1))\Big)\Big]\,.
%\end{align*}
\begin{align*}
&\P\Big(h_n^{-1/2}\min_{i|t_i^n\in[\tau,\tau+h_n]} \big(M_{t_i^n}+\epsilon_{i}-X_{\tau}\big)>x\sigma_{\tau}\Big)=\\
~&=\E\bigg[\exp\Big(-h_n^{1/2}\sigma_{\tau}\eta \sum_{i|t_i^n\in[\tau,\tau+h_n]}\big(h_n^{-1/2}(W_{t_i^n}-W_{\tau})-x\big)_-\Big)\bigg]\\
~&=\E\bigg[\exp\Big(-h_n^{1/2}\sigma_{\tau}\eta n \int_{\tau}^{\tau+h_n}\big(h_n^{-1/2}(W_{t}-W_{\tau})-x\big)_-\,d\mu_n^S(t)\Big)\bigg]\\
~&=\E\Big[\exp\Big(-h_n^{1/2}n\sigma_{\tau}\eta\int_0^{h_n}(h_n^{-1/2}B_t-x)_{-}\,f^S(t+\tau)\,dt\,(1+\KLEINO(1))\Big)\Big]\\
~&=\E\Big[\exp\Big(-h_n^{3/2}n\sigma_{\tau}f^S(\tau)\eta\int_0^1(B_s-x)_{-}\,ds\,(1+\KLEINO(1))\Big)\Big]\,.
\end{align*}
The first equality is taken from the proof of Lemma \ref{lem2}. The second equality follows from the definition of the empirical measure $\mu_n^S$ in Assumption \ref{noise}. Since the complement of the intersection of $\Omega_0^X=\{\omega\in\Omega^X|\mu_n^S(\omega)\stackrel{w}{\rightarrow}\mu^S(\omega)\}$ and $\{\omega\in\Omega^X|\omega\mapsto\big(h_n^{-1/2}(W_{t}-W_{\tau})-x\big)_-~\text{is continuous}\}$ is a $\P^X$ null set, and the continuous function is moreover bounded on the considered compact time interval, the almost sure weak convergence with limit density $f^S$ from Assumption \ref{noise} guarantees the convergence of the integrals in the third equality, also under dependence between $(W_t)$ and $\mu_n^S$. Since the integrals converge and the expression inside the expectation is bounded by 1, dominated convergence yields convergence of the expectations. We use the shift and the self-similarity properties of the Brownian motion $(B_t)$, $t\in[0,1]$, with a change of variables $s=h_n^{-1}\cdot t$ in the last equality. We then use the imposed continuity of $f^S$ to approximate it locally constant over the block. With the additional factor $f^{S}(\tau)$, using the imposed condition $f^S(\tau)>0$, for any $\tau$, the proof is completed along the lines of the proof of Lemma \ref{lem2} below \eqref{crucial}.
%The first equality has been adopted from the proof of Lemma \ref{lem2}. The second equality uses the empirical measure $\mu_n^S$ from Assumption \ref{noise} and the third follows with the assumed (almost sure) weak convergence with limit density $f^S$. The integrals thus converge almost surely and the expression in the expectation being bounded by 1, the convergence of the expectations is concluded with dominated convergence. We use the shift and the self-similarity property of the Brownian motion $(B_t)$, $t\in[0,1]$, with a change of variables $s=h_n^{-1}\cdot t$ in the last equation. In the last equality we use the imposed continuity of $f^S$ to approximate it locally constant over the block. With the additional factor $f^{S}(\tau)$ and using the imposed condition $f^S(\tau)>0$, for any $\tau$, the proof is completed along the lines of the proof of Lemma \ref{lem2} below \eqref{crucial}.}}
%The first equality is taken from the proof of Lemma \ref{lem2}.
\end{proof}

\subsection{Uniformly consistent spot volatility estimation}
In the sequel, we write $A_n\lesssim B_n$ for two real sequences, if there exists some $n_0\in\N$ and a constant $K$, such that $A_n\le K B_n$, for all $n\ge n_0$. With the estimate from Lemma \ref{lem1} and using \eqref{help1} under the stronger condition that $(\sigma_t)$ is Hölder continuous with regularity $\alpha$, we obtain with the process 
\[M_t=X_{\lfloor th_n^{-1}\rfloor h_n}+\int_{\lfloor th_n^{-1}\rfloor h_n}^t \sigma_{\lfloor th_n^{-1}\rfloor h_n}\,dW_s\,,\]
that
\begin{align*}
&\max_{k=0,\ldots,h_n^{-1}-1}\Big|\min_{i\in \mathcal{I}_k^n}\big(Y_i-X_{kh_n}\big)-\min_{i\in \mathcal{I}_k^n}\big(\epsilon_i+\sigma_{kh_n}(W_{t_i^n}-W_{kh_n})\big)\Big|\\
&\le \max_{k=0,\ldots,h_n^{-1}-1}\sup_{t\in[kh_n,(k+1)h_n]}|X_t-M_t|\le\sup_{t\in[0,1]}|X_t-M_t|\\
&\le \sup_{t\in[0,1]}\int_{\lfloor th_n^{-1}\rfloor h_n}^t|a_s|\,ds+\sup_{t\in[0,1]}\int_{\lfloor th_n^{-1}\rfloor h_n}^t\big(\sigma_t-\sigma_{\lfloor th_n^{-1}\rfloor h_n}\big)\,dW_s\\
&=\mathcal{O}_{\P}\big(h_n^{(1/2 +\alpha)\wedge 1 }\big)\,,
\end{align*}
with $\alpha$ from \eqref{vola}. Subtracting $X_{rh_n}$ from $m_{r,n}$ and $m_{r-1,n}$ in differences $m_{r,n}-m_{r-1,n}$, we obtain with this error bound that
\begin{align*}\max_{k=1,\ldots,h_n^{-1}-1}\bigg|\hat\sigma^2_{kh_n-}-\frac{\pi}{2(\pi-2)K_n}\sum_{r=(k-K_n)\wedge 1}^{k-1}h_n^{-1}\big(\tilde m_{r,n}-\tilde m^*_{r-1,n})^2\bigg|=\mathcal{O}_{\P}\big(h_n^{\alpha\wedge 1/2}\big)\,,
\end{align*}
with 
\begin{align*}
\tilde m_{r,n}&=\min_{i\in\mathcal{I}_r^{n}}\big(\epsilon_i+\sigma_{(r-1)h_n}(W_{t_i^n}-W_{rh_n})\big)\,,~\text{and}\\
\tilde m^*_{r-1,n}&=\min_{i\in\mathcal{I}_{r-1}^{n}}\big(\epsilon_i-\sigma_{(r-1)h_n}(W_{rh_n}-W_{t_i^n})\big)\,.
\end{align*}
Denote by $\E_{\sigma_{(k-1)h_n}}$ expectations with respect to conditional probability measures given $\sigma_{(k-1)h_n}$. The remainder of the proof relies on a maximum and a moment inequality, for which we use that the conditional moments of $\tilde m_{r,n}$ satisfy
\begin{align*}\E_{\sigma_{(r-1)h_n}}\big[\big|h_n^{-1/2}\tilde m_{r,n}\big|^p\big]&=p\int_0^{\infty}x^p\,\P_{\sigma_{(r-1)h_n}}\big(|h_n^{-1/2}\tilde m_{r,n}|>x\big)\,dx\\
&=p\int_0^{\infty}x^p\,\P_{\sigma_{(r-1)h_n}}\big(\sigma_{(r-1)h_n}\sup_{0\le t\le 1} B_t>x\big)\,dx+\KLEINO_{\P}(1)\\
&=p\int_0^{\infty}x^p\,\P_{\sigma_{(r-1)h_n}}\big(\sigma_{(r-1)h_n}|B_1|>x\big)\,dx+\KLEINO_{\P}(1)\\
&=\sigma_{(r-1)h_n}^p M_p+\KLEINO_{\P}(1)\,,
\end{align*}
with $M_p$ the moments of the standard half-normal distribution. We have used Lemma \ref{lem2}. We conclude the existence of all moments of $\tilde m_{r,n}$, and analogously for $\tilde m_{r,n}^*$.
We use a generalization of Rosenthal's inequality which states for i.i.d.\ random variables $Y_1,\ldots, Y_N$, with zero mean and $\E[|Y_1|^p]<\infty$, $p\in\N$, that
\[\E\Big[\Big|\sum_{i=1}^NY_i\Big|^p\Big]\le C_p\,\max\Big(\sum_{i=1}^N\E[|Y_i|^p]\,,\,\Big(\sum_{i=1}^N\E[Y_i^2]\Big)^{p/2}\Big)\,,\]
with some constant $C_p$ depending on $p$, such that for $p>2$, it holds that
\[\E\Big[\Big|\frac{1}{N}\sum_{i=1}^NY_i\Big|^p\Big]\lesssim N^{-p/2}\,.\]
By \cite{Burkholder} the inequality extends to martingale increments. Note that while $\tilde m_{r,n}$ and $\tilde m_{r,n}^*$ are correlated for the same $r$, $(\tilde m_{r,n})$ is a sequence of uncorrelated random variables and we conclude that
\begin{align*}
\E\bigg[\Big|\frac{1}{K_n}\sum_{r=1}^{K_n}\Big(\big(h_n^{-1/2}\tilde m_{r,n}\big)^2-\E_{\sigma_{(r-1)h_n}}\big[\big(h_n^{-1/2}\tilde m_{r,n}\big)^2\big]\Big)\Big|^p\bigg]\lesssim K_n^{-p/2}\,,
\end{align*}
and analogous bounds when replacing $\tilde m_{r,n}$ by $\tilde m_{r-1,n}^*$, or $\tilde m_{r,n}\tilde m_{r-1,n}^*$. By the tower rule for conditional expectations, the considered random variables in the sum have mean zero. With the Markov inequality we hence obtain for all $\epsilon>0$ that
\begin{align*}
&\P\bigg(K_n^{\gamma}\max_{k=1,\ldots,h_n^{-1}-1}\Big|\frac{1}{K_n}\sum_{r=(k-K_n)\wedge 1}^{k-1}\Big(\frac{h_n^{-1}\pi}{2(\pi-2)}\big(\tilde m_{r,n}-\tilde m_{r-1,n}^*\big)^2\Big)-\Psi_n(\sigma^2_{(r-1)h_n})\Big|>\epsilon\bigg)\\
&\le h_n^{-1}\bigg(\P\bigg(K_n^{\gamma-1}\Big|\sum_{r=1}^{K_n}\Big(\big(h_n^{-1/2}\tilde m_{r,n}\big)^2-\E_{\sigma_{(r-1)h_n}}\big[\big(h_n^{-1/2}\tilde m_{r,n}\big)^2\big]\Big)\Big|>\frac{\epsilon}{3}\bigg)\\
&\quad + \bigg(\P\bigg(K_n^{\gamma-1}\Big|\sum_{r=1}^{K_n}\Big(\big(h_n^{-1/2}\tilde m_{r-1,n}^*\big)^2-\E_{\sigma_{(r-1)h_n}}\big[\big(h_n^{-1/2}\tilde m_{r-1,n}^*\big)^2\big]\Big)\Big|>\frac{\epsilon}{3}\bigg)\\
&\quad + \bigg(\P\bigg(K_n^{\gamma-1}\Big|\sum_{r=1}^{K_n}\Big(2h_n^{-1/2}\tilde m_{r,n}\tilde m_{r-1,n}^*-\E_{\sigma_{(r-1)h_n}}\big[2h_n^{-1/2}\tilde m_{r,n}\tilde m_{r-1,n}^*\big]\Big)\Big|>\frac{\epsilon}{3}\bigg)\bigg)\\
&\lesssim h_n^{-1} \bigg(\E\bigg[\Big|K_n^{\gamma-1}\sum_{r=1}^{K_n}\Big(\big(h_n^{-1/2}\tilde m_{r,n}\big)^2-\E_{\sigma_{(r-1)h_n}}\big[\big(h_n^{-1/2}\tilde m_{r,n}\big)^2\big]\Big)\Big|^p\bigg]\\
&\quad + \E\bigg[\Big|K_n^{\gamma-1}\sum_{r=1}^{K_n}\Big(\big(h_n^{-1/2}\tilde m_{r-1,n}^*\big)^2-\E_{\sigma_{(r-1)h_n}}\big[\big(h_n^{-1/2}\tilde m_{r-1,n}^*\big)^2\big]\Big)\Big|^p\bigg]\\
&\quad + \E\bigg[\Big|K_n^{\gamma-1}\sum_{r=1}^{K_n}\Big(2h_n^{-1/2}\tilde m_{r,n}\tilde m_{r-1,n}^*-\E_{\sigma_{(r-1)h_n}}\big[2h_n^{-1/2}\tilde m_{r,n}\tilde m_{r-1,n}^*\big]\Big)\Big|^p\bigg]\\
&\quad \lesssim K_n^{(\gamma-1/2)p} h_n^{-1}\to 0\,.
\end{align*}
Choosing $p$ sufficiently large, the term converges to zero as $n\to\infty$. For $(\sigma_t)$ being Hölder continuous with regularity $\alpha$, we have that
\begin{align*}&\max_{k=1,\ldots,h_n^{-1}-1}\Big|\frac{1}{K_n}\sum_{r=(k-K_n)\wedge 1}^{k-1}\sigma^2_{(r-1)h_n}-\sigma^2_{(k-1)h_n}\Big|\\
&=\max_{k=1,\ldots,h_n^{-1}-1}\Big|\frac{1}{K_n}\sum_{r=(k-K_n)\wedge 1}^{k-1}\big(\sigma^2_{(r-1)h_n}-\sigma^2_{(k-1)h_n}\big)\Big|\\
&=\mathcal{O}_{\P}\Big(K_n^{-1}\sum_{j=1}^{K_n}(j h_n)^{\alpha}\Big)=\mathcal{O}_{\P}\big((K_n h_n)^{\alpha}\big)=\KLEINO_{\P}\big(K_n^{-1/2}\big)\,.
\end{align*}
By the differentiability of $\Psi_n(\,\cdot\,)$, based on Eq.\ (A.35) from \cite{BJR}, a local approximation of the volatility in the argument of $\Psi_n$ is asymptotically negligible as well. This finishes the proof of Proposition \ref{uniform}.

Since by \eqref{psiapprox} it holds that $\Psi_n(x)\to x$, for all $x>0$, we obtain that $|{{\Psi_n}}\big(\sigma_{kh_n}^2\big)-\sigma_{kh_n}^2|\to 0$, uniformly over all $k$, such that Corollary \ref{uniformcor} readily follows with the triangle inequality.

\subsection{Asymptotic distribution of jump estimates}
By Lemmas \ref{lem1}, \ref{lem2}, \ref{lem4} and \ref{lem5}, we obtain for any $\tau$, $0\le\tau\le 1-h_n$, conditional on $\sigma_{\tau}$, that 
\begin{subequations}
\begin{align}\label{locconv1}
-h_n^{-1/2} \Big(\min_{i\in \mathcal{I}_{\tau}^n}Y_{i}-X_{\tau}\Big)\,\stackrel{d}{\longrightarrow} \hmn(0,\sigma_{\tau}^2)\,.
\end{align}
With $M_t=X_{\tau-}-\int_{t}^{\tau-}\sigma_{\tau-}\,dW_s$, completely analogous proofs as for Lemmas \ref{lem1}, \ref{lem2}, \ref{lem4} and \ref{lem5}, using that $(-W_t)$ is as well a Brownian motion, shows that for any $\tau$, $h_n\le\tau\le 1$, conditional on $\sigma_{\tau}$, it holds true that
\begin{align}\label{locconv2}
-h_n^{-1/2}\Big(\min_{i\in \mathcal{I}_{\tau-}^n}Y_{i}-X_{\tau-}\Big)\,\stackrel{d}{\longrightarrow} \hmn(0,\sigma_{\tau-}^2)\,.
\end{align}
\end{subequations}
Moreover, by the strong Markov property of $(W_s)$ and since the $(\epsilon_i)$ are $m_n$-dependent, covariances of the statistics in \eqref{locconv1} and \eqref{locconv2} for $h_n\le\tau\le 1-h_n$ vanish, such that we deduce joint weak convergence. For any $\tau\in(0,1)$, $h_n\le\tau\le 1-h_n$ holds true for sufficiently large $n$. Continuous mapping readily yields \eqref{stablecltbhr}.\\
We show that, when not conditioning on $\sigma_{\tau}$, the convergences in \eqref{locconv1} and \eqref{locconv2} are stable in law with respect to the $\sigma$-field $\mathcal{F}^X$. The proof is analogous for both sequences, and we restrict to the first one. The stable convergence is equivalent to the joint weak convergence of $V_n=-h_n^{-1/2}\min_{i\in \mathcal{I}_{\tau}^n} (Y_{i}-X_{\tau})$ with any $\mathcal{F}^X$-measurable, bounded random variable $Z$. That is,
\begin{align}\label{stable2}\E\left[Z g(V_n)\right]\rightarrow \E\left[Z \cdot g(\sigma_{\tau} |U|)\right]\, \end{align}
as $n\to\infty$, for any continuous bounded function $g$, and with $U\sim\mathcal{N}(0,1)$ being independent of $\mathcal{F}^X$. By Lemma 1 it suffices to prove this for $\tilde V_n=-h_n^{-1/2}\min_{i\in\mathcal{I}_{\tau}^n} \big(M_{t_i^n}+\epsilon_{i}-X_{\tau}\big)$, and $Z$ measurable w.r.t.\ $\sigma(\int_0^t\sigma_s\,dW_s,0\le t\le 1)$. Define the sequence of intervals $A_n=(\tau, (\tau+h_n)\wedge 1]$, and consider the sequences of decompositions 
\begin{align*}\tilde C(n)_t=\int_0^t\1_{A_n}\negthinspace(s)\,\sigma_{s}\,dW_s\,,\,\bar C(n)_t=\int_0^t\sigma_{s}\,dW_s-\tilde C(n)_t\,,\end{align*}
of $(\int_0^t\sigma_{s}\,dW_s)_{t\ge 0}$. If $\mathcal{H}_n$ denotes the $\sigma$-field generated by $\bar C(n)_t$ and $\mathcal{F}_0$, then $\big(\mathcal{H}_n\big)_n$ is an isotonic sequence with $\bigvee_n \mathcal{H}_n=\sigma(\int_0^t\sigma_s\,dW_s,0\le t\le 1)$. Since $\E[Z|\mathcal{H}_n]\rightarrow Z$ in $L^1(\P)$, it suffices to show that
\begin{align}\E[Z g (\tilde V_n)]\rightarrow \E[Z\cdot g(\sigma_{\tau} |U|)]\,,\end{align}
for $Z$ being $\mathcal{H}_q$-measurable for some $q$. Note that $\sigma_{\tau}$ is $\mathcal{H}_q$-measurable for any $q$. Since, for all $n\ge q$, conditional on $\mathcal{H}_q$, $\tilde V_n$ has a law independent of $\bar C(n)_t$, we obtain with the tower rule of conditional expectations:
\begin{align*}
\lim_{n\to\infty}\E[Z g (\tilde{V}_n)]&=\lim_{n\to\infty}\E\big[\E[Z g (\tilde{V}_n)\vert \mathcal{H}_q]\big]\\
&=\lim_{n\to\infty}\E\big[Z\E[g (\tilde V_n)\vert \mathcal{H}_q]\big]=\E[Z\cdot g(\sigma_{\tau} |U|)]\,,
\end{align*}
for $Z$ being $\mathcal{H}_q$-measurable where we can use Lemma \ref{lem2} and dominated convergence in the last step. The stability allows to conclude \eqref{stablecltf} from \eqref{spotclttr} and the analogous consistency of the truncated version of \eqref{simpleestimator2} with \eqref{locconv1} and \eqref{locconv2}. This completes the proof of Theorem \ref{thmloc}.
%%%%%%%%%%%%%%%%%%%%%%%%%%%%%%%%%%%%%%%%%%%%%%%%%%%%%%%%
%%%%%%%%%%%%%%%%%%%%%%%%%%%%%%%%%%%%%%%%%%%%%%%%%%%%%%
\subsection{Asymptotic distribution under the null hypothesis of the global test}
To prove the asymptotic result for the global test, we establish the extreme value convergence for the maximum of i.i.d.\ random variables distributed as the absolute difference of two independent, standard half-normally distributed random variables in the next proposition. This is based on classical extreme value theory and an expansion of convolution tails. 
\begin{prop}\label{gumbelprop}
Let $(V_1,\ldots ,V_n,\tilde V_1,\ldots, \tilde V_n)$ be a $2n$-dimensional vector of i.i.d.\ standard normally distributed random variables. It holds true that
\begin{align}\label{gumbelh}
\frac{\max_{1\le i\le n}\big||V_i|-|\tilde V_i|\big|-b_n}{a_n}\stackrel{d}{\longrightarrow} \Lambda\,,
\end{align}
where $\Lambda$ denotes the standard Gumbel distribution, with the sequences
\begin{align}\label{g6}
a_n=\frac{1}{\sqrt{2\log (2n)}}~, ~\mbox{and}~~b_n=\sqrt{2\log (2n)}+\delta_n~, ~\mbox{with}~~\delta_n=-\frac{\log(\pi\log (2n))}{\sqrt{2\log (2n)}}\,.
\end{align}
\end{prop}
\begin{proof}
Denote with $g$ the density of $|V_1|-|\tilde V_1|$ on the positive real line and $G$ and $\bar G=1-G$ the associated cdf and survival function, respectively. For $g(x)$, $x>0$, we compute
\begin{align*}g(x)&=\frac{2}{\pi}\int_0^{\infty}e^{-u^2/2}\,e^{-(x+u)^2/2}\,du\,\\
&=\frac{\sqrt{2}}{\pi}\,e^{-x^2/4}\int_{x/\sqrt{2}}^{\infty}e^{-v^2/2}\,dv\,\\
&=\frac{2}{\sqrt{\pi}}\,e^{-x^2/4}\Big(1-\Phi\big(x/\sqrt{2}\big)\Big)\,\\
&=\sqrt{\frac{1}{\pi}}\,e^{-x^2/4}\,\text{erfc}(x/2)\,,
\end{align*}
with $\Phi$ the cdf of the standard normal distribution and 
\[\Phi(x)=\frac{1+\text{erf}(x/\sqrt{2})}{2}~,~\text{erfc}(x)=1-\text{erf}(x)~.\] 
For asymptotic equivalence of two positive functions $f$ and $g$, we write $f\asymp g$, which means that
\[\lim_{x\to\infty}\frac{f(x)}{g(x)}=1\,.\]
Based on l'H\^{o}pital's rule we obtain that
\[ \text{erfc}(x)\asymp \frac{e^{-x^2}}{\sqrt{\pi} x}\,,\]
and conclude that
\begin{align*}g(x)\asymp \frac{2}{\pi}\frac{e^{-x^2/2}}{x}\,.\end{align*}
Based on l'H\^{o}pital's rule, we obtain that the associated survival function $\bar G$ satisfies
\[\bar G(x)=\int_x^{\infty} g(t)\,d t \asymp \frac{2}{\pi}\,\frac{e^{-x^2/2}}{x^2}\,.\]
By Equation (1.2.4) of \cite{haan}, 
\begin{align}\label{gumbel2}
\frac{\max_{1\le i\le n}\big(|V_i|-|\tilde V_i|\big)-b_n}{a_n}\stackrel{d}{\longrightarrow} \Lambda\,,
\end{align}
is satisfied with some sequences $(a_n)$ and $(b_n)$, if there exists a function $f$, such that for all $x\in\R$, the survival function $\bar G$ satisfies
\begin{align}\,\lim_{t\uparrow x^*}\frac{\bar G(t+xf(t))}{\bar G(t)}=e^{-x}~.\label{crit}\end{align}
$x^*$ is the right end-point of the distribution which is $x^*=+\infty$ here. In this case, \eqref{crit} is satisfied with $f(t)=t^{-1}$, since
\[\lim_{t\uparrow\infty}\frac{\bar G\Big(t+\tfrac{x}{t}\Big)}{\bar G(t)}=\lim_{t\uparrow\infty}\frac{\Big(t+\tfrac{x}{t}\Big)^{-2}\exp{\big(-t^2/2-x-x^2/(2t^2)\big)}}{t^{-2}e^{-t^2/2}}=e^{-x}\,,\,\forall\,x\in\R\,.\]
We show that \eqref{gumbel2} applies with 
\begin{align}
a_n=\frac{1}{\sqrt{2\log (n)}}~, ~\mbox{and}~~b_n=\sqrt{2\log (n)}+\delta_n~, ~\mbox{with}~~\delta_n=-\frac{\log(\pi\log (n))}{\sqrt{2\log (n)}}\,.
\end{align} 
We can determine $(a_n)$ and $(b_n)$ based on
\begin{align}\label{seq}\lim_{n\to\infty}n\,\bar G(a_n t+ b_n)=-\log(\Lambda(t))=e^{-t}\,,\end{align}
or use $b_n=U(n)$, with $U$ the general notation for the left-continuous generalized inverse of $1/(1-G)$, see Remark 1.1.9 in \cite{haan}. Setting $b_n=\sqrt{2\log(n)}+\delta_n$, with a null sequence $\delta_n$, yields that
\[n\asymp\frac{\pi}{2}\,b_n^2\,e^{b_n^2/2}\asymp \pi\log(n)\exp\big(\log(n)+\sqrt{2\log(n)}\delta_n\big)\,,\]
and we find that the identity holds true for
\[\delta_n=-\frac{\log(\pi\log (n))}{\sqrt{2\log (n)}}\,.\]
Computing $n U^{\prime}(n)$, starting with $U(n)=b_n$, gives for the sequence $(a_n)$ that 
\[a_n=(2\log(n))^{-1/2}~.\] 
We exploit the symmetry of the distribution to conclude \eqref{gumbelh} readily from \eqref{gumbel2}, see Lemma 1 in \cite{arxiv2} for a more detailed argument.
\end{proof}
Let us point out that the sequences $a_n$ and $b_n$ are different compared to the Gumbel convergence in the standard normal case, see \cite{arxiv2} for a comparison and discussion. These differences are of course crucial.\\

The triangle inequality yields for all $k=1,\ldots,h_n^{-1}-1$, that
\begin{align*}\big|m_{k,n}-m_{k-1,n}\big|\le \big| m_{k,n}-\tilde m_{k,n}\big|+\big|\tilde m_{k-1,n}^*- m_{k-1,n}\big|+\big|\tilde m_{k,n}-\tilde m_{k-1,n}^*\big|.
\end{align*}
By an analogous bound starting with $\big|\tilde m_{k,n}-\tilde m_{k-1,n}^*\big|$, and elementary transformations, we obtain that
\begin{align*}&\bigg|\max_{k=1,\ldots,h_n^{-1}-1}\big|m_{k,n}-m_{k-1,n}\big|-\max_{k=1,\ldots,h_n^{-1}-1}\big|\tilde m_{k,n}-\tilde m_{k-1,n}^*\big|\bigg|\\
&\le \max_{k=1,\ldots,h_n^{-1}-1}\big|m_{k,n}-\tilde m_{k,n}\big|+\max_{k=1,\ldots,h_n^{-1}-1}\big|m_{k,n}-\tilde m_{k,n}^*\big|\\
&\le 2\sup_{t,s:\,|t-s|\le h_n}\Big|\int_t^s (\sigma_u-\sigma_s)\,dW_u\Big|+2\sup_{t,s:\,|t-s|\le h_n}\int_t^s|a_u|\,du\\
&=\mathcal{O}_{\P}\big(h_n^{1/2+\alpha}\log(h_n^{-1})\big)=\KLEINO_{\P}\Big(h_n^{1/2}\Big(\log(2h_n^{-1})\Big)^{-1/2}\Big)\,,\end{align*}
using that $(\sigma_t)$ is Hölder continuous with exponent $\alpha$. For the difference between statistics with estimated and true volatilities we use the estimate
\begin{align*}&\bigg|\max_{k=1,\ldots,h_n^{-1}-1}\frac{\big|\tilde m_{k,n}-\tilde m_{k-1,n}^*\big|}{\hat\sigma_{kh_n}}-\max_{k=1,\ldots,h_n^{-1}-1}\frac{\big|\tilde m_{k,n}-\tilde m_{k-1,n}^*\big|}{\sigma_{kh_n}}\bigg|\\
&\le \max_{k=1,\ldots,h_n^{-1}-1}\frac{\big|\tilde m_{k,n}-\tilde m_{k-1,n}^*\big|}{\sigma_{kh_n}}\max_{k=1,\ldots,h_n^{-1}-1}\Big|\frac{\sigma_{kh_n}}{\hat\sigma_{kh_n}}-1\Big|\,.
\end{align*}
The uniform consistency of the spot volatility estimator yields that
\begin{align*}\max_{k=1,\ldots,h_n^{-1}-1}\frac{\big|\tilde m_{k,n}-\tilde m_{k-1,n}^*\big|}{\hat\sigma_{kh_n}}=\max_{k=1,\ldots,h_n^{-1}-1}\frac{\big|\tilde m_{k,n}-\tilde m_{k-1,n}^*\big|}{\sigma_{kh_n}}+\KLEINO_{\P}\Big(\big(\log(2h_n^{-1})\big)^{-1/2}\Big)\,.
\end{align*}
Based on these two preliminary approximation steps, we obtain that
\begin{align*}&n^{1/3}\max_{k=1,\ldots,h_n^{-1}-1}\frac{\big|m_{k,n}-m_{k-1,n}\big|}{\hat\sigma_{kh_n}}=h_n^{-1/2}\sqrt{2\log(2h_n^{-1}-2)}\max_{k=1,\ldots,h_n^{-1}-1}\frac{\big| m_{k,n}- m_{k-1,n}\big|}{\hat\sigma_{kh_n}}\\
&\quad =h_n^{-1/2}\sqrt{2\log(2h_n^{-1}-2)}\max_{k=1,\ldots,h_n^{-1}-1}\frac{\big|\tilde m_{k,n}-\tilde m_{k-1,n}^*\big|}{\hat\sigma_{kh_n}}+\KLEINO_{\P}(1)\\
&\quad =\sqrt{2\log(2h_n^{-1}-2)}\max_{k=1,\ldots,h_n^{-1}-1}\frac{\big|h_n^{-1/2}\tilde m_{k,n}-h_n^{-1/2}\tilde m_{k-1,n}^*\big|}{\sigma_{kh_n}}+\KLEINO_{\P}(1)\,.
\end{align*}

For any fix $K$, by the independence between statistics on different blocks and Lemma \ref{lem2}, the vector $h_n^{-1/2}(\tilde m_{k,n}-\tilde m_{k-1,n}^*)_{1\le k\le K}$ converges in distribution to a vector $(|U_k|-|U_{k-1}^*|)_{1\le k\le K}$, with two sequences of independent normally distributed random variables $(U_k)$ and $(U_{k}^*)$, with $U_k$ and $U_{k}^*$ correlated only for the same index $k$. 

Using continuous mapping and the Skorokhod representation we derive that
\begin{align*}n^{1/3}\,T^{BHR}(Y_0,Y_1\ldots,Y_n)&=n^{1/3}\,\max_{k=1,\ldots,h_n^{-1}-1}\Big|\frac{m_{k,n}-m_{k-1,n}}{\hat\sigma_{kh_n}}\Big|\\
&=\sqrt{2\log(2h_n^{-1}-2)} \max_{k=1,\ldots,h_n^{-1}-1}\big||U_k|-|U_{k-1}^*|\big|+\KLEINO_{\P}(1)\,,\end{align*}
where $(U_k)_{k=1,\ldots,h_n^{-1}-1}$ and $(U_k^*)_{k=1,\ldots,h_n^{-1}-1}$ are i.i.d.\ sequences of standard normally distributed random variables, with $U_k$ and $U_{j}^*$ independent for $k\ne j$. The Gumbel convergence shown in Proposition \ref{gumbelprop} generalizes from an i.i.d.\ to a 1-dependent non-i.i.d.\ sequence as shown in \cite{watson1954}. As $h_n^{-1}\to\infty$, we can thus apply Proposition \ref{gumbelprop} replacing $n$ by the number of differences between blocks, $h_n^{-1}-1$. This proves \eqref{gumbel}.

\subsection{Proofs of consistency of the tests\label{sec:proofcons}}
We are left to prove the consistency of the tests, \eqref{consist} and \eqref{consistloc}. Under the alternative hypothesis, there is some $k^*\in\{1,\ldots,h_n^{-1}-2\}$ with $\theta\in (k^*h_n,(k^*+1)h_n)$, such that $X_t$ admits a jump at time $\theta\in(0,1)$: $|\Delta X_{\theta}|=|\Delta J_{\theta}|=|J_{\theta}-J_{\theta-}|>0$. By standard bounds for the jump component, we have that
\begin{align*}
\min_{i\in\mathcal{I}_{k^*}^n} Y_i&=\min_{i\in\mathcal{I}_{k^*}^n} \big(J_{t_i^n}+C_{t_i^n}+\epsilon_i\big)\\
&=\min\Big(\min_{i:t_i^n\in (k^*h_n,\theta)} \big(J_{t_i^n}+C_{t_i^n}+\epsilon_i\big),\min_{i:t_i^n\in [\theta,(k^*+1)h_n)} \big(J_{t_i^n}+C_{t_i^n}+\epsilon_i\big)\Big)\\
&=\min\Big(\Big(J_{\theta-}+\min_{i:t_i^n\in (k^*h_n,\theta)} \big(C_{t_i^n}+\epsilon_i\big)\Big),\Big(J_{\theta}+\min_{i:t_i^n\in [\theta,(k^*+1)h_n)} \big(C_{t_i^n}+\epsilon_i\big)\Big)\Big)+\KLEINO_{\P}\big(h_n^{1/2}\big)\,,
\end{align*}
where the remainder is due to possible additional jumps on $(k^*h_n,(k^*+1)h_n)$. We obtain the elementary lower bound 
\begin{align*}
\min_{i\in\mathcal{I}_{k^*}^n} Y_i\ge \min\big(J_{\theta-},J_{\theta}\big)+\min_{i\in\mathcal{I}_{k^*}^n}\big(C_{t_i^n}+\epsilon_i\big)+\KLEINO_{\P}\big(h_n^{1/2}\big)\,,
\end{align*}
and the upper bound
\begin{align*}
\min_{i\in\mathcal{I}_{k^*}^n} Y_i&\le \min\big(J_{\theta-},J_{\theta}\big)+\max\Big(\min_{i:t_i^n\in (k^*h_n,\theta)} \big(C_{t_i^n}+\epsilon_i\big),\min_{i:t_i^n\in [\theta,(k^*+1)h_n)} \big(C_{t_i^n}+\epsilon_i\big)\Big)+\KLEINO_{\P}\big(h_n^{1/2}\big)\\
&=\min\big(J_{\theta-},J_{\theta}\big)+\min_{i\in\mathcal{I}_{k^*}^n}\big(C_{t_i^n}+\epsilon_i\big)+\mathcal{O}_{\P}\big(h_n^{1/2}\big)\,,
\end{align*}
since we know that the difference  between the two minima in the maximum is $\mathcal{O}_{\P}\big(h_n^{1/2}\big)$. In case that $\Delta J_{\theta}>0$, we obtain that
\begin{align*}\min_{i\in\mathcal{I}_{k^*+1}^n} Y_i-\min_{i\in\mathcal{I}_{k^*}^n} Y_i&=\Delta J_{\theta}+\min_{i\in\mathcal{I}_{k^*+1}^n}\big(C_{t_i^n}+\epsilon_i\big)-\min_{i\in\mathcal{I}_{k^*}^n}\big(C_{t_i^n}+\epsilon_i\big)+\mathcal{O}_{\P}\big(h_n^{1/2}\big)\\
&=\Delta J_{\theta}+\mathcal{O}_{\P}\big(h_n^{1/2}\big)\,,
\end{align*}
while for $\Delta J_{\theta}<0$, we obtain that 
\begin{align*}\min_{i\in\mathcal{I}_{k^*}^n} Y_i-\min_{i\in\mathcal{I}_{k^*-1}^n} Y_i&=\Delta J_{\theta}+\min_{i\in\mathcal{I}_{k^*}^n}\big(C_{t_i^n}+\epsilon_i\big)-\min_{i\in\mathcal{I}_{k^*-1}^n}\big(C_{t_i^n}+\epsilon_i\big)+\mathcal{O}_{\P}\big(h_n^{1/2}\big)\\
&=\Delta J_{\theta}+\mathcal{O}_{\P}\big(h_n^{1/2}\big)\,.
\end{align*}
Under the alternative hypothesis in Theorem \ref{thmgumbel}, we thus have that
\begin{align*}
n^{1/3}\max_{k=1,\ldots,h_n^{-1}-1}\Big|\frac{m_{k,n}-m_{k-1,n}}{\hat\sigma_{kh_n}}\Big|&\ge n^{1/3} \frac{|\Delta X_{\theta}|}{\sigma_{\theta}}\big(1+\KLEINO_{\P}(1)\big)+\mathcal{O}_{\P}\big(n^{1/3}h_n^{1/2}\big)\,,
\end{align*}
with $\theta\in(0,1)$, for which $\liminf_{n\to\infty} n^{\beta}|\Delta X_{\theta}|>0$, for some $\beta<1/3$. The consistency of the global test, \eqref{consist}, now follows from
\begin{align*}
n^{1/3}\max_{k=1,\ldots,h_n^{-1}-1}\Big|\frac{m_{k,n}-m_{k-1,n}}{\hat\sigma_{kh_n}}\Big|&\ge n^{1/3} \frac{|\Delta X_{\theta}|}{\sup_{t\in[0,1]}\sigma_{t}}\big(1+\KLEINO_{\P}(1)\big)+\mathcal{O}_{\P}\Big(\sqrt{\log(2h_n^{-1})}\Big)\\
&= n^{1/3-\beta} \frac{n^{\beta}|\Delta X_{\theta}|}{\sup_{t\in[0,1]}\sigma_{t}}\big(1+\KLEINO_{\P}(1)\big)+\mathcal{O}_{\P}\Big(\sqrt{\log(2h_n^{-1})}\Big)\\
&\stackrel{\P}{\longrightarrow} \infty\,,
\end{align*}
since $1/3-\beta>0$ and $\liminf_{n\to\infty} n^{\beta}|\Delta X_{\theta}|>0$. This completes the proof of Theorem \ref{thmgumbel}.

The consistency of the local test, \eqref{consistloc}, follows with similar considerations:
\begin{align*}
h_n^{-1/2}\Big|\frac{\hat X_{\tau}}{\hat\sigma_{\tau+}}-\frac{\hat X_{\tau-}}{\hat\sigma_{\tau-}}\Big|&=h_n^{-1/2}\Big|\frac{J_{\tau}}{\hat\sigma_{\tau+}}-\frac{J_{\tau-}}{\hat\sigma_{\tau-}}\Big|-\mathcal{O}_{\P}(1)\\
&=h_n^{-1/2}\Big(\Big|\frac{J_{\tau}\sigma_{\tau-}-J_{\tau-}\sigma_{\tau}}{\sigma_{\tau}\sigma_{\tau-}}\Big|+\KLEINO_{\P}(1)\Big)-\mathcal{O}_{\P}(1)\\
&\stackrel{\P}{\longrightarrow} \infty\,,
\end{align*}
since $(\sigma_{t})$ is uniformly bounded and $|{J_{\tau}\sigma_{\tau-}-J_{\tau-}\sigma_{\tau}}|>0$ under the alternative hypothesis. Since the asymptotic level $\alpha$ of the test readily follows from Theorem \ref{thmloc}, Corollary \ref{corollary} is proved.
%\end{proof}
\subsection{Proof of consistent localization and online test}
\begin{proof}[Proof of Proposition \ref{localize}:]
Using analogous notation as in the proof of consistency of the test in Section \ref{sec:proofcons}, suppose the jump at time $\theta\in(k^*h_n,(k^*+1)h_n)$, with some $k^*\in\{1,\ldots,h_n^{-1}-2\}$. We bound the expected absolute deviation of 
\[\hat\theta_n=h_n\operatorname{argmax}_{1\le k\le h_n^{-1}-1}\frac{|m_{k,n}-m_{k-1,n}|}{\hat\sigma_{kh_n}}\]
by
\begin{align*}
\E\big[|\hat\theta_n-\theta|\big]&=\E\big[|\hat\theta_n-\theta|\1(h_n^{-1}\hat\theta_n\in\{k^*,k^*+1\})\big]+\E\big[|\hat\theta_n-\theta|\1(h_n^{-1}\hat\theta_n\notin\{k^*,k^*+1\})\big]\\
&\le 2h_n+2\P\big(h_n^{-1}\hat\theta_n\notin\{k^*,k^*+1\}\big)\,,
\end{align*}
where we use that $\theta,\hat\theta_n\in[0,1]$ with the triangle inequality. The uniform upper bound $|\sigma_s|\le K$, obtained at the beginning of Section \ref{subsec:pre}, with the uniform consistency of $(\hat\sigma_{kh_n})_{1\le k\le h_n^{-1}-1}$, yield an almost sure inequality
\[\max_{k\in\{k^*,k^*+1\}}\frac{|m_{k,n}-m_{k-1,n}|}{\hat\sigma_{kh_n}}\ge \frac{1}{2K}\max_{k\in\{k^*,k^*+1\}}|m_{k,n}-m_{k-1,n}|\,.\]
The factor 1/2 in the inequality can be replaced by any constant strictly smaller than one. The approximations in the proof of consistency of the test in Section \ref{sec:proofcons} show that
\[\max_{k\in\{k^*,k^{*}+1\}}|m_{k,n}-m_{k-1,n}|=|J_{\theta}-J_{\theta-}|+\mathcal{O}_{\P}(h_n^{1/2})\,.\]
The set on which $h_n^{-1}\hat\theta_n\notin\{k^*,k^*+1\}$ is hence almost surely a subset of $\{T^{H_0}\ge \delta\}$, with 
\[T^{H_0}=\max_{1\le k\le h_n^{-1}-1}\frac{|\min_{i\in\mathcal{I}_{k}^n}\big(C_{t_i^n}+\epsilon_i\big)-\min_{i\in\mathcal{I}_{k-1}^n}\big(C_{t_i^n}+\epsilon_i\big)|}{\hat\sigma_{kh_n}}\,,\]
and $\delta=|J_{\theta}-J_{\theta-}|/(3K)$, where the factor 1/3 can be replaced by any constant strictly smaller than the factor 1/2 from above. On this event, the maximum absolute difference of local minima, based on observations without jumps, exceeds a positive threshold that does not tend to zero as $n\to\infty$. This maximum, $T^{H_0}$, however, behaves asymptotically as our test statistic under the null hypothesis of no jumps. We proved in Theorem \ref{thmgumbel} that, $T^{H_0}=\mathcal{O}_{\P}\big(n^{-1/3}\log(h_n^{-1})\big)=\KLEINO_{\P}\big(n^{-1/3}\log(n)\big)$, and, moreover, the normalized tight sequence converges weakly to a Gumbel limit distribution. Since
\begin{align*}
\P\big(h_n^{-1}\hat\theta_n\notin\{k^*,k^*+1\}\big)\le \P\big(n^{1/3}T^{H_0}-B_n\ge n^{1/3}\delta-B_n\big)\,,
\end{align*}
and since Gumbel convergence is equivalent to an exponential decay of tail probabilities, see Eq.\ (7) of \cite{arxiv2}, the probability right-hand side tends to zero faster than any polynomial decay $n^{-s}$, $s>0$. Therefore, we conclude based on the above bound that
\[|\hat\theta_n-\theta|=\mathcal{O}_{\P}(h_n)\,.\]
\end{proof}
\begin{proof}[Proof of Corollary \ref{onlinecor}:]
The quantiles of the asymptotic Gumbel distribution of $T^{BHR}$ in Theorem \ref{thmgumbel} are readily determined from \eqref{gumbelex}: $q_{1-\alpha}^{\Lambda}=-\log(-\log(1-\alpha))$. The subset relation \eqref{subsetineq} for minima grants for $i\in\mathcal{I}_k^n$ that
\[Y_i\ge \min_{j\in\mathcal{I}_k^n,j\le i} Y_j\ge \min_{j\in\mathcal{I}_k^n}Y_j=m_{k,n}\,,\]
where the first inequality will be an equality to the running minimum on the $k$th block for an observation $Y_i$, which triggers an online jump detection (since previous observations have then not exceeded the threshold). Since the threshold criterion \eqref{onlineeq} yields that
\begin{align*}\frac{m_{k,n}-m_{k-1,n}}{\hat\sigma_{kh_n}}\le \frac{Y_i-m_{k-1,n}}{\hat\sigma_{kh_n}}&<n^{-1/3}\big(\log(-\log(1-\alpha))-B_n\big)\\ &=-n^{-1/3}(B_n+q_{1-\alpha}^{\Lambda})\,,
\end{align*}
by Theorem \ref{thmgumbel} the global test detects a negative jump. By the same considerations as in the previous proof of Proposition \ref{localize}, it is located to the interval $(kh_n,t_i^n)$.
\end{proof}
\section*{Acknowledgement}
The authors are grateful to two anonymous reviewers and editors for many helpful comments. Markus Bibinger is grateful to Adrian Grüber for helpful comments on some proofs.
%\section*{Funding sources}
%This work was supported by the Deutsche Forschungsgemeinschaft (DFG) under grant 403176476.
%\section*{Supplementary material}
%A web appendix with additional simulations, more empirical illustrations and results from the data analysis and the underlying extreme value theory is available. A previous version of Supplement S.4 is public as a note \cite{arxiv2}. 
\bibliographystyle{apalike}
\bibliography{literatur}
\addcontentsline{toc}{section}{References}
\clearpage
\appendix
\section{Additional simulation results\label{S1}}
\subsection{Performance of the test for i.i.d.\ noise and equispaced observations}
We present results on the size and power of the global test in a standard setup with one-sided i.i.d.\ noise and equispaced observations. As in Section \ref{sec:3}, we simulate the efficient log-price process $(X_t)$ based on the model
\begin{align*}
	dX_t &= v_t \sigma_t dW_t\\
	d\sigma_t^2 &= 0.0162 \cdot (0.8465 -\sigma_t^2)dt + 0.117\cdot \sigma_tdB_t\\
	v_t &= (1.2-0.2\cdot \sin(3/4 \pi t)) \cdot 0.01 \quad\text{with}\quad t\in[0,1].
\end{align*}
Under the alternative, a jump is added at a random time point, however neither close to the beginning nor to the end of the sampled trajectory. Jumps are with equal probabilities positive or negative and their absolute sizes are given in Table \ref{tab:most_globaln}. We simulate $n = 23{,}400$ discrete, equispaced, noisy observations generated by
\begin{align}
	Y_{i}&= X_{i/n} + q \varepsilon_{i}  \quad\text{with}\quad\varepsilon_{i}\overset{iid}{\sim}Exp(1),\quad i=0,\ldots,n,
\end{align}
where values of the noise level $q$ are shown in Table~\ref{tab:most_globaln}. Exponentially distributed noise is a standard choice for the one-sided noise model and has been considered in the simulations of \cite{bibinger2022} on spot volatility estimation also.

\begin{table}[h]
\centering
	\begin{tabular}{c|ccccccc}
		\hline\hline
			& & \multicolumn{6}{c}{$\vert \text{jump size}\vert$} \\
		\hline
		$q$	& $0.00\%$ & $0.10\%$ & $0.15\%$ & $0.20\%$& $0.25\%$ & $0.30\%$ & $0.50\%$\\
		\hline
		$ 0.010\%$ & 0.05 & 0.16 & 0.59 & 0.91 & 0.99 & 1.00 & 1.00\\
		$ 0.025\%$ & 0.05 & 0.15 & 0.56 & 0.90 & 0.99 & 1.00 & 1.00\\
		$ 0.050\%$ & 0.05 & 0.14 & 0.54 & 0.89 & 0.98 & 1.00 & 1.00\\
		$ 0.075\%$ & 0.05 & 0.13 & 0.52 & 0.87 & 0.98 & 1.00 & 1.00\\
		$ 0.100\%$ & 0.05 & 0.13 & 0.50 & 0.85 & 0.97 & 1.00 & 1.00 \\
		%$ 0.250\%$ & 0.05 & 0.11 & 0.43 & 0.80 & 0.95 & 0.99 & 1.00\\	
		\hline\hline
	\end{tabular}
	\caption{Simulation results for the global test on a significance level of $\alpha=5\%$. The column $\vert \text{jump size}\vert=0.00\%$ gives the estimated size and the columns with $\vert \text{jump size}\vert>0\%$ give the estimated power for the corresponding jump sizes.}
	\label{tab:most_globaln}
\end{table}
The global test is based on the statistic $T^{BHR}$ from \eqref{bhrgl}, i.e.\ the maximal absolute difference between block-wise local minima normalized with the estimated spot volatility. The spot volatility estimator from \eqref{simpleestimator2} is computed with $K_n = 200$ and $n h_n = 30$. As in Section \ref{sec:3}, we account for the finite-sample bias of the spot volatility estimator for this tuning multiplying it with the correction factor $0{.}954$ based on Sec.\ 5.1 of \cite{bibinger2022}. The global Gumbel test is based on Theorem \ref{thmgumbel}.

\begin{figure}[t]
	\begin{minipage}{1\textwidth}
     		\centering
        		\includegraphics[width=6.5cm]{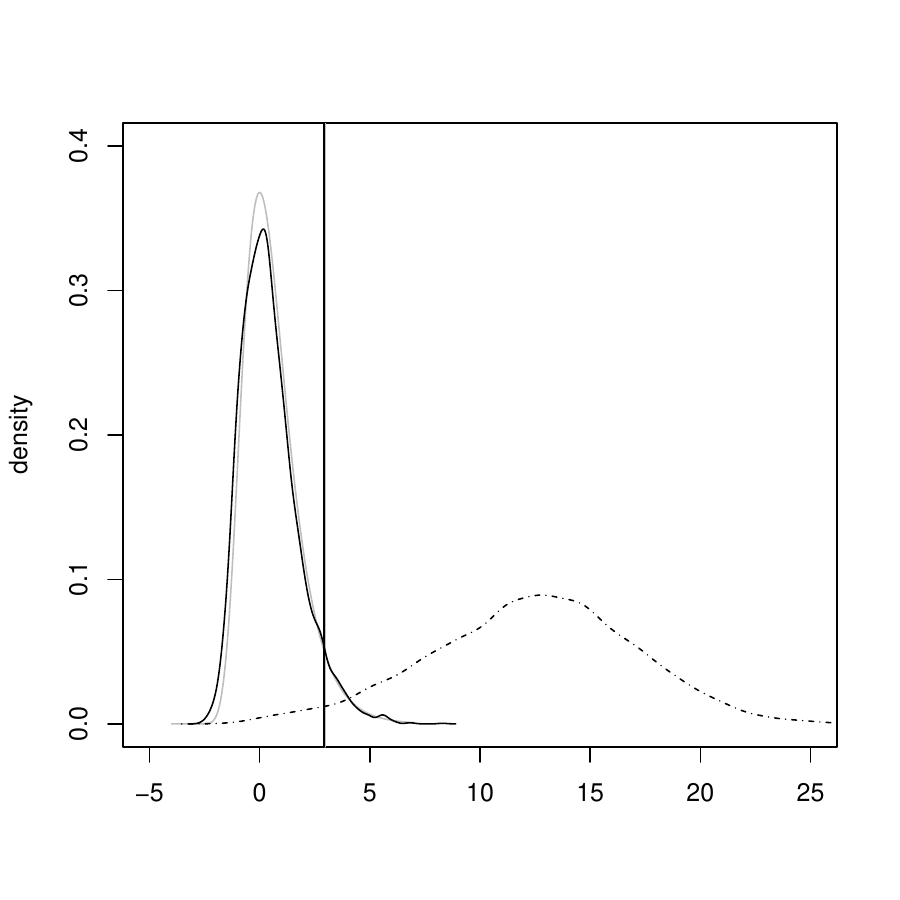}
		\includegraphics[width=6.5cm]{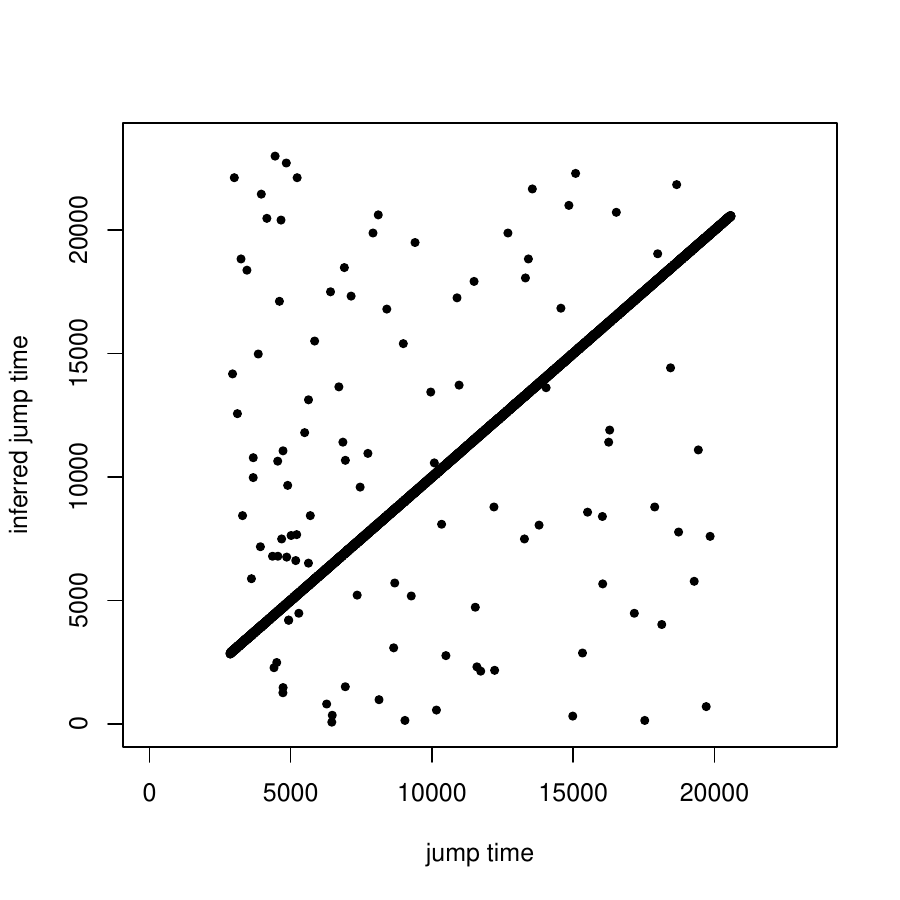}
    	\end{minipage}
	\caption{\label{Fig:densityn}Left panel: Kernel density estimates of the standardized version of the test statistic $T^{BHR}$ under the null (black solid line) and under the alternative (dashed-dotted line) for $\vert \text{jump size}\vert = 0.25\%$ and noise level $q = 0.1\%$. The gray solid line depicts the density function of the standard Gumbel distribution. The black and gray vertical lines (almost indistinguishable) are the $95\%$-quantile of the standardized version of the test statistic $T^{BHR}$ under the null and the $95\%$-quantile of the standard Gumbel distribution. Right panel: Jump times vs.\ inferred jump times for |jump size| = 0.25\% and noise level $q = 0.1\%$.}
\end{figure}

Table~\ref{tab:most_globaln} gives the size ($\vert \text{jump size}\vert=0\%$) and power ($\vert \text{jump size}\vert>0\%$ ) of the global test for different noise levels $q$. Under the null hypothesis, when there is no jump, the test keeps its level. The power of the test decreases for increasing noise level and some fix jump size. It increases for growing |jump size|. Compared to Table \ref{tab:most_global}, the power is higher for the same noise levels and jump sizes. This is, however, mainly due to the reduced effective sample sizes for the noise specification in Section \ref{sec:sim_global}. The additional noise autocorrelations in Section \ref{sec:sim_global} instead hardly reduce the power. For $\vert \text{jump size}\vert = 0.25\%$ and $q = 0.1\%$, the performance of the test is illustrated in Figure~\ref{Fig:densityn} which shows kernel density estimates of the standardized version of the statistic $T^{BHR}$ under the null and under the alternative. The density function of the standard Gumbel distribution is plotted to compare the empirical distribution under the null hypothesis to this theoretical limit distribution. There are minor deviations in the center of the two distributions, but the right tail of the distribution is quite accurately fitted. In the right panel of Figure~\ref{Fig:densityn}, the true jump times are plotted against the inferred jump times determined by the estimator from Proposition \ref{localize}. Since most of the points are on the main diagonal, the jumps are correctly located in most cases. That more points are off the diagonal in the morning than around lunch time can be explained by the higher volatility which makes jump detection more difficult. Hence, we further conclude that the localization performs slightly worse in intraday periods where the volatility is higher. The plots of Figure \ref{Fig:densityn} are analogous to those in Figure \ref{Fig:density}. Note that we use a slightly smaller jump size $\vert \text{jump size}\vert = 0.25\%$ instead of $\vert \text{jump size}\vert = 0.3\%$ here, with the same noise level, and attain a similar discriminatory power as left-hand side in Figure \ref{Fig:density}. This means that the performance is better in this setup with i.i.d.\ noise which is again mainly due to the larger effective sample size. The accuracy of the localization right-hand side in Figure \ref{Fig:densityn} is as well even better than the one in Figure \ref{Fig:density}.

\subsection{Comparison with standard market microstructure noise}
\begin{table}[t]
\centering
	\begin{tabular}{ccc|ccccccc}
		\hline\hline
			&&& \multicolumn{7}{c}{$\vert \text{jump size}\vert$} \\
		\hline
		$q$ & $nh_n$ & test & 0.000\% &  $0.100\%$ & $0.125\%$ & $0.150\%$ & $0.175\%$ & $0.200\%$ & $0.225\%$\\
		\hline
		\multirow{4}{*}{0.01\%}		& \multirow{2}{*}{3} 	& BHR &0.00 & 0.96 & 1.00 & 1.00 & 1.00 & 1.00 & 1.00\\
							& 				& LM   &0.00 & 0.74 & 0.97 & 1.00 & 1.00 & 1.00 & 1.00\\
							& \multirow{2}{*}{4} 	& BHR &0.01 & 0.96 & 1.00 & 1.00 & 1.00 & 1.00 & 1.00\\
							&  				& LM   &0.00 & 0.75 & 0.94 & 0.99 & 1.00 & 1.00 & 1.00\\
		\hline\hline
	\end{tabular}
	\caption{Simulation results for the global test in Scenario 1 on a significance level of $\alpha=5\%$. The column $\vert \text{jump size}\vert=0.000\%$ gives the estimated size and the columns with $\vert \text{jump size}\vert>0\%$ give the estimated power for the corresponding jump sizes. The value $nh_n$ is the number of noisy observations per interval.}
	\label{tab:global_comparisonn}
\end{table}
This section considers the same scenarios as Section \ref{sec:3.2} and provides additional material. {\textbf{Scenario 1}} compares the baseline models with i.i.d.\ noise and equispaced observations for equal sample sizes, using observations generated by
\begin{align}
	Y_{i}^{A} &= X_{i/n} + (1 - 2 \pi^{-1})^{-1/2} q \vert\varepsilon_{i}\vert \quad\text{with}\quad\varepsilon_{i}\overset{iid}{\sim}\mathcal{N}(0, 1),\label{eq.sim1n}\\
	Z_{i} &= X_{i/n} + q\varepsilon_{i}\label{eq.sim2n}
\end{align}
where $Y_{i}^{A}$ are the simulated ask quotes based on the LOMN-model and $Z_{i}$ are the observations under MMN, $i=0,\ldots,n$. 

{\textbf{Scenario 2}} generates observations from a model adopted from \citet{li2018}, that incorporates additive MMN and additionally rounding errors, with
\begin{align}
	Z_{i} &= \log \left( \exp \left(X_{i/n} + \epsilon_{i}\right)^{(s)} \right)\quad \text{with}\quad\epsilon_{i}\overset{iid}{\sim}\mathcal{N}(0, q^2),\label{eq.sim1altn}
\end{align}
for $i=0,\ldots,n$ and $x^{(s)}:=s \vee ([x/s]s)$ for any $x>0$, such that $x$ is rounded to the nearest multiples of tick size $s=0.01$. In case of lower price levels, e.g., $Z_i \leq \log (10)$, the rounding error is dominating, while in case of higher price levels, e.g., $Z_i \geq \log (50)$, the additive MMN dominates. Samples of simulated ask and bid quotes are generated implicitly setting $Y_{i}^{A}= Z_{i}$, whenever $Z_{i} > X_{i/n}$, and $Y_{i}^{B}= Z_{i}$, when $ Z_{i} < X_{i/n}$. 

The implementation of the test statistic $T^{LM}$ from \cite{lm12} for MMN and the bootstrap procedure follow Eqs.\ \eqref{lm}-\eqref{adboot}.

Table~\ref{tab:global_comparisonn} reports additional results to the ones given in Table \ref{tab:global_comparison} for Scenario 1 with small noise level $q=0.01\%$. Different to the larger noise levels, for $q=0.01\%$, the tests do not keep the size $\alpha$, but reject the null hypothesis (provided that the null is true) even with smaller probabilities. The size is underestimated because of the short block lengths which are again chosen such that the power of the corresponding test is maximized. It can be reported that the size is kept for longer blocks, such as $nh_n = 8$, what has been found to be appropriate for the test statistic $T^{LM}$ in \citet{lm12}. The power of both tests increases with the $\vert\text{jump size}\vert$. In comparison, the power of the test in the LOMN-model (BHR) outperforms the test in the MMN-model (LM).

Figure~\ref{Fig:dens_comparen} illustrates the difference in the power of the tests in Scenario 1 (similarly as Figure \ref{Fig:dens_compare} for Scenario 2) for the special case of $\vert \text{jump size}\vert = 0.20\%$ and $q = 0.1\%$. Even though both plots use different scalings, i.e., the non-standardized test statistics are not directly comparable, the better discriminatory power of the test in the LOMN-framework is evident.
\begin{figure}[t]
	\begin{minipage}{1\textwidth}
     		\centering
        		\includegraphics[width=6.5cm]{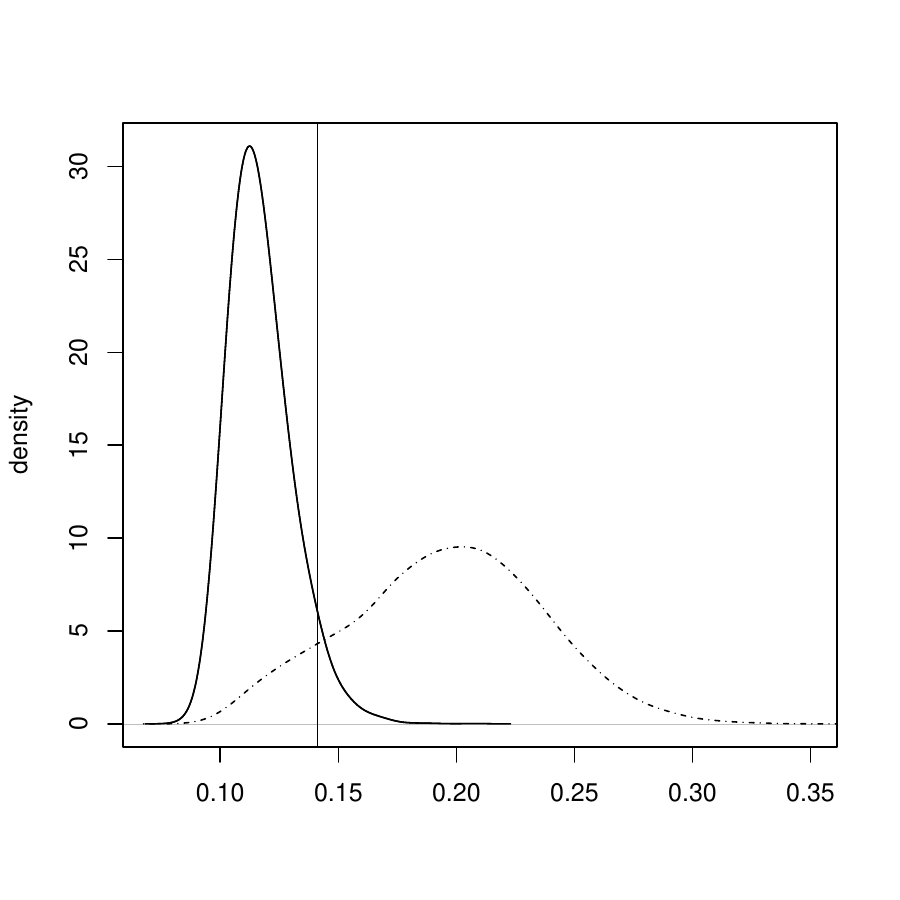}
		\includegraphics[width=6.5cm]{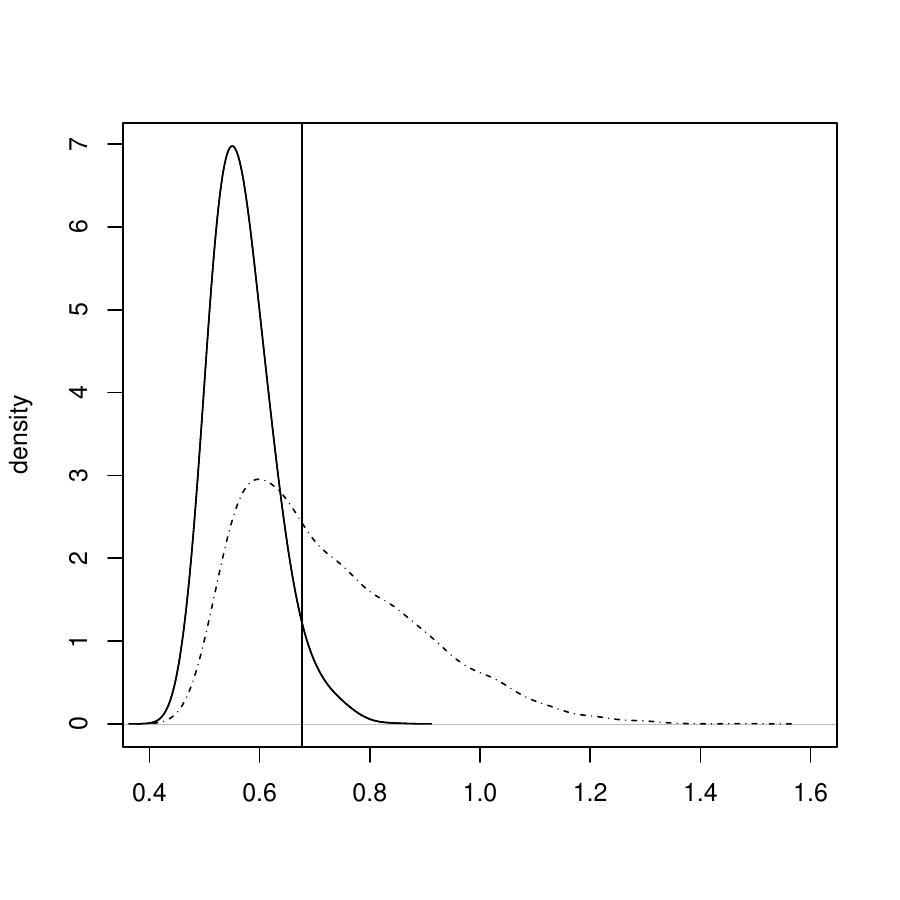}
    	\end{minipage}
	\caption{\label{Fig:dens_comparen}Left panel: Kernel density estimates of the test statistic $T^{BHR}$  in Scenario 1 under the null (black solid line) and under the alternative (dashed-dotted line) for $\vert \text{jump size}\vert = 0.20\%$ and noise level $q = 0.1\%$ for the optimal $nh_n=20$. Right panel: Kernel density estimates of the test statistic $T^{LM}$ under the null (black solid line) and under the alternative (dashed-dotted line) for $\vert \text{jump size}\vert = 0.20\%$ and noise level $q = 0.1\%$ for the optimal $nh_n=34$. The vertical lines refer to the $95\%$-quantile of the respective test statistic under the null. }
\end{figure}

Considering smaller noise levels $q$ in Scenario 2 than in Table \ref{tab:global_comparison2}, e.g., $q=0.01\%$ and $X_0=10$, tantamount to rounding errors which totally dominate the MMN, we report problems in the performance, in particular of the test based on $T^{LM}$ for MMN. While the test under LOMN keeps its size approximately for suitable block lengths, it turns out to be almost impossible to find adequate block lengths for the test based on $T^{LM}$, such that it approximately keeps its size. Because of this size distortion, a comparison in terms of power can be misleading. The root-cause for this low performance of the test based on $T^{LM}$ lies in the misspecified bootstrap that tries to approximate a model with additive noise and dominating rounding errors by a pure additive noise model. As illustrated in the boxplots in Figure~\ref{Fig:box_q}, the parameter $q$ is severely over-estimated in case of dominating rounding errors which partly compensates the ignorance of rounding errors. The positive bias is evidently much more pronounced for smaller price levels, i.e., $X_0 = 10$, than for higher price levels, i.e., $X_0 = 50$. However, local averages are not capable to transform the bootstrap-samples (generated with an over-estimated $q$) into test statistics $\{T^{LM^*}_j\}_{j=1}^m$ whose $95\%$-quantile is a reasonable critical value for the test based on $T^{LM}$. In more details, the smoothed log-returns of a sample stemming from the considered data generating process are more extreme than those generated by the bootstrap relying on additive noise, so that the critical values become too small. In contrast, the $95\%$-quantile of $\{T^{BHR^*}_j\}_{j=1}^m$ for the test under LOMN yields robust results. In a setting with $q=0.01\%$ and $X_0=50$ (when the rounding errors are less dominant), both tests yield reasonable results for $nh_n = 4$ under LOMN for ask and bid quotes and $nh_n = 8$ under MMN. The power is much higher than for $q=0.05\%$ in Table \ref{tab:global_comparison2}. Nevertheless, one conclusion from these simulations is that when rounding errors dominate other noise sources, the methods should be adapted to the different model and without doing so they yield less reliable results. This seems less plausible for best ask and best bid quotes, since it would dilute the \emph{one-sidedness} of the noise structure (seen in Figures \ref{Fig:lob} and \ref{plotlocreal} for data).
\begin{figure}[t]
	\begin{minipage}{1\textwidth}
     		\centering
        		\includegraphics[width=6.5cm]{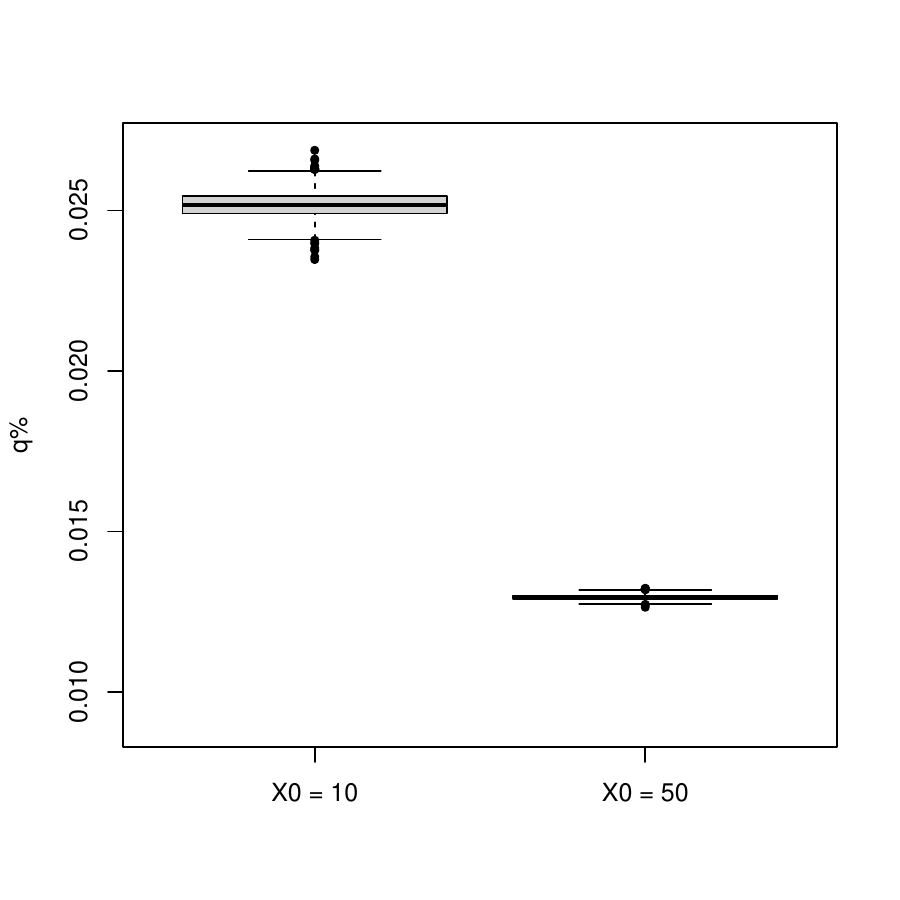}
		\includegraphics[width=6.5cm]{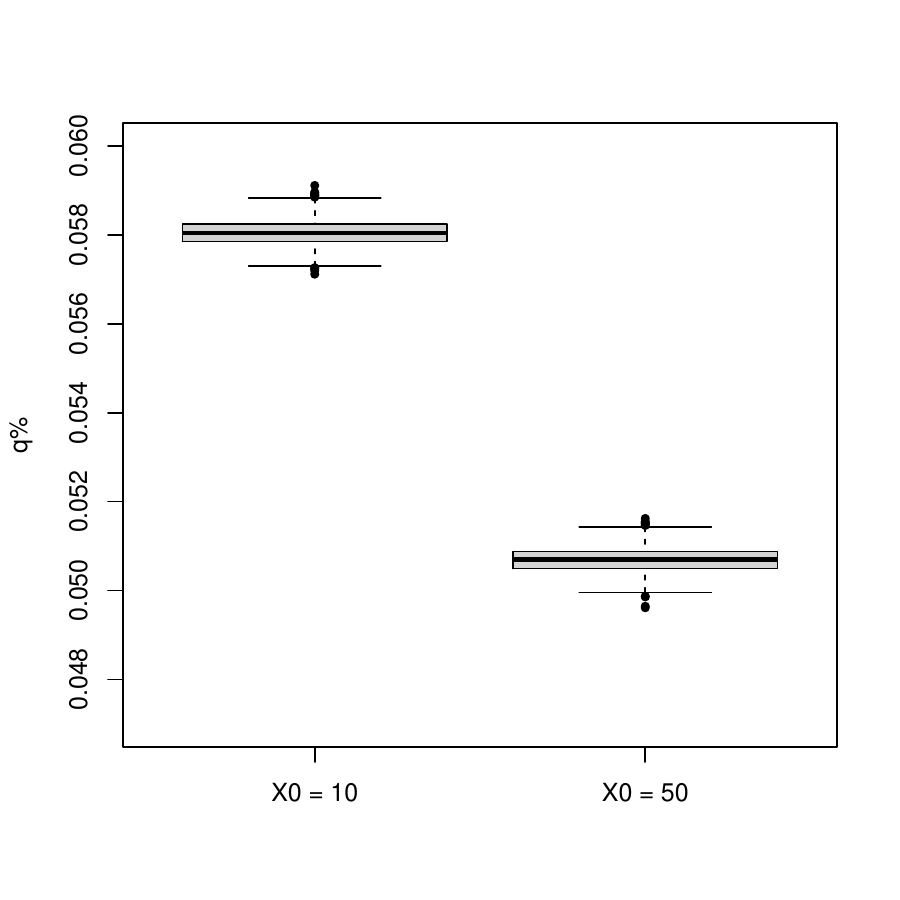}
    	\end{minipage}
	\caption{\label{Fig:box_q}Left panel: Boxplots in Scenario 2 of the estimated $\hat{q}_n$ for $q=0.01\%$. Right panel: Boxplots of the estimated $\hat{q}_n$ for $q=0.05\%$.}
\end{figure}
\clearpage

\section{Additional illustrations of local averages under MMN and local minima and maxima under one-sided noise\label{S2}}
The purpose of Figure \ref{Fig:locmin} with the discussion at the end of Section \ref{sec:2.2} was to illustrate 
\begin{enumerate}
\item the effect of ``pulverization of jumps by pre-averages'' for local averages under MMN;
\item that there is no such effect for local minima in the model with LOMN;
\item the advantage of speed of online detection of a negative jump based on local minima of observations with LOMN.
\end{enumerate}
Here, this discussion is extended to clarify that
\begin{enumerate}
\setcounter{enumi}{3} 
\item the sign of jumps do not affect testing for jumps based on local minima in the model with LOMN;
\item the advantage of speed of online detection of a positive jump holds likewise considering local maxima of observations with one-sided, upper-bounded noise.
\end{enumerate}
To convey the important insights 1.-3., without being diluted too much by finite-sample difficulties of jump detection, Figure \ref{Fig:locmin} showed one rather large (downward) jump, or a diffusive part of $(X_t)$ with rather small volatility, equivalently. In such plots the bars showing local averages of observations with MMN, and local minima of observations with LOMN, are nevertheless subject to some randomness of the diffusive price evolution. Therefore, for simplistic illustrations piecewise constant efficient log-prices are appropriate which are not subject to a random, diffusive behaviour of the efficient log-price. Figure \ref{plotR1} shows such an example. Similar to Figure \ref{Fig:locmin} in Section \ref{sec:2}, four blocks are presented and on the left-hand side observations with i.i.d.\ centered normally distributed MMN, while on the right-hand side observations with i.i.d.\ exponentially distributed LOMN are shown. Different to Figure \ref{Fig:locmin}, there is no simulated $(X_t)$, since it is simply constant up to one jump on the third block. Also different to Figure \ref{Fig:locmin}, the jump size is positive instead of negative. Block-wise, local averages of the observations with MMN are given by the bars in the left plot of Figure \ref{plotR1}. Again, instead of correctly identifying the single jump and its size, the differences between local averages split it into two adjacent jumps of smaller sizes, which illustrates the effect of ``pulverisation of jumps by pre-averages'' discussed in \cite{zhangmykland3}. Again, the differences between local minima of observations with LOMN, given by the bars in the right plot of Figure \ref{plotR1}, correctly suggest one jump and accurately estimate its size. Having an upward jump on the third block, the minimum on this block is taken before the jump. This demonstrates that local minima of observations with LOMN detect positive jumps without a pulverization effect in the same way as negative jumps (4.).

\begin{figure}[t]
\begin{framed}
\centering{\includegraphics[width=5cm]{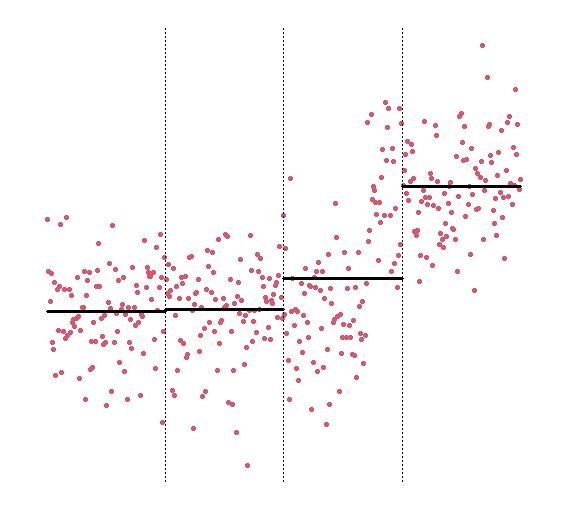}~\includegraphics[width=5cm]{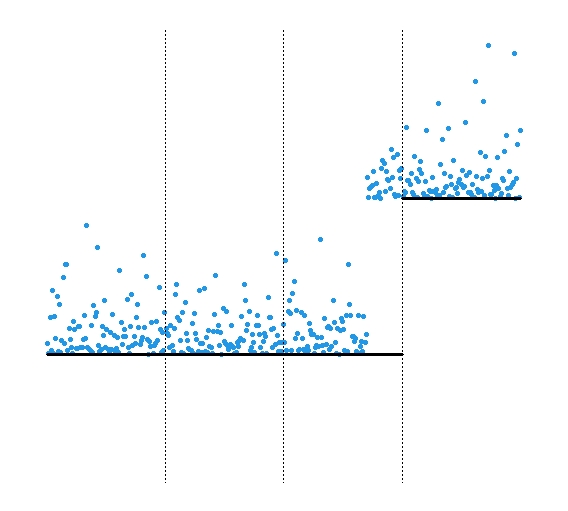}}\end{framed}
\caption{\label{plotR1}Plots similar to Figure \ref{Fig:locmin} with constant $(X_t)$ and upward jump.}
\end{figure}

We demonstrate 5.\ using similar plots as in Figure \ref{Fig:locmin} with a simulated, diffusive efficient log-price $(X_t)$ and one rather large upward jump. Based on real data from a limit order book, local averages are computed from trade prices or mid quotes. For the latter, the model with MMN is standard. Local minima are computed from best ask quotes (or seller-initiated trades) which are modeled with LOMN. Local maxima of bid quotes (or buyer-initiated trades) are also available and modeled in a symmetric way with one-sided, upper-bounded noise.
\begin{figure}[t]
\begin{framed}\includegraphics[width=4.5cm]{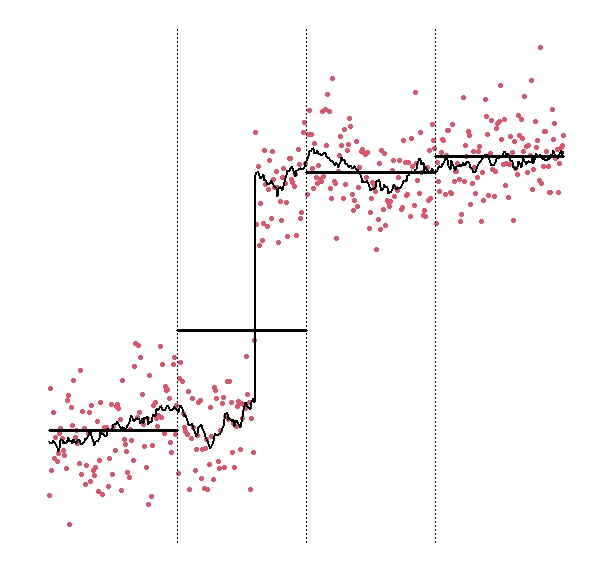}\,\includegraphics[width=4.5cm]{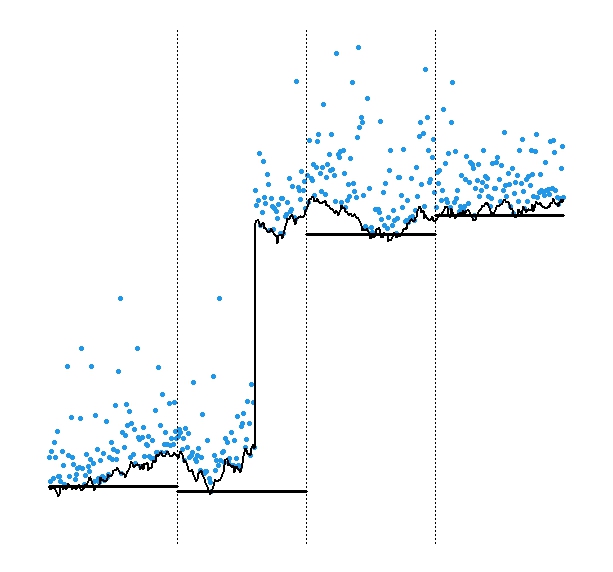}\,\includegraphics[width=4.5cm]{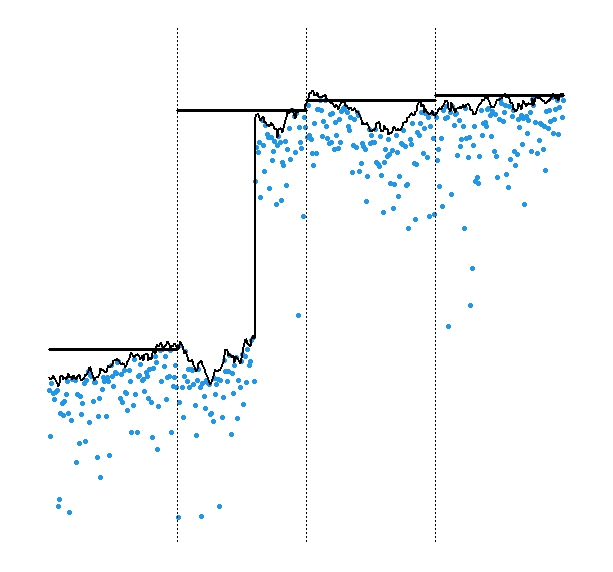}\end{framed}
\caption{\label{plotR2}Plots similar to Figure \ref{Fig:locmin} with upward jump for three noise specifications.}
\end{figure}
Figure \ref{plotR2} illustrates an efficient log-price (the same in all three plots) with an upward jump observed with MMN (left), LOMN as a model for ask quotes (middle), and one-sided, upper-bounded noise as a model for bid quotes (right). Accordingly, local averages (left), local minima (middle) and local maxima (right) of the observations are computed and given by the bars. Noise variances are the same in all three plots. The plots confirm that local maxima and local minima under one-sided noise both detect the jump without a pulverization effect. Since the maximum on a block is always greater or equal than the running maximum, an upward jump is detected as soon as the distance from an observation to the local maximum on the previous block exceeds a threshold. For the illustrated large positive jump, the local maxima of bid quotes hence detect the jump almost instantaneously after its occurrence, while the other statistics based on mid and ask quotes react with a delay of approximately one block length. This illustrates the advantage of speed based on bid quotes for an upward jump (5.). In the illustrations, the block lengths under MMN and one-sided noise are identical. However, the theory for MMN requires larger blocks of lengths proportional to $n^{-1/2}$ than for one-sided noise. Therefore, the delay in online jump detection will typically be more pronounced for local averages under MMN.

\begin{figure}[ht]
\begin{framed}
\centering{\includegraphics[width=13cm]{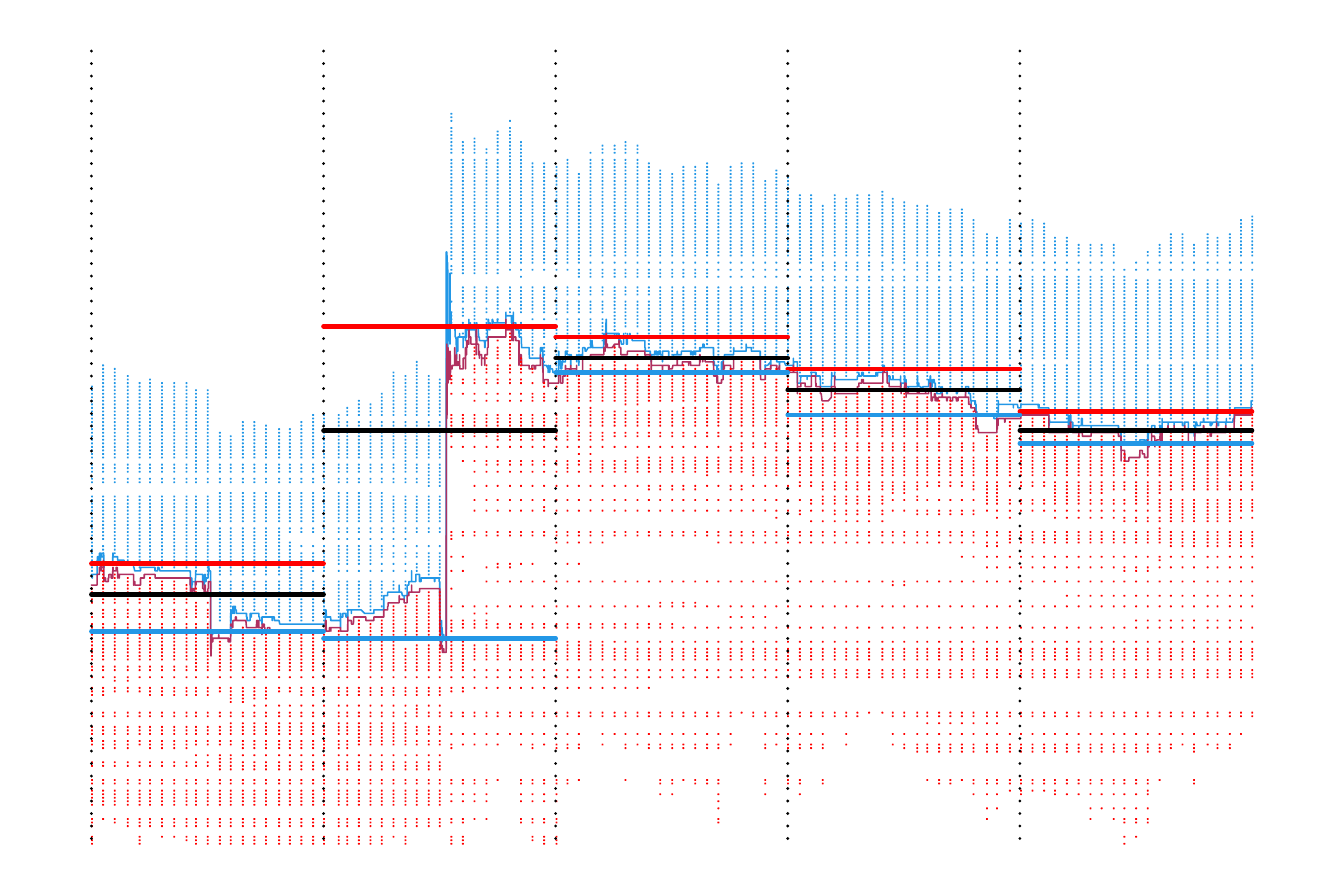}}\end{framed}
\vspace*{0.5cm}
\caption{\label{plotlocreal}Order quotes for AAPL on 12/20/2024 around 03:55pm over 5 blocks, each 10 seconds, with local averages of mid quotes (black bars), local minima of ask quotes (blue) and local maxima of bid quotes (red).}
\end{figure}
\begin{figure}[ht]
\begin{framed}
\centering{\includegraphics[width=5cm]{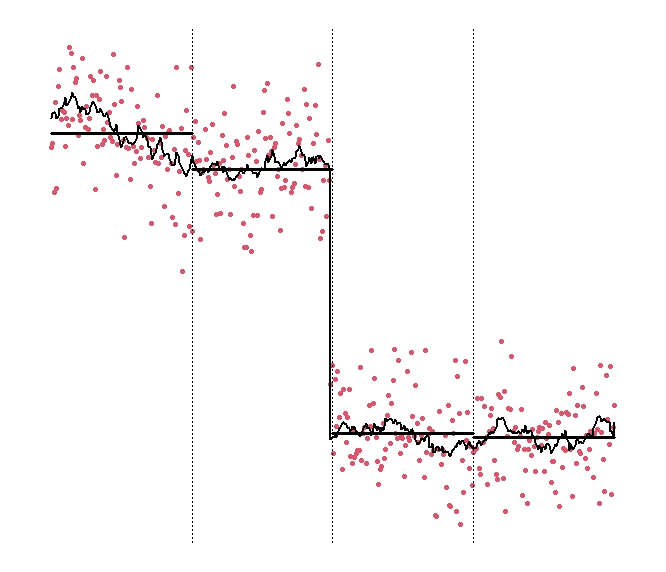}~~~~\includegraphics[width=5cm]{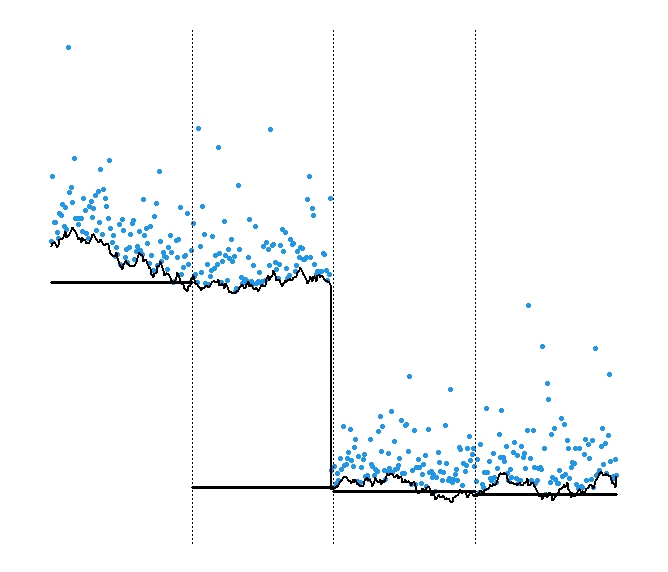}}\end{framed}
%\centering{\includegraphics[width=5cm]{locavgDownAlt.jpeg}~\includegraphics[width=5cm]{locminDownAlt.jpeg}}\end{framed}
\caption{\label{plotR3}Plots similar to Figure \ref{Fig:locmin} with shifted jump time shortly before the end of the second block.}
\end{figure}
\enlargethispage*{1cm}
Figure \ref{plotlocreal} shows a 50 seconds interval of AAPL quotes on 12/20/2024, which is split in 5 blocks, with a positive jump on the second block. Only a subsets of the available order quotes are plotted as data points to get a meaningful illustration. Bid quotes are given by red points and ask quotes by blue points. We consider local averages of mid quotes (black bars), local minima of ask quotes (blue bars) and local maxima of bid quotes (red bars). Although the illustration is more complex than Figures 2, \ref{plotR1} and \ref{plotR2}, the same effects 1., 2., 4.\ and 5.\ are evident.

We complement our discussion of pulverization and speed advantage effects with an additional illustration showing what happens when a jump occurs near the boundary between two blocks. Figure \ref{plotR3} presents a plot similar to Figure \ref{Fig:locmin}, but with the jump shifted to occur shortly before the end of the second block. In the left panel, the bars show local averages of the data with MMN (red). This situation is beneficial compared to the one in Figure \ref{Fig:locmin} in the sense that the jump size is now more accurately estimated from differences of local averages, i.e., the pulverization effect is mitigated. However, there is an adverse effect on the speed of online jump detection. The local average on the second block is mainly determined by observations before the jump, so its difference from the previous one does not exceed the threshold. As a result, the jump is detected only at the end of the third block, with a delay of more than one block length. In the right panel, the bars show local minima of the data with LOMN (blue). Here, the jump located just to the left of the second vertical dashed line is detected almost instantaneously. The speed advantage of local minima is even more pronounced than in Figure \ref{Fig:locmin}. In conclusion, the position of blocks relative to jump times strongly influences the severity of the pulverization effect, but the speed advantage of local order statistics is preserved in all cases. 
 
\clearpage
\section{Data characteristics: acfs\label{S3}}
\begin{figure}[t]
\centering{\includegraphics[width=5cm]{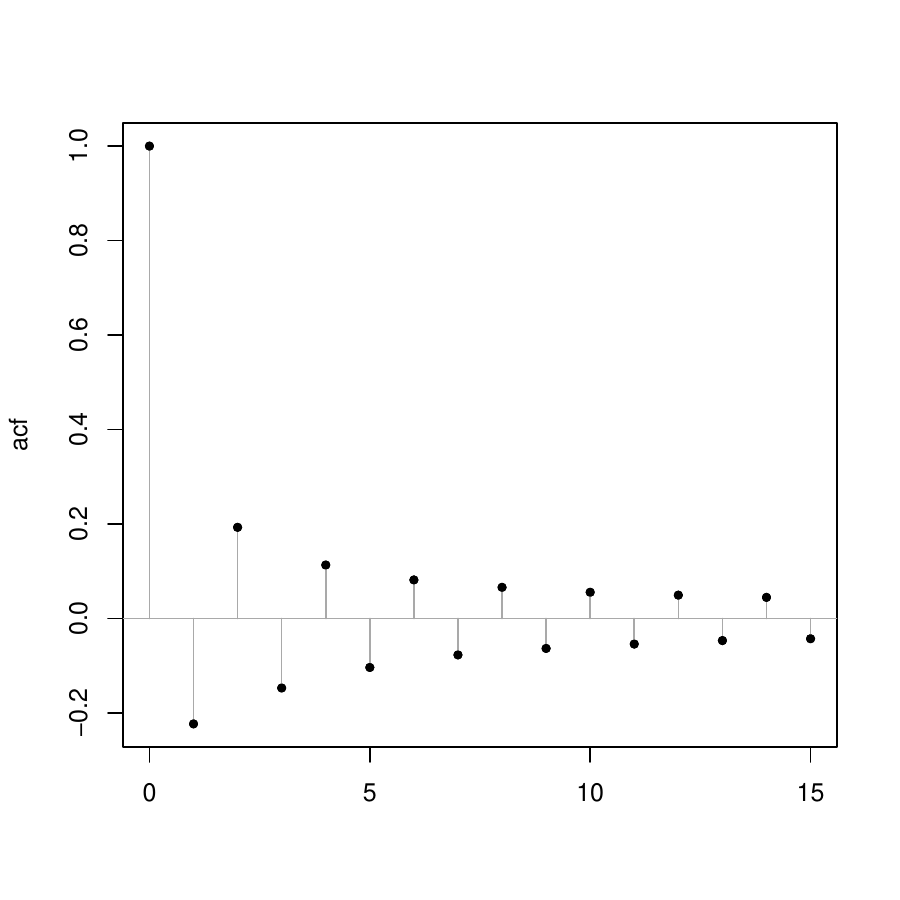}~\includegraphics[width=5cm]{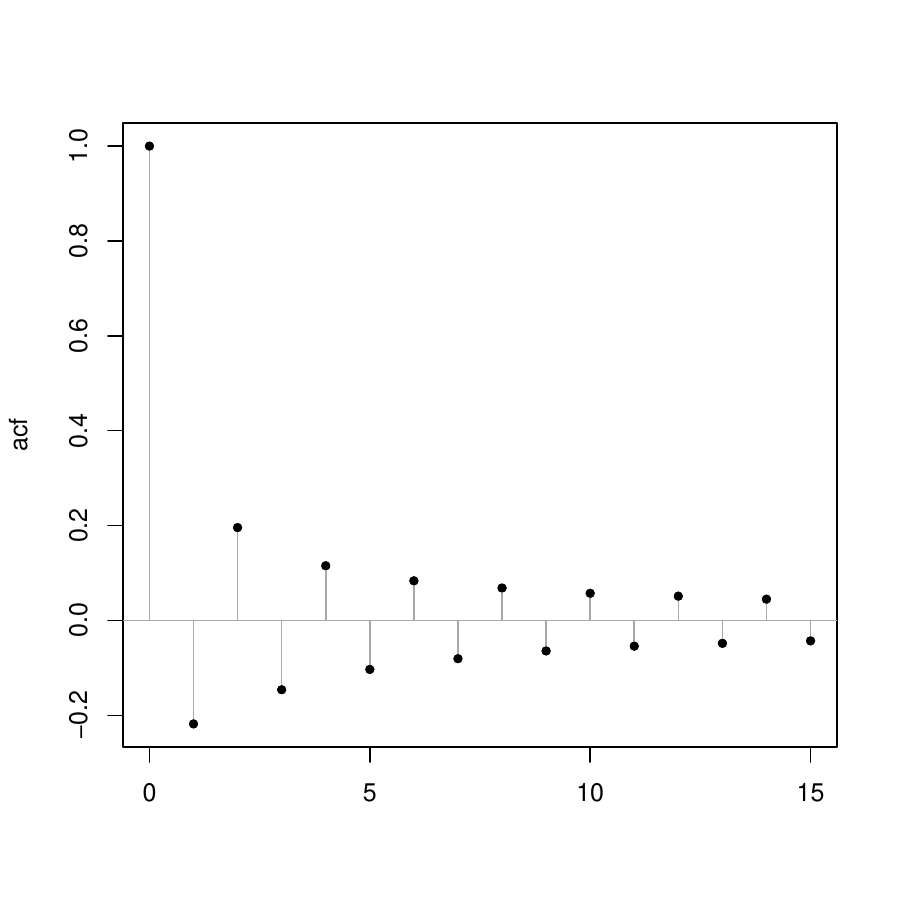}}\\
\centering{\includegraphics[width=5cm]{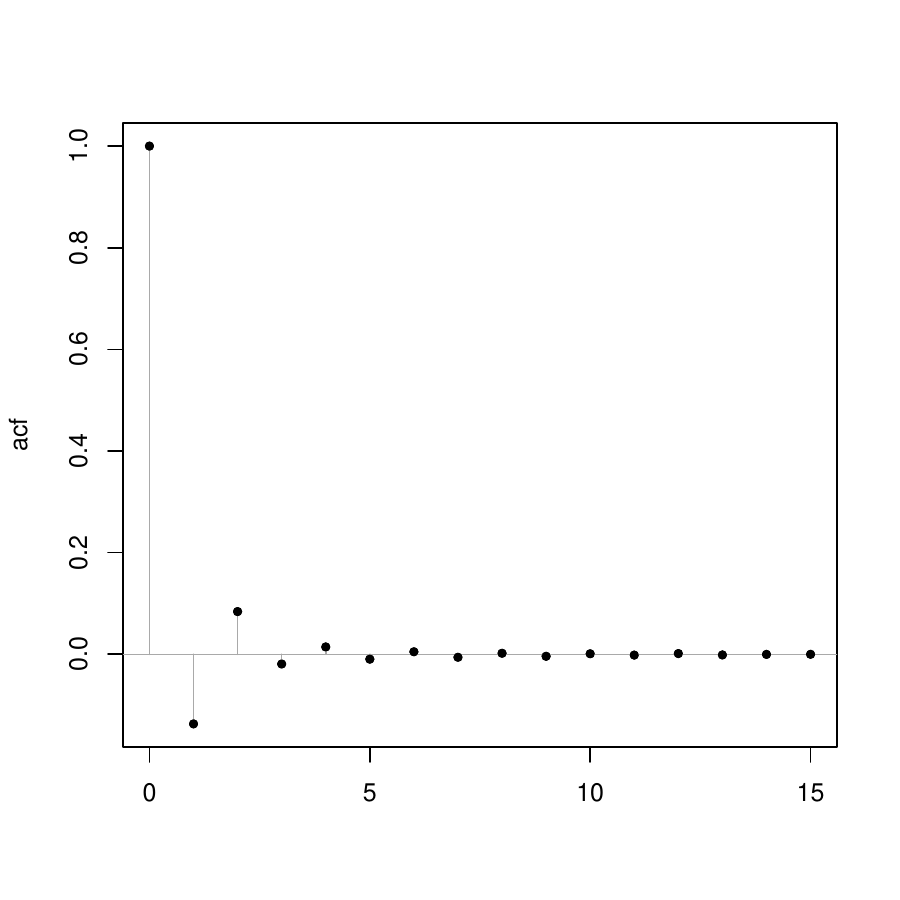}~\includegraphics[width=5cm]{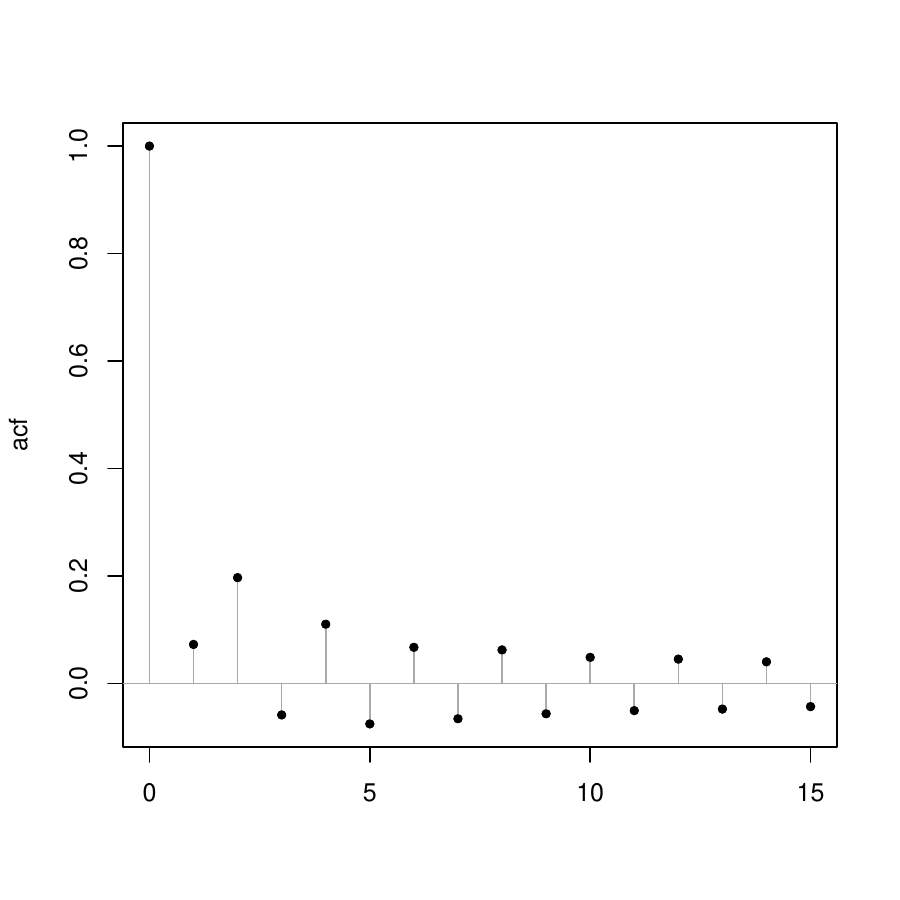}}
\caption{\label{plotACF}Medians of intra-daily acfs of data example JPM, July 2007 to September 2009.\\ Top: Best ask quotes (left) and best bid quotes (right).\\ Bottom: Executed trade prices (left) and mid quotes (right).}
\end{figure}
Figure \ref{plotACF} shows the autocorrelation structure of the high-frequency returns in our data example from Section 4.1. This can be used as a proxy of the autocorrelation function (acf) of the noise increments $(\epsilon_i-\epsilon_{i-1})_{1\le i\le n}$ in the LOMN- and the MMN-model, respectively, since the noise dominates high-frequency increments of the efficient log-price $(X_t)$ over the considered very short time intervals (the sample sizes are given in Table \ref{tab:summary} of Section 4.1). We cannot directly approximate autocorrelations of $(\epsilon_i)_{0\le i\le n}$ by the ones of $(Y_i)_{0\le i\le n}$, however, in the additive noise model. Note that standard significance lines given in acf-plots are not valid due to the neglected increments of the efficient log-price, i.e., the time series $(Y_i)_{0\le i\le n}$ is not the same as $(\epsilon_i)_{0\le i\le n}$. 

Looking at acfs for various days in our sample, we found that the autocorrelation structure can change over time. An acf-plot for the total time series of JPM quotes between July 2007 to September 2009 appears therefore not to be appropriate. For each lag, Figure \ref{plotACF} presents the median over all days in the sample of the intra-daily estimated acfs. The upper panel depicts these medians for best ask quotes (left), which we model with LOMN, and best bid quotes (right), which we model with one-sided, upper-bounded noise, analogously. The lower panel shows the results for executed trade prices (left) and mid quotes (right), which are both often modeled with MMN in the literature. Recall that for i.i.d.\ noise the lag 1 autocorrelation should be significantly negative and autocorrelations for larger lags small. For several single days in the data sample, the acfs of best ask quotes, best bid quotes and executed trade prices show such a pattern. However, on other days the acf decays more slowly and shows significant non-zero values also for larger lags. The medians of executed trade prices in Figure \ref{plotACF} indicate that there are usually two significant autocorrelations at lag 1 and 2, while autocorrelations for larger lags can be neglected. The median autocorrelations of best ask quotes and best bid quotes look rather similar to that of an AR(1) time series. We use such a model in the simulations of Section \ref{sec:sim_global}. The values decay exponentially, but in general there seem to be a few non-negligible autocovariances for small lags. The mid quotes exhibit a slightly different acf pattern, but as well decaying values with a few non-negligible autocovariances for small lags. One conclusion is, that an extension of the theory for the model with LOMN from i.i.d.\ noise to a LOMN-model which allows for serial correlation is important due to these stylized facts of best ask quotes from the limit order book data. Therefore, we develop the theory under Assumption \ref{noise} which allows for serially correlated noise.
\end{document}